\documentclass[twocolumn, pra, superscriptaddress,notitlepage]{revtex4-2}
\usepackage{graphicx}
\usepackage{amsthm}
\usepackage{dcolumn}
\usepackage{bm}
\usepackage[usenames,dvipsnames]{color}
\usepackage[dvipsnames]{xcolor}
\usepackage{enumerate}
\usepackage[shortlabels]{enumitem}
\usepackage[most]{tcolorbox}
\usepackage{multirow}
\usepackage{gensymb}
\usepackage{amssymb}
\usepackage{amsmath}
\usepackage{xcolor}
\usepackage{algorithm}
\usepackage{algorithmic}
\usepackage[normalem]{ulem}
\usepackage{newtxmath}
\usepackage{amsfonts}
\usepackage[colorlinks, linkcolor=blue,anchorcolor=blue,citecolor=blue,urlcolor=blue]{hyperref}
\usepackage{mathtools}
\usepackage{amssymb}
\usepackage{pifont}
\usepackage{physics}
\usepackage{natbib}
\usepackage{xcolor}
\usepackage{placeins}
\newtheorem{theorem}{Theorem}
\newtheorem{proposition}{Proposition}

\newtheorem{corollary}{Corollary}

\begin{document}

\title{
Distributed estimation of many-body Hamiltonians via punctured surface code
}

\author{Linmu Qiao}\email[]{linmuqiao3@gmail.com}
\affiliation{
   Department of Mechanical and Automation Engineering, The Chinese University of Hong Kong, Shatin, Hong Kong}

\author{Zhichun Ouyang}
\affiliation{Department of Physics, Hong Kong University of Science and Technology, Clear Water Bay, Hong Kong, China}

\author{Sisi Zhou}
\email[]{sisi.zhou26@gmail.com}
\affiliation{
Perimeter Institute for Theoretical Physics, Waterloo, Ontario N2L 2Y5, Canada
}
\affiliation{
Department of Physics and Astronomy, Department of Applied Mathematics, and Institute for Quantum Computing, University of Waterloo, Ontario N2L 3G1, Canada
}

\begin{abstract}
We study how a punctured surface code can turn many local $Z$-type couplings into one protected logical signal for distributed quantum metrology, where the goal is to estimate a weighted sum of the coupling strengths. 
We consider an ordinary planar patch with two $X$-cut holes and develop a distributed sensing protocol in which all $Z$-type couplings correspond to the same nontrivial logical $\bar{Z}$ of the punctured surface code. 
When the couplings have prescribed placements on the lattice, we show that the relevant global condition is equivalent to the existence of a simple closed dual loop, called a \textit{witness}, that has an odd number of intersections with every signal path. Together with a local clean-opening condition, this witness criterion gives a concrete punctured-code construction in which all signal paths realize the same nontrivial logical $\bar Z$.
When only the abstract support structure is specified, disjoint supports are straightforward to realize, while overlapping supports are constrained by a stabilizer-difference condition, which we illustrate for three-body interactions.
Overall, our results provide a novel, noise-robust distributed sensing protocol for many-body interactions, with corresponding topological design criteria. 
\end{abstract}

\maketitle

\section{Introduction}

Quantum sensing can gain from entanglement, especially when the goal is not to estimate every local parameter separately, but to estimate one chosen linear functional of them. In distributed quantum sensing~\cite{Zhang:2020skj,Eldredge:2018wna,Proctor:2017bih,Proctor:2017thq,Rubio:2020qyh,zhuang2018distributed,ge2018distributed,bringewatt2021protocols,qian2021optimal,abbasgholinejad2026,hu2025optimalschemedistributedquantum}, a basic task is to estimate
\begin{equation}
    q=\sum_{\alpha} s_{\alpha}\lambda_{\alpha},
\label{eq:target_weight}
\end{equation}
where the weights $s_{\alpha}$ are known and the couplings $\lambda_{\alpha}$ are unknown. This task includes weighted averages and other spatial modes of a field. In ideal models, entangled probes can reach the Heisenberg limit for such a target quantity~\cite{Eldredge:2018wna,Proctor:2017bih,Proctor:2017thq,Rubio:2020qyh}. However, the best known protocols usually rely on long-range coherence across the whole sensor network, and this makes them sensitive to noise and imperfect control. This leads to a natural question: can the same sensing task be realized in a form that is protected by quantum error correction~\cite{NielsenChuang2010,PhysRevA.55.900,Zhou2018HeisenbergQEC,PhysRevLett.112.150802,PhysRevLett.116.230502,PRXQuantum.2.010343,PhysRevLett.122.040502,Layden2018SpatialNoiseQEC,PhysRevLett.133.190801,PhysRevLett.112.150801,Gottesman:1997zz,AntuZhou2025StabilizerHL,Zhang2026Distributed}?

In this work, we study this question for a family of local many-body $Z$-type couplings. The sensing Hamiltonian has the form
\begin{equation}
H=\frac{1}{2}\sum_{\alpha}\lambda_{\alpha}\prod_{j\in S_{\alpha}} Z_j,
\label{eq:intro_hamiltonian}
\end{equation}
where $\lambda_{\alpha}$ is an unknown coupling strength, $S_{\alpha}$ is a known local support, and $Z_j$ is the single-qubit Pauli $Z$ on qubit $j$. Thus each term is a product of $Z$ operators over a known support. Our target is still the weighted sum in Eq.~\eqref{eq:target_weight}. We do not try to estimate all couplings separately. Instead, we ask whether these different physical signal terms can be made to act as one protected logical generator, so that the target quantity $q$ is written onto a single logical qubit.

Our setting is a planar surface code~\cite{KITAEV20032,Dennis2002TopologicalMemory,PhysRevA.86.032324,RAUSSENDORF20062242,Horsman2012LatticeSurgery,Gorecki:2020orm}. Data qubits live on edges of a lattice. $X$-type stabilizers live on vertices, and $Z$-type stabilizers live on faces. Let $E$ be the set of data-qubit edges. In the main text, each signal support $S_\alpha$ corresponds to a connected simple primal path
$c_\alpha \subseteq E $.
That is, $c_\alpha$ runs along adjacent primal edges, has no self-intersection, and forms one connected piece. In the algebraic topology language used in the appendices, $c_\alpha$ is the corresponding $1$-chain in $C_1(G;\mathbb F_2)$. We then write
\begin{equation}
Z(c_{\alpha}) := \prod_{e\in c_{\alpha}} Z_e,
\end{equation}
where $Z_e$ is the single-qubit Pauli $Z$ on the qubit on edge $e$. In this way, $c_{\alpha}$ labels the support, while $Z(c_{\alpha})$ is the corresponding multi-qubit $Z$-string. The Hamiltonian in Eq.~\eqref{eq:intro_hamiltonian} can therefore be written as
\begin{equation}
H=\frac{1}{2}\sum_{\alpha}\lambda_{\alpha} Z(c_{\alpha}).
\end{equation}
Our distributed estimation strategy is to find a punctured surface code such that $Z(c_{\alpha})$ are logical operators and can be sensed jointly.

We now modify the planar patch by puncturing two rough holes, denoted by $R_0$ and $R_1$.
In the usual surface-code language, these are two $X$-cut holes or defects.
Equivalently, they are created by turning off star checks in two simply connected regions and correspondingly modifying the affected $Z$-type plaquette checks.
Below, we also use witness-guided regions $\Omega_0$ and $\Omega_1$, which contain $R_0$ and $R_1$, respectively, and help identify the rough holes.
After these two rough holes are introduced, a primal $Z$ path is allowed to start on $R_0$ and end on $R_1$.
Here and below, this means that the endpoint lies on the rough boundary of $R_i$, the part of $R_i$ adjacent to the active code region.
The corresponding operator $Z(c)$ is then a candidate logical $\bar Z$ on the code space.

Unless stated otherwise, we choose the punctured patch to have code distance
at least three, so that arbitrary single-qubit Pauli errors are correctable.
For the $Z$ sector relevant to the sensing Hamiltonian, the corresponding
distance is the minimum Pauli weight of a nontrivial primal path connecting
$R_0$ and $R_1$.
Thus no Pauli-weight-one or Pauli-weight-two path may realize the logical
$\bar Z$ in this setting.
This restriction applies to nontrivial logical operators, not to $Z$-type
stabilizers, which may have weight two near a truncated rough-boundary corner.
A layout with a Pauli-weight-two logical path falls outside the
arbitrary-single-qubit-error-correcting setting considered here, although it
may still be useful in a located-error model, such as detected qubit loss.

To represent the same logical operator $\bar Z$, the signal supports $\{c_\alpha\}$ need not look the same in real space.
They may have different shapes, different lengths, and different locations on the patch.
What matters is not their local form, but their action on the code space.
This leads to the central question of the paper:
when do all operators $\{Z(c_\alpha)\}$ act as the same nontrivial logical $\bar Z$ on the code space?

We therefore distinguish two complementary design problems.
In the first, the signal paths $c_\alpha$ are already embedded on the lattice, and the task is to choose compatible rough holes.
In the second, only the abstract supports $S_\alpha$ and their overlap
relations are specified, and the task is to find a lattice realization of this abstract structure on the punctured surface-code lattice.
The two problems have the same logical target but different design variables, and we treat them in this order.

Once all relevant signal operators act as the same logical $\bar Z$, the sensing task becomes a logical Ramsey experiment. We prepare a logical superposition state. With suitable control that sets the target weights $s_\alpha$, the unknown couplings build up a single logical phase proportional to the weighted sum $q$. This phase can be accumulated in two closely related ways. In a parallel picture, several compatible signal terms contribute to the logical phase during the same sensing window, with their target weights $s_\alpha$ set by the local control. In a sequential picture, one lets different signal terms contribute one after another, with their target weights set by the chosen time schedule. These two pictures lead to the same target quantity, but they highlight different ways to implement the protocol. Because the phase is stored at the logical level, the scheme also inherits the usual protection of the surface code against errors that remain correctable during the sensing process. The paper then develops this idea in stages.
We start from a simple motivating planar example and the logical picture behind it.
Next, we derive the topological criterion and a constructive rough-hole realization for prescribed, lattice-embedded signal paths, including the local clean-opening condition.
We then turn to the complementary problem of realizing an abstract support structure on the lattice, with overlapping supports governed by a stabilizer-difference condition.
We conclude by discussing the scope of the protocol and possible extensions.

\section{A motivating example}

Figure~\ref{fig:thin_patch_example} fixes the geometric language used in the rest of the paper.
We first explain panel~(a), which shows representative untruncated local stabilizers away from the boundaries and punctures.
The data qubits live on the edges of a planar lattice.
Each face supports a plaquette stabilizer, which is the product of $Z$ on the four edges around that face.
\begin{equation} B_p^Z := \prod_{e \in \partial p} Z_e  \end{equation}
Each vertex supports a star stabilizer, which is the product of $X$ on the four edges that meet at that vertex.
\begin{equation} A_v^X := \prod_{e \ni v} X_e  \end{equation}
The code space $\mathcal C$ is the simultaneous $+1$ eigenspace of all stabilizers that are kept active.
We write $P_{\mathcal C}$ for the projector onto this code space.

We now explain panel (b). The shaded blue region is the rough hole $R_0$, and the shaded orange region is the rough hole $R_1$. This is the geometric picture used throughout the main text.
In the stabilizer description, each rough hole is associated with a connected region in which the $X$-type star checks are turned off.
Forming the full puncture also requires corresponding changes to the $Z$-type plaquette checks.
Inside the punctured region, both the $X$-type star checks and the $Z$-type plaquette checks are removed from the stabilizer group.
Along the boundary, the affected $Z$-type plaquette checks are replaced by truncated plaquette checks supported only on the data qubits that remain outside the hole.
Together, these stabilizer modifications remove the punctured region from the active code and define the resulting rough boundary. In the modified code, a primal $Z$ path can terminate on this boundary without creating a stabilizer violation at its endpoint.

The two rough holes change which $Z$-strings can act as logical operators.
After they are introduced, a primal $Z$ path is allowed to start on $R_0$ and end on $R_1$.
Such a path is a candidate support for a logical $\bar Z$ operator.
This is why the three colored primal paths $c_1$, $c_\star$, and $c_2$ are drawn as curves connecting the two rough holes.

\begin{figure}[tb]
    \centering
    \includegraphics[width=\linewidth]{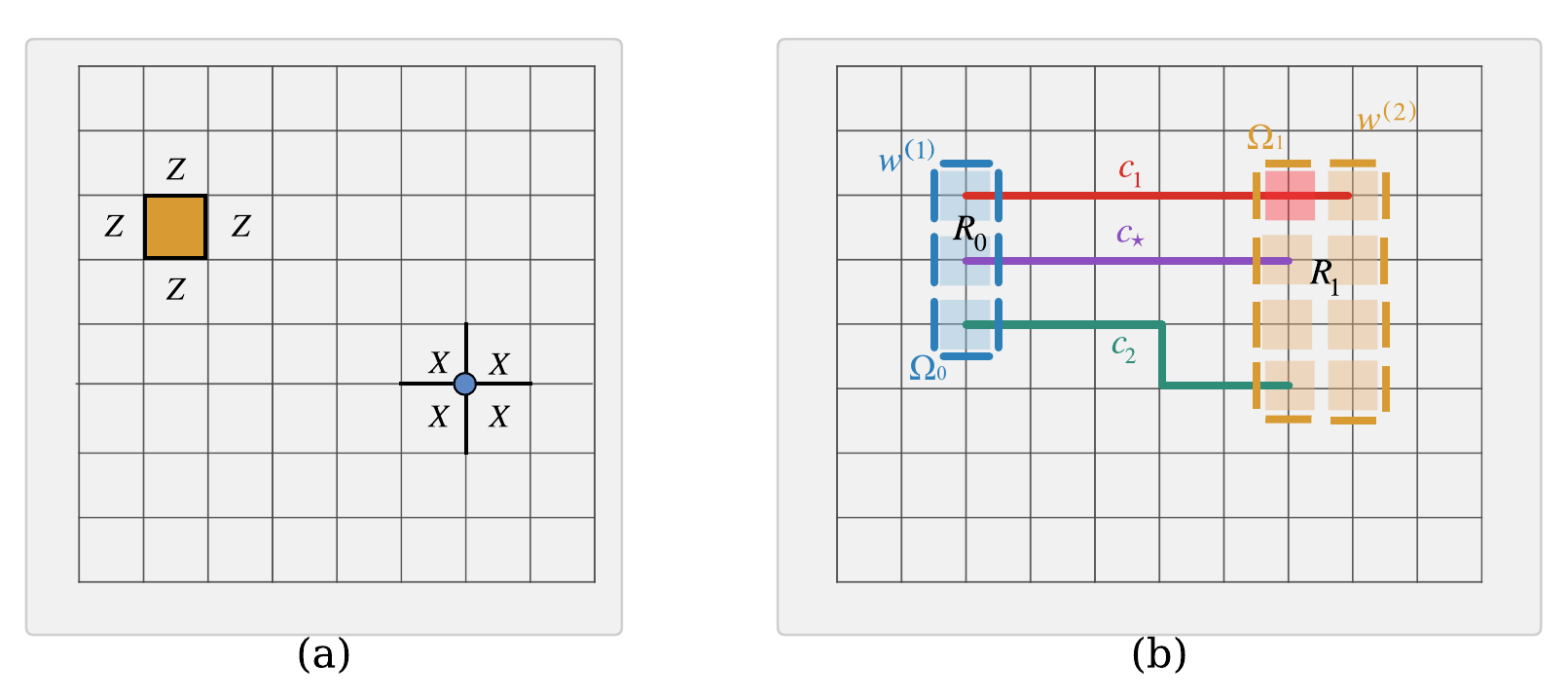}
    \caption{
    A planar-code dictionary and the simple example used throughout this work.
    Panel (a) shows representative untruncated local stabilizers.
    Data qubits live on edges.
    The displayed plaquette stabilizer is the product of $Z$ operators on the four edges surrounding a face, while the displayed star stabilizer is the product of $X$ operators on the four edges incident on a vertex.
    Panel (b) shows a thin planar patch with two rough holes $R_0$ and $R_1$.
    The dual loops $w^{(1)}$ and $w^{(2)}$ play the role of witnesses.
    They enclose the witness-guided regions $\Omega_0$ and $\Omega_1$.
    In this example, $\Omega_0=R_0$, while $R_1\subsetneq\Omega_1$.
    The primal paths $c_1$ and $c_2$ are the signal supports of the physical operators $Z(c_1)$ and $Z(c_2)$, while $c_\star$ is a reference path in their common logical class.
    The $X$-products supported on the witnesses have the geometry of logical $\bar X$ strings.
    In this example, each compatible signal path crosses the relevant witness with odd parity.
    This is the geometric pattern behind the common logical $\bar Z$ action established later.
    }
    \label{fig:thin_patch_example}
\end{figure}

At this point, it is useful to separate \emph{support} from \emph{operator}.
In the lattice-embedded setting considered here, a signal support $c$ is a simple primal path of data-qubit edges. It is a geometric object first. The corresponding operator is
\begin{equation}
Z(c) := \prod_{e \in c} Z_e .
\end{equation}
Thus the red, purple, and green curves in panel~(b) are first signal supports, and then, after we attach $Z$ operators to their edges, they become the physical operators $Z(c_1)$, $Z(c_\star)$, and $Z(c_2)$. The simple-path assumption ensures compatibility with the active star checks.
At each non-terminal vertex of a simple primal path, $Z(c)$ overlaps the corresponding star check $A_v^X$ on two incident edges and therefore commutes with it.
The only vertices with odd incidence are the two endpoints, which must lie on rough boundaries where the corresponding star checks are turned off.

Next, we introduce the concept of the witness, which is closely related to the rough holes that define the punctured surface code. 
First, a witness is a geometric object.
It is drawn as a simple closed loop on the \emph{dual} lattice, with no self-intersections.
Here ``dual'' means that the loop runs along dual edges, which cross primal edges rather than lying on them.
Second, the same dual loop defines an $X$-type string operator.
Let $E(w)\subseteq E$ be the set of primal edges crossed by $w$.
We write
\begin{equation}
X(w) := \prod_{e \in E(w)} X_e .
\end{equation}
In other words, the product is over all primal edges crossed by $w$.
Thus a witness is not only a geometric test curve.
This edge set $E(w)$ is itself the support of an $X$-type string operator.
For this reason, it is drawn in the same style as a logical $\bar X$ string.

This dual point of view is exactly what we need later.
Geometrically, the witness tells us how a primal $Z$ path crosses the patch.
The same witness also has the geometry of a logical $\bar X$ string.
In the appendices, the crossing data of the same simple witness are encoded algebraically as a $1$-cochain.
But at the level of this example, it is enough to think of it as a dual loop that crosses primal edges and supports an $X$ product.

A witness also selects a witness-guided region $\Omega(w)$ by choosing one side of the loop.
When such a witness is used to guide the $i$-th rough hole, as in Fig.~\ref{fig:thin_patch_example}(b), we write $ \Omega_{i} $ for the resulting region $ \Omega(w) $. 
This region serves only as a geometric guide.
The actual rough hole is later chosen as a clean connected subregion $R\subseteq\Omega(w)$.
It is chosen so that the assigned signal endpoints lie on the rough boundary, while the non-terminal parts of all prescribed signal paths remain in the active code region.

We can now read panel (b) in a simple way.
The three primal paths $c_1$, $c_\star$, and $c_2$ have different microscopic shapes.
The red path $c_1$ takes an upper route.
The purple path $c_\star$ is the direct reference path.
The green path $c_2$ takes a lower route before entering the right rough hole.
So the figure already shows that the logical action does not depend on one special path in real space.

The two dual loops $w^{(1)}$ and $w^{(2)}$ are the witnesses for the left and right rough holes.
Each compatible primal path crosses each relevant witness an odd number of times.
This odd crossing is the key visual fact.
It means that, although the three paths look different microscopically, they play the same topological role in the patch.
After the rough holes are fixed, they act as the same logical $\bar Z$ on the code space, up to stabilizers.

The witness itself is not unique.
It may be deformed on the dual lattice as long as its parity of crossings with the signal paths does not change.
This is the topological content that the later theorem makes precise.

Whenever a signal support has been embedded on the lattice, we represent it by a simple primal path $c_\alpha$.
This geometric model is sufficient for the sensing protocol and for Theorem~\ref{thm:global_odd_witness}.
In the appendices, the same embedded path is written as a $1$-chain whose support is that path.
At this stage, however, the figure is meant only to build intuition.
It tells us what kind of logical compression we want.
In the next section, we explain the distributed sensing protocol that uses such a compatible family of logical $\bar Z$ strings.

\section{Distributed sensing protocols}

Before stating the topological criterion, we explain the sensing protocol that it supports.
This is the surface-code version of weighted-sum sensing in a quantum sensor network.

We now briefly explain why this surface-code construction corresponds to the same sensing task as the standard optimal protocol for a chosen linear functional in distributed quantum sensing~\cite{Eldredge:2018wna}. The essential difference is the encoding. In the usual network picture, the phase is stored in a nonlocal probe state, such as a GHZ state. Here, the same phase is stored in a single logical qubit of the surface code. Moreover, as long as syndrome extraction and decoding are accurate and fast enough during sensing, correctable errors do not erase the accumulated logical phase. This includes both low-weight incoherent Pauli errors that are correctable by the punctured surface code and additional coherent error terms in the Hamiltonian that are not nontrivial logical operators.

We first fix the normalization of the target coefficients.
As in the standard network setting, we write the target linear functional as
\begin{equation}
q=\sum_{\alpha} s_\alpha \lambda_\alpha,
\qquad
|s_\alpha|\le 1 \ \text{for all }\alpha,
\qquad
\max_\alpha |s_\alpha|=1.
\end{equation}
This normalization simply fixes the scale.
If the desired coefficients are $\tilde s_\alpha$, we define
\begin{equation}
s_*:=\max_\alpha |\tilde s_\alpha|,
\qquad
s_\alpha:=\frac{\tilde s_\alpha}{s_*},
\qquad
q:=\sum_\alpha s_\alpha \lambda_\alpha.
\end{equation}
Then the original target is recovered by a known rescaling,
\begin{equation}
\tilde q=\sum_\alpha \tilde s_\alpha \lambda_\alpha = s_*\, q.
\end{equation}
With this convention, the effective logical time $t$ is the
time entering the logical Ramsey phase. It is compared against the
normalized functional $q$, and the final precision statement takes the same
form as in the sensor-network Heisenberg limit. When a concrete implementation is specified, we distinguish this effective
logical time from the actual dwell time of the schedule.
We denote the latter by $T_{\rm par}$ for parallel loading and by
$T_{\rm seq}$ for sequential hole motion. These times only count intervals
during which signal phase is accumulated. They do not include additional
overheads from control operations, hole motion, or stabilizer updates.

We assume that, for a fixed configuration, the planar patch carries two rough holes $R_0$ and $R_1$, so that one logical qubit is encoded.
We write its logical Pauli operators as $\bar X$ and $\bar Z$, and we denote the code space by $\mathcal C$, with projector $P_{\mathcal C}$.
For sequential hole motion, the same notation is used within each dwell interval, while the intervening stabilizer updates transport the encoded logical qubit between configurations.

Within each such configuration, we also assume that we are given a family of legal $Z$-type paths $\{c_\alpha\}$ such that each path $c_\alpha$ is a simple primal path from $R_0$ to $R_1$, and each operator $Z(c_\alpha)$ acts as the same logical operator $\bar Z$ on $\mathcal C$.
The signal paths may have different shapes, but they all write phase onto the same logical degree of freedom.
Theorem~\ref{thm:global_odd_witness} and Proposition~\ref{prop:clean_puncture_realization} in the next section tell us when such a compatible family exists.
Here we first explain what the protocol does once this compatibility is available.

\begin{proposition}[Surface-code weighted-sum sensing]
\label{pro:Surface-code_weighted-sum_sensing}
Assume that every active physical signal operator $Z(c_\alpha)$ acts as the same logical operator $\bar Z$ on the code space $\mathcal C$.
Then the normalized target
\begin{equation}
q=\sum_\alpha s_\alpha \lambda_\alpha
\end{equation}
can be estimated by a logical Ramsey experiment.
Here $t$ is the effective logical time, while
$T$ is the dwell time used by the chosen schedule.
On the code space, this schedule implements
\begin{equation}
U_{\rm log}(T)
=
\exp\!\left(
-\frac{i t}{2}
\sum_\alpha s_\alpha\lambda_\alpha \bar Z
\right).
\label{eq:effective_evolution}
\end{equation}
The coefficients $s_\alpha$ are set by control.
They can be programmed either by parallel loading with time-averaged echo control,
or by sequential hole motion with chosen dwell times.
\end{proposition}

\begin{figure}[tb]
\centering
\includegraphics[width=\linewidth]{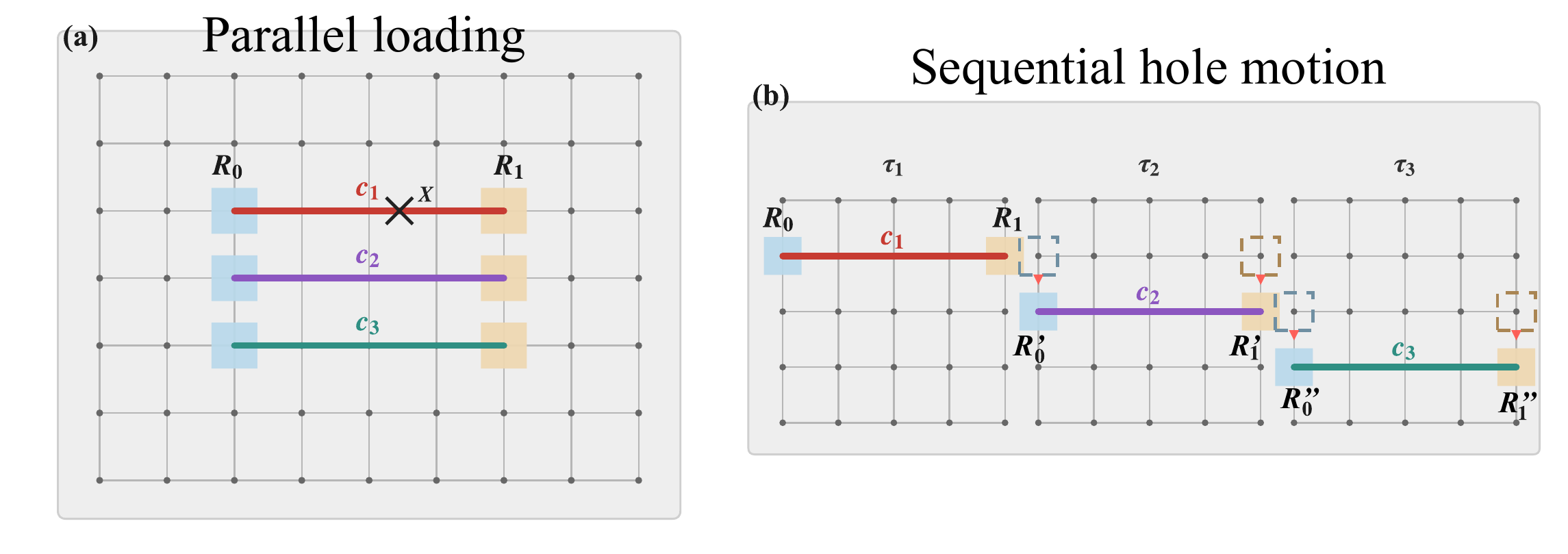}
\caption{
(a) \emph{Parallel loading}.
The rough holes $R_0$ and $R_1$ are fixed.
The signal operators supported on the legal paths $c_1,c_2,c_3$ act jointly on the same logical qubit.
The mark $X$ denotes a local $X$ flip on an edge qubit along the support.
(b) \emph{Sequential hole motion}.
The sensing process is divided into three intervals, $\tau_1,\tau_2,\tau_3$.
During dwell time $\tau_1$, the active rough holes are $R_i$, with $i=0,1$.
At the end of $\tau_1$, local measurements change the active holes from $R_i$ to $R_i'$.
During dwell time $\tau_2$, the active rough holes are $R_i'$.
At the end of $\tau_2$, local measurements change the active holes from $R_i'$ to $R_i''$.
During dwell time $\tau_3$, the active rough holes are $R_i''$.
Thus the rough hole pair is moved through positions $R_i \to R_i' \to R_i''$ over the three intervals.
During each interval, only the intended signal operator contributes.
Between intervals, local measurements are performed to move the holes and update the stabilizers.
These updates preserve the encoded logical qubit.
Therefore, they do not erase the logical phase accumulated in the previous intervals.
In both panels, the total phase is written onto the same logical operator $\bar Z$.
The protocol ends with one logical $\bar X$ measurement.
}
\label{fig:protocol}
\end{figure}

The protocol realizing this proposition has four steps.

\textit{1. Input state.}
We prepare the logical $+1$ eigenstate of $\bar X$,
\begin{equation}
|+_L\rangle
=
\frac{1}{\sqrt{2}}\bigl(|0_L\rangle+|1_L\rangle\bigr).
\end{equation}
This is the logical version of the usual Ramsey input state.
It is most sensitive to a phase generated by $\bar Z$.

\textit{2. Sensing evolution.}
During the sensing window, the unknown couplings act through physical operators $Z(c_\alpha)$.
Because all active signal operators represent the same logical operator $\bar Z$, the full physical evolution reduces to one logical $Z$ rotation.

Let $\eta_\alpha(t')$ denote the signed activation profile of the $\alpha$-th signal during $0\le t'\le T$.
We use a bang-bang profile,
\begin{equation}
\eta_\alpha(t')\in\{-1,0,+1\}.
\end{equation}
Here $\eta_\alpha(t')=+1$ means that the signal is active with its original sign.
The value $\eta_\alpha(t')=-1$ means that its sign is reversed by an echo operation.
The value $\eta_\alpha(t')=0$ means that the signal is not active during that part of the schedule.
General real weights are obtained by changing the duty cycle of this profile.

In the common logical frame, the effective logical Hamiltonian is
\begin{equation}
H_{\rm log}(t')
=
\frac{1}{2}
\left(
\sum_\alpha \eta_\alpha(t')\lambda_\alpha
\right)
\bar Z .
\end{equation}
Since this Hamiltonian is always proportional to the same logical operator $\bar Z$, the phases add directly.
Thus the logical evolution generated during the sensing part of the schedule is
\begin{equation}
U_{\rm log}(T)
=
\exp\!\left[
-\frac{i}{2}
\left(
\sum_\alpha \lambda_\alpha
\int_0^T \eta_\alpha(t')\,dt'
\right)
\bar Z
\right].
\end{equation}
The schedule realizes the effective evolution in Eq.~\eqref{eq:effective_evolution} if
\begin{equation}
\int_0^T \eta_\alpha(t')\,dt'
=
t s_\alpha
\label{eq:effective_evolution_condition}
\end{equation}
for every $\alpha$.
Here $t$ is the effective logical time in Eq.~\eqref{eq:effective_evolution}, while $T$ is the actual dwell time used for sensing.
When this condition holds,
\begin{equation}
U_{\rm log}(T)
=
\exp\!\left(
-\frac{i t}{2}
\sum_\alpha s_\alpha\lambda_\alpha \bar Z
\right).
\end{equation}

There are two natural ways to realize these coefficients, illustrated in Fig.~\ref{fig:protocol}.

\emph{Parallel loading.}
We keep the rough holes $R_0$ and $R_1$ fixed.
We assume each legal path $c_\alpha$ has at least one path-specific
data-qubit edge that is not shared by other paths, so that the $X$ gate control described below applies. 
Several signal paths may contribute during the same sensing window.
A convenient model is
\begin{equation}
H_{\mathrm{par}}(t')
=
\frac{1}{2}\sum_\alpha \eta_\alpha(t')\,\lambda_\alpha\, Z(c_\alpha),
\qquad
\eta_\alpha(t')\in\{\pm 1\}.
\end{equation}

To realize the effective evolution in Eq.~\eqref{eq:effective_evolution} with effective logical time $t$, the control must satisfy Eq.~\eqref{eq:effective_evolution_condition}
for every $\alpha$.
Since $|\eta_\alpha(t')|\le 1$, this requires
\begin{equation}
T_{\rm par}\ge t\max_\alpha |s_\alpha|.
\end{equation}
With our normalization $\max_\alpha |s_\alpha|=1$, the minimum parallel dwell time is
$ T_{\rm par}^{\min}=t $.

A negative sign is obtained by applying an odd number of local $X$ gates to
edge qubits on the support of $c_\alpha$.
This echo operation flips
\begin{equation}
Z(c_\alpha)\longmapsto -\,Z(c_\alpha)
\end{equation}
during that time segment.
These echo controls are treated as known fast control pulses. Any stabilizer
sign changes caused by these pulses are tracked in the Pauli frame, so they do not act as unknown physical Pauli errors.
Thus the control produces a signed activation profile for each signal.
A coefficient with $0<|s_\alpha|<1$ is obtained from the time average of this signed activation profile.
For example, when $T_{\rm par}=t$, one may take $\eta_\alpha(t')=+1$ for a fraction $(1+s_\alpha)/2$ of the time and $\eta_\alpha(t')=-1$ for a fraction $(1-s_\alpha)/2$ of the time.
This gives $\int_0^{T_{\rm par}} \eta_\alpha(t')\,dt'=t s_\alpha$.
If the same unknown coupling $\lambda_\alpha$ is available through several signal paths, multiplicity can also be used
(Appendix~\ref{app:logical-ramsey-protocol}).
In particular, multiplicity alone is enough for integer and rational weights, while partial-time activation gives a direct route to general normalized real weights.

\emph{Sequential hole motion.}
Instead of keeping all signal paths active on one fixed patch, we may move rough holes through a sequence of nearby positions.
For the simple schedule below, we require that at each dwell step, the chosen rough-hole positions make only the intended signal path legal.
Paths with the same endpoints are treated as one parallel group with a common
weight, unless additional local echo controls distinguish them. 
At step $r$, one legal path $c_r$ is active for a dwell time $\tau_r$.
If that step senses the coupling $\lambda_{\alpha_r}$, then its contribution is
\begin{equation}
H_r
=
\frac{1}{2}\,\sigma_r\,\lambda_{\alpha_r}\, Z(c_r),
\qquad
\sigma_r\in\{\pm1\},
\end{equation}
where $\sigma_r$ records whether an echo is used to reverse the sign in that step.

Let
\begin{equation}
T_{\rm seq}
=
\sum_r \tau_r
\end{equation}
be the actual dwell time of the sequential sensing schedule.
To realize the effective evolution in Eq.~\eqref{eq:effective_evolution} with effective logical time $t$, the dwell times must satisfy
\begin{equation}
\sum_{r:\,\alpha_r=\alpha}\sigma_r \tau_r
=
t s_\alpha
\end{equation}
for every $\alpha$.
So in the sequential picture, the weights are set by dwell times. The hole motion selects which physical path is active at each moment, but the phase is always written onto the same logical operator $\bar Z$.

For a fixed $\alpha$, the constraint above and $\tau_r\ge 0$ imply
\begin{equation}
t|s_\alpha|
=
\left|
\sum_{r:\,\alpha_r=\alpha}\sigma_r\tau_r
\right|
\le
\sum_{r:\,\alpha_r=\alpha}\tau_r .
\end{equation}
Summing over $\alpha$ therefore gives
\begin{equation}
T_{\rm seq}
\ge
t\sum_\alpha |s_\alpha|.
\end{equation}
This bound is achieved by assigning one dwell interval to each
$\alpha$ with $s_\alpha\neq 0$, with
$\tau_\alpha=t|s_\alpha|$ and
$\sigma_\alpha=\operatorname{sign}(s_\alpha)$.
Hence,
\begin{equation}
T_{\rm seq}^{\min}
=
t\sum_\alpha |s_\alpha|.
\end{equation}

Under the normalization $\max_\alpha |s_\alpha|=1$, sequential hole motion
is therefore generally slower than parallel loading by a factor
$\sum_\alpha |s_\alpha|$.
This factor can nevertheless be close to one when the target is dominated
by a single coefficient.
For example, if $|s_{\alpha_0}|=1$ and
$\sum_{\alpha\ne\alpha_0}|s_\alpha|\ll 1$, then
$T_{\rm seq}^{\min}\approx t=T_{\rm par}^{\min}$.

Between two dwell periods, we may update the code by local Pauli measurements that move the holes, as described in Appendix~\ref{app:stabilizer-updates-logical-invariance}. These intermediate code projections change the active stabilizers, but they preserve the encoded logical qubit up to a known Pauli-frame update. Therefore, they do not erase the logical phase already accumulated in the earlier steps.
At each step, the active path still represents the same logical operator $\bar Z$ on the instantaneous code space. The phases from all steps therefore add coherently, and the protocol ends with one logical $\bar X$ measurement on the final code patch.

We now give a simple example.
Suppose the desired quantity is the average of $M$ local couplings,
\begin{equation}
\bar\lambda
:=
\frac{1}{M}\sum_{\alpha=1}^{M}\lambda_\alpha.
\end{equation}
Under our normalization convention, we first sense the normalized sum
\begin{equation}
q_{\rm sum}
:=
\sum_{\alpha=1}^{M}\lambda_\alpha,
\qquad
s_\alpha=1
\quad
\text{for all }\alpha.
\end{equation}
The average is then obtained by the known rescaling
\begin{equation}
\bar\lambda
=
\frac{q_{\rm sum}}{M}.
\end{equation}

In the parallel picture, all $M$ paths can remain active during the same sensing window.
Thus $q_{\rm sum}$ is sensed for effective time $t$ using actual dwell time $T_{\rm par}=t$.
In the sequential picture, the paths are visited one by one.
Since $s_\alpha=1$ for every $\alpha$, each path must dwell for time $t$.
Thus
$T_{\rm seq}=Mt$.
Therefore, in this equal-weight example, sequential hole motion realizes the same effective logical evolution as parallel loading, but it takes $M$ times longer.
If the actual dwell time is fixed, the sequential protocol has a smaller effective logical time by the same factor.

\textit{3. Final measurement.}
After the schedule realizes the effective logical time $t$, we measure the logical operator $\bar X$.
Starting from $|+_L\rangle$, the two outcomes occur with probabilities
\begin{equation}
p_\pm
=
\frac{1}{2}\bigl[1\pm \cos(qt)\bigr].
\end{equation}
Thus the unknown weighted sum $q$ is read out from one logical Ramsey fringe.
This is the same weighted-sum sensing task as in distributed quantum sensing, but now it is encoded in one protected logical qubit.

\textit{4. Precision and the average case.}
As usual for Ramsey metrology, the value of Fisher information reveals the local sensitivity for $q$-estimation, e.g., when $q$ is known in advance to be within some small neighborhood of its true value. 
The binary distribution above has classical Fisher information $F=t^2$.
Hence, after $\nu$ repetitions, any unbiased estimator $Q$ of the normalized target $q$ satisfies
\begin{equation}
\mathrm{Var}(Q)\ge \frac{1}{\nu t^2}.
\end{equation}
This is the sensor-network Heisenberg limit for one normalized linear functional.

For the average task, however, it is cleaner to apply this bound to the normalized sensed variable
\begin{equation}
q_{\mathrm{sum}}
:=
\sum_{\alpha=1}^{M}\lambda_\alpha,
\qquad
s_\alpha=1
\quad
\text{for all }\alpha.
\end{equation}

If $Q_{\mathrm{sum}}$ is an unbiased estimator of $q_{\mathrm{sum}}$, then
\begin{equation}
\mathrm{Var}(Q_{\mathrm{sum}})\ge \frac{1}{\nu t^2}.
\end{equation}
We then estimate the mean by the known rescaling
\begin{equation}
\bar\Lambda
:=
\frac{Q_{\mathrm{sum}}}{M}.
\end{equation}
Therefore,
\begin{equation}
\mathrm{Var}(\bar\Lambda)
=
\frac{\mathrm{Var}(Q_{\mathrm{sum}})}{M^2}
\ge
\frac{1}{\nu M^2 t^2}.
\end{equation}
Equivalently,
\begin{equation}
\Delta\bar\Lambda
\ge
\frac{1}{\sqrt{\nu}\,Mt}.
\end{equation}
Thus, as a function of the effective logical time $t$, the average case recovers the usual Heisenberg scaling in the number of sensors.
For the parallel implementation, this is also the scaling in the actual sensing time because $T_{\rm par}=t$.
For the sequential implementation, the same effective time requires $T_{\rm seq}=Mt$.

The key difference from the usual GHZ picture is the physical form of the protocol.
Here the accumulated phase is stored at the logical level.
Therefore, as long as standard surface-code error correction can be carried out sufficiently well during the sensing process, all correctable local errors can be removed without erasing the target phase.
The remaining question is geometric:
when can a prescribed family of lattice-embedded signal paths $\{c_\alpha\}$ be realized as one common logical $\bar Z$?
Theorem~\ref{thm:global_odd_witness} and Proposition~\ref{prop:clean_puncture_realization} together answer exactly that question.

\section{Design rules for lattice-embedded signal supports}

We now turn to the geometric question behind Proposition~\ref{pro:Surface-code_weighted-sum_sensing}.
At this stage, the signal Hamiltonian has already been specified on the surface-code lattice. Each signal support is represented by a simple primal path $c_\alpha$ of data-qubit edges, and the corresponding physical signal operator is $Z(c_\alpha)$.
The question is the following:
when do all operators $Z(c_\alpha)$ act as the same nontrivial logical $\bar Z$ on the code space?

The desired punctured patch will contain two disjoint, simply connected (disk-like) rough holes $R_0$ and $R_1$.
Each hole is specified by a simply connected region in which the corresponding $X$-type star checks are turned off, together with the required removal or truncation of the affected $Z$-type plaquette checks.
These actual rough holes are distinct from the auxiliary witness-guided regions $\Omega_0$ and $\Omega_1$ introduced below.
A prescribed signal path $c_\alpha$ is legal as a $Z$-type path when one endpoint lies on the rough boundary of $R_0$, the other lies on the rough boundary of $R_1$, and its non-terminal part remains in the active code region.
The corresponding operator $Z(c_\alpha)$ is then a candidate logical $\bar Z$ operator.

The key object is again the witness.
As in Fig.~\ref{fig:thin_patch_example}, a witness is a simple closed loop on the dual lattice.
To compare it with a primal path $c$, we look edge by edge.
For each primal edge $e$, let $c_e,w_e\in\{0,1\}$ record, respectively,
whether the primal path $c$ uses $e$ and whether the dual loop $w$ crosses $e$.
We then define
\begin{equation}
\langle w,c\rangle
:=
\sum_e w_e c_e
\pmod 2.
\label{eq:F2_pairing}
\end{equation}
So $\langle w,c\rangle$ is just the parity of the crossings between $w$ and $c$.
If they cross once, or three times, then $\langle w,c\rangle=1$.
If they cross zero, two, or four times, then $\langle w,c\rangle=0$.

There are two separate questions.

First, there is a global topological question.
It asks whether the signal paths admit one common odd witness.
This question does not yet involve the actual rough holes.
It is answered by the global odd-witness criterion below.

Second, there is a local geometric realization question.
It asks whether a common odd witness can be converted into two actual
clean rough holes.
This question does involve the removed star checks.
It depends on whether the witness-guided regions can be opened without
turning off star checks at non-terminal vertices of the prescribed signal paths.
It is answered by the clean puncture realization statement after the theorem.

\subsection{Odd-witness criterion}

We first answer the global question.
For each signal path $c_\alpha$, let its endpoints (vertices) be $u_\alpha$ and
$v_\alpha$. Define the endpoint graph $\Gamma_{\rm end}$ as follows.
The vertices of $\Gamma_{\rm end}$ are all endpoints that appear among the
signal paths.
For each $c_\alpha$, we add one abstract edge $\eta_\alpha$ between
$u_\alpha$ and $v_\alpha$.
The edge $\eta_\alpha$ records only the two endpoints of $c_\alpha$.
It does not record the geometric shape of $c_\alpha$.

\begin{theorem}[Global odd-witness criterion]
\label{thm:global_odd_witness}
Let $\{c_\alpha\}_{\alpha\in\mathcal I}$ be simple primal paths on one
planar patch. Then the following
statements are equivalent.

\noindent(a) The endpoint graph $\Gamma_{\rm end}$ is bipartite.
Equivalently, every cycle in $\Gamma_{\rm end}$ has even length.

\noindent(b) There is no subset $\mathcal A\subseteq\mathcal I$ such that
$|\mathcal A|$ is odd and
\begin{equation}
D\left(\bigoplus_{\alpha\in\mathcal A}c_\alpha\right)=0 .
\end{equation}

\noindent(c) There exists a simple closed dual witness $w$ such that
\begin{equation}
\langle w,c_\alpha\rangle=1
\qquad
\text{for all }\alpha\in\mathcal I .
\label{eq:global_odd_witness}
\end{equation}
\end{theorem}

Condition (a) says that the endpoint vertices of the signal paths can be colored with two colors in such a way that every signal path connects vertices of opposite colors. Equivalently, $\Gamma_{\rm end}$ contains no odd cycle.
In condition (b), $D$ is an endpoint-parity map. Given a collection of paths, $D$ lists all their endpoint vertices and cancels identical vertices in pairs. Thus $D$ returns precisely those endpoint vertices that occur with odd multiplicity. Therefore, the equation in condition (b) says that the paths indexed by $\mathcal A$ have no unpaired endpoints; equivalently, every endpoint vertex appears an even number of times. In terms of $\Gamma_{\rm end}$, this means that each endpoint vertex is incident to an even number of the selected paths. Condition (b) rules out the possibility that such complete endpoint cancellation occurs for an odd number of selected paths. This is exactly the obstruction given by an odd cycle, and hence is equivalent to the bipartiteness of $\Gamma_{\rm end}$.
Condition (c) says that every signal path crosses the same simple closed dual loop with odd parity.

For example, the paths do not share endpoint vertices in Fig.~\ref{fig:witness_to_holes}, and the endpoint graph is a disjoint union of edges and is bipartite by default, satisfying condition (a); an odd witness can then be found by Theorem~\ref{thm:global_odd_witness}. In Fig.~\ref{fig:caveat}, we provide an example where the paths share endpoint vertices. Note that even though the paths are edge-disjoint, their endpoint
vertices may still coincide. In this case, the
existence of an odd witness depends on whether the endpoint graph contains
an odd cycle. Even cycles, as in Fig.~\ref{fig:caveat}, are allowed, whereas odd cycles are the
obstruction.
Fig.~\ref{fig:caveat} provides an example in which the global odd-witness criterion of Theorem~\ref{thm:global_odd_witness} is satisfied, whereas the local clean-opening condition required by Proposition~\ref{prop:clean_puncture_realization} fails.
We will discuss this situation later in detail. 

Theorem~\ref{thm:global_odd_witness} therefore completes the global
odd-witness test for the prescribed signal paths.

\subsection{Cleanly openable regions}

We now turn to the second, local question.
Choose a simple odd-witness representative $w^{(1)}$.
Because $w^{(1)}$ crosses every signal path $c_\alpha$ with odd parity, the two
endpoints of $c_\alpha$ lie on opposite sides of $w^{(1)}$.
Choose one side as the first witness-guided region $\Omega_0$.
The endpoints on this side are the endpoints that will terminate on the first
rough hole $R_0$---a region within $\Omega_0$ to be found later.
The other endpoints will terminate on the second rough hole $R_1$.

The regions $\Omega_0$ and $\Omega_1$ are only witness-guided regions.
The actual rough holes are the clean subregions
$R_i\subseteq\Omega_i$ chosen later.

A prescribed signal path may pass through a witness-guided region $\Omega_i$.
This does not by itself cause a problem.
When the actual rough hole $R_i$ is chosen, its assigned endpoint vertices may
lie on the rough boundary, but every non-terminal vertex of every prescribed
path must remain outside $R_i$.
This keeps the corresponding star stabilizers active and ensures that $Z(c_\alpha)$ commutes with them.
For a simple primal path, a non-terminal vertex is incident on two edges of the
path, whereas an endpoint vertex is incident on only one.

We say that the witness-guided regions are \emph{cleanly openable} if one can
choose two disjoint, simply connected actual rough holes
$R_0\subseteq \Omega_0$ and $R_1\subseteq \Omega_1$
satisfying the following two conditions.
First, each prescribed signal path $c_\alpha$ has one endpoint on the rough
boundary of $R_0$ and the other endpoint on the rough boundary of $R_1$.
Second, the non-terminal part of every prescribed signal path $c_\alpha$
remains in the active code region (outside $R_0\cup R_1$).
Equivalently, when the actual rough holes are chosen, no star check at a
non-terminal vertex of any $c_\alpha$ is turned off.

\begin{proposition}[Clean puncture realization]
\label{prop:clean_puncture_realization}
Assume that the global odd-witness condition in
Theorem~\ref{thm:global_odd_witness} holds and that the resulting 
witness-guided regions $\Omega_{0}$ and $\Omega_{1}$ are cleanly openable.
Let $R_0\subseteq\Omega_0$ and $R_1\subseteq\Omega_1$ be the resulting two disjoint actual rough holes.

Then each signal path $c_\alpha$ has one endpoint on the rough boundary of
$R_0$ and the other endpoint on the rough boundary of $R_1$, while its
non-terminal part remains in the active code region.

Let $c_\star$ be any legal reference path with one endpoint on the rough boundary
of $R_0$ and the other endpoint on the rough boundary of $R_1$, whose
non-terminal part also remains in the active code region.
For fixed clean rough holes in the two-hole planar setting considered here,
the $Z$-string operators supported on any two such paths differ by a product of active $Z$-type plaquette stabilizers.
Hence all prescribed signal operators have the same code-space action as $Z(c_\star)$:
\begin{equation}
P_{\mathcal C} Z(c_\alpha) P_{\mathcal C}
=
P_{\mathcal C} Z(c_\star) P_{\mathcal C}
=
P_{\mathcal C} \bar Z P_{\mathcal C}
\qquad
\text{for all }\alpha .
\end{equation}
Thus all prescribed signal operators $Z(c_\alpha)$ realize the same
nontrivial logical operator $\bar Z$, represented for example by
$Z(c_\star)$.

Conversely, suppose two clean actual rough holes $R_0$ and $R_1$ are already
fixed in the two-hole planar setting considered here.
If every prescribed signal path $c_\alpha$ has one endpoint on the rough
boundary of $R_0$ and the other endpoint on the rough boundary of $R_1$, with
its non-terminal part in the active code region, then there is a simple closed
dual witness $w$ separating $R_0$ from $R_1$.
Every such signal path crosses this separator with odd parity.
Hence
\begin{equation}
\langle w,c_\alpha\rangle=1
\qquad
\text{for all }\alpha .
\end{equation}
Thus the global odd-witness condition is necessary for any fixed clean
two-hole realization.
\end{proposition}

The forward direction says that the global odd-witness condition (Theorem~\ref{thm:global_odd_witness}) becomes
sufficient for the existence of the required punctured surface code, after imposing the local clean-openability condition. 
The witness gives the global endpoint assignment, while the clean condition
checks whether this assignment can be realized by actual rough
holes.
The converse direction is a necessity statement for fixed clean rough holes.
Once $R_0$ and $R_1$ are fixed, any prescribed signal path connecting their rough boundaries must cross a dual separator between the two holes with odd parity.

The witness loop and the support of a logical $Z$ operator obey different boundary rules.
The witness is closed in the dual sense: it has no endpoints on the dual lattice.
By contrast, after the rough holes are introduced, a logical $Z$ operator may
be supported on a primal path with one endpoint on the rough boundary of $R_0$ and the other on the rough boundary of $R_1$.
In relative-homology language, this path is a relative cycle rather than an ordinary closed primal cycle.
Its endpoints are allowed because the corresponding star stabilizers are inactive at the rough boundaries.

For readers who want the algebraic formulation, Appendix~\ref{app:chain-cochain-notation-odd-witness} writes the closedness of the witness as
\begin{equation}
B^\top w = 0,
\end{equation}
where $B$ is the face--edge incidence matrix over $\mathbb F_2$.
This formula only says that the dual support of $w$ has no endpoints:
around each plaquette, the loop enters and leaves in pairs.
Likewise, saying that two $Z$-strings differ only by plaquette stabilizers means that their supports differ by a sum of plaquette boundaries.
We keep that algebraic language mostly in the appendices.
In the main text, the geometric picture is the more useful one.

If no global odd witness satisfying Eq.~\eqref{eq:global_odd_witness} exists,
then the prescribed signal family cannot be realized as one common logical operator by this clean two-hole construction. In that case, the family
must be modified, split into smaller compatible groups, or treated by another
control strategy.

\begin{figure}[tb]
    \centering
    \includegraphics[width=\linewidth]{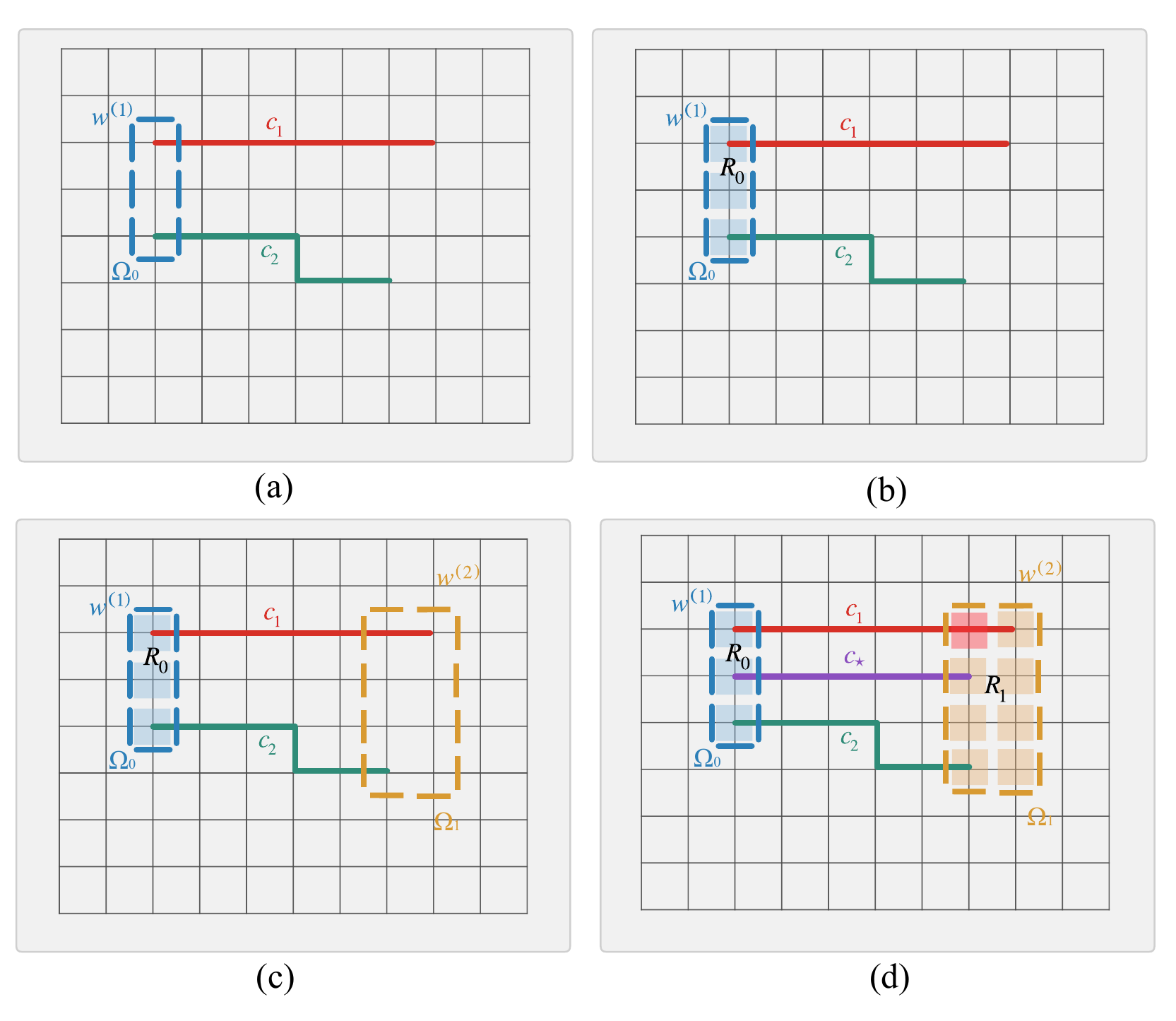}
    \caption{
Stepwise synthesis of two clean rough holes from the global odd-witness condition and the local clean-opening condition.
Panel (a) shows the first simple odd-witness representative $w^{(1)}$.
It crosses every prescribed signal path with odd parity and guides the choice of the first witness-guided region $\Omega_0$.
Panel~(b) opens the first actual rough hole
$R_0\subseteq\Omega_0$ by turning off the star stabilizers in the selected region $R_0$.
Here and below, opening a rough hole also includes the corresponding removal of interior plaquette stabilizers and truncation of the affected plaquette stabilizers along its boundary.
Panel (c) chooses a second simple odd-witness representative $w^{(2)}$ on the
side away from $R_0$.
It guides the choice of second witness-guided region $\Omega_1$.
Panel~(d) opens the second actual rough hole
$R_1\subseteq\Omega_1$ and shows a reference path $c_\star$ connecting the two rough boundaries.
Here $R_1\subsetneq\Omega_1$ because a non-terminal portion of $c_1$ passes through $\Omega_1$.
That portion is excluded from $R_1$ and remains in the active code region.
Proposition~\ref{prop:clean_puncture_realization} therefore applies.
Hence all prescribed signal operators $Z(c_\alpha)$ realize the same nontrivial logical operator $\bar Z$ on the code space.}
\label{fig:witness_to_holes}
\end{figure}

\begin{figure}[tb]
    \centering
    \includegraphics[width=\linewidth]{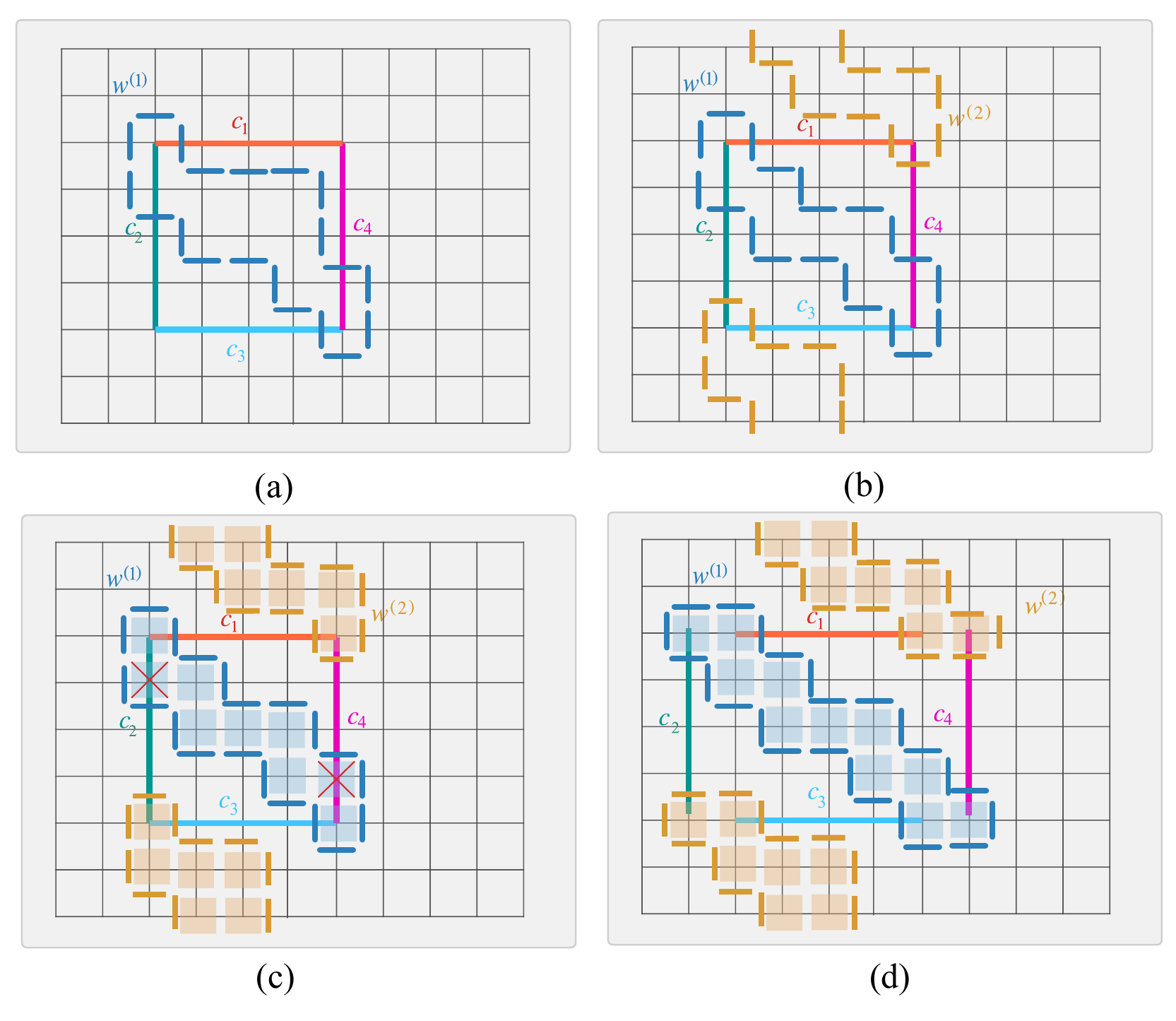}
\caption{
A local obstruction to the clean rough-hole realization step and a simple local repair.
The portions of $w^{(2)}$ that leave through the top and bottom of the displayed window are connected outside the window within a larger planar patch.
The code geometry remains an ordinary planar patch, not a torus.
Panel~(a) shows a first simple odd-witness representative $w^{(1)}$,
which crosses each of the four signal paths $c_1,c_2,c_3,c_4$ exactly once.
Panel~(b) shows a second simple odd-witness representative $w^{(2)}$ that also crosses each signal path exactly once.
The endpoint graph $\Gamma_{\rm end}$ is a four-cycle and is therefore bipartite.
Thus the global odd-witness criterion is satisfied.
In panel~(c), clean rough-hole realization nevertheless fails locally.
The shaded regions are the witness-guided regions $\Omega_0$ and $\Omega_1$.
The two marked sites are non-terminal star-check vertices of the signal paths lying inside these witness-guided regions.
The corresponding star checks must remain active and therefore cannot be included in the actual rough holes.
For each signal path, only its assigned endpoint vertices may contact the
rough boundaries, while all non-terminal star-check vertices must remain in
the active code region.
Excluding these two forbidden star-check sites disconnects the remaining
allowed part of the affected witness-guided region.
Hence no connected actual rough hole satisfying the clean condition can be chosen in panel~(c).
Panel~(d) shows a local repair.
If $c_2$ and $c_4$ are shifted by one column while $c_1$ and $c_3$ are kept unchanged, the resulting witness-guided regions are cleanly openable.
The failure in panel~(c) is therefore a local layout obstruction, not a failure of the global odd-witness criterion.
}
    \label{fig:caveat}
\end{figure}

\subsection{Constructing the punctured surface code}

The global odd-witness criterion is not only a test.
Together with the local clean-opening condition, it gives a concrete construction.
Figure~\ref{fig:witness_to_holes} shows the construction.

We first state the clean-up rule used in the construction.

\emph{Local clean-opening rule.}
Given a witness-guided region $\Omega_i$, choose a simply connected rough hole $R_i\subseteq\Omega_i$ by specifying the star stabilizers to be turned off.
The affected plaquette stabilizers are removed or truncated as described above.
The rough boundary must contain the assigned endpoint vertices, while every non-terminal vertex of every prescribed signal path remains in the active code region.

A signal path may therefore pass through $\Omega_i$.
This is harmless, because its non-terminal portion is carved out when $R_i$ is
chosen.
In practice, one may start from $\Omega_i$ as the candidate hole-opening
region, then carve out any non-terminal parts of the prescribed signal paths
and keep them in the active code region.
The clean-up succeeds if the remaining region still contains the assigned
endpoints and can be chosen as the actual rough hole $R_i$.

We now describe the construction.

\begin{enumerate}[wide, labelwidth=!,labelindent=0pt, leftmargin=2em, label={(\arabic*)}]
\item
By Theorem~\ref{thm:global_odd_witness}, choose a simple odd-witness
representative $w^{(1)}$.
Choose one side of $w^{(1)}$ as the first witness-guided region $\Omega_0$.
The endpoints on this side are assigned to the first rough hole.
The other endpoints are assigned to the second rough hole.
The region $\Omega_0$ is only a guide, not the actual rough hole.

A shorter representative may be used to keep the construction local, but this
choice does not change the topological criterion.

\item
Apply the clean-up rule inside $\Omega_0$ to the endpoints assigned to the
first rough hole.
This gives a simply connected actual rough hole $R_0\subseteq\Omega_0$.
The rough boundary of $R_0$ contains these assigned endpoints, while the
non-terminal parts of all prescribed signal paths remain in the active code
region.

\item
Keep $R_0$ fixed and collapse it to one collapsed vertex $\tilde r_0$.
This collapse is only an auxiliary geometric construction; it is not a physical code-deformation step.
After this collapse, each signal path has one endpoint at $\tilde r_0$ and one
remaining endpoint outside $R_0$.
Thus the quotient endpoint graph is a star graph, possibly with repeated
edges, and is therefore bipartite.

Apply Theorem~\ref{thm:global_odd_witness} to this quotient patch.
This gives a second simple odd-witness representative $w^{(2)}$.
Lift it back to the original patch.
Let $\Omega_1$ be the side of $w^{(2)}$ that contains the remaining endpoints
and does not contain $R_0$.

\item
Apply the same clean-up rule inside $\Omega_1$ to the remaining endpoints.
This gives a simply connected actual rough hole $R_1\subseteq\Omega_1$.
The rough boundary of $R_1$ contains the remaining endpoints, while the
non-terminal parts of all prescribed signal paths remain in the active code
region.
\end{enumerate}

Panels (a)--(d) of Fig.~\ref{fig:witness_to_holes} follow these four steps:
first $w^{(1)}$, then $\Omega_0$ and $R_0$, then $w^{(2)}$, and finally $\Omega_1$ and $R_1$ together with
the reference path $c_\star$.
Appendix~\ref{app:rough-hole-construction-simple-witnesses} gives the
detailed rough-hole construction, and
Appendix~\ref{app:relative-homology-fixed-rough-holes} proves that, for the
resulting fixed holes, all prescribed signal operators $Z(c_\alpha)$ realize
the same nontrivial logical operator $\bar Z$.

Figure~\ref{fig:caveat} shows why clean opening must be checked separately.
Its endpoint graph is an even cycle, so the global odd-witness condition is satisfied.
However, excluding the non-terminal path vertices disconnects the witness-guided regions, and no connected clean rough holes can be chosen for the layout in panel~(c).
This is a local geometric obstruction rather than a failure of the global criterion.
As panel~(d) shows, a small local deformation of the prescribed paths can restore clean openability.

Together, Proposition~\ref{pro:Surface-code_weighted-sum_sensing},
Theorem~\ref{thm:global_odd_witness}, and
Proposition~\ref{prop:clean_puncture_realization} give the full strategy to perform distributed sensing for lattice-embedded $Z$-type couplings.
Proposition~\ref{pro:Surface-code_weighted-sum_sensing} gives the sensing
protocol once all signal operators act as the same logical $\bar Z$.
Theorem~\ref{thm:global_odd_witness} gives the global parity test.
Proposition~\ref{prop:clean_puncture_realization} gives the local condition
that turns the witness-guided regions into actual clean rough holes.
When both the global odd-witness condition and the local clean-opening condition hold, the prescribed signal operators are compressed into one protected logical generator.

\section{Lattice realization of abstract support structures}

The preceding section treats a prescribed family of lattice-embedded signal paths $\{c_\alpha\}$.
It determines when compatible rough holes can be chosen so that all operators $Z(c_\alpha)$ realize the same nontrivial logical $\bar Z$.
We now turn to the complementary placement problem.
Suppose instead that the signal Hamiltonian specifies only a family of abstract supports $\{S_\alpha\}$, namely which qubit labels appear in each term and which labels are shared between different terms.
This abstract support structure fixes the overlaps and symmetric differences between the signal terms, but it does not yet specify a placement on the planar surface-code lattice.

A lattice realization assigns distinct abstract qubit labels to distinct data-qubit edges, while every repeated occurrence of a shared label is assigned to the same edge.
After placement, we denote the image of $S_\alpha$ by
$c_\alpha\subseteq E$.
We require each $c_\alpha$ to form a simple primal path and seek two rough holes $R_0$ and $R_1$ such that every $c_\alpha$ is a legal path between the two rough boundaries.
The goal is that all corresponding operators $Z(c_\alpha)$ realize the same nontrivial logical $\bar Z$.

If the abstract supports are pairwise disjoint, they are straightforward to realize on a sufficiently large patch.
One may choose two clean rough holes and realize each abstract support as a legal primal path between the two rough boundaries,
with the resulting paths mutually separated.
For these fixed holes, Proposition~\ref{prop:clean_puncture_realization} implies that all operators $Z(c_\alpha)$ have the same logical $\bar Z$ action.

When the abstract supports overlap, each shared qubit label must instead be mapped to one shared data-qubit edge.
Such placements are constrained by a \textit{stabilizer-difference condition}.
Let $\mathcal S_Z^{\mathrm{act}}$ denote the subgroup generated by all active full and truncated $Z$-type plaquette stabilizers.
If two placed signal operators realize the same nontrivial logical $\bar Z$, then their product acts trivially on the code space.
\begin{equation}
Z(c_a)Z(c_b)
=
Z(c_a\oplus c_b)
\in
\mathcal S_Z^{\mathrm{act}},
\label{eq:stabilizer-difference-principle}
\end{equation}
where $\oplus$ denotes the symmetric difference of the two edge sets.
Conversely, once one of the two operators is known to realize the logical $\bar Z$, Eq.~\eqref{eq:stabilizer-difference-principle} ensures that the other has the same logical action.

To illustrate this general condition, we now specialize to three-body $Z$-type signal terms.
Thus $|S_\alpha|=3$, and each placed support $c_\alpha$ is required to form a simple three-edge primal path.

For the overlapping constructions below, we allow the punctured code to have distance two rather than imposing the distance-three condition
used elsewhere in the paper. A distance-two code can detect an
arbitrary single-qubit Pauli error and can correct a known
single-qubit erasure, but it cannot correct an unknown single-qubit
Pauli error. This reduction in protection permits the local
rough-boundary geometry needed for nontrivial overlaps between
three-body terms. Under a distance-three restriction, the examples
considered here reduce to disjoint three-body supports, which can be
trivially placed independently on a sufficiently large patch, for example on a
thin surface-code strip.

We consider two local overlap patterns: two three-body supports may share either one data qubit or two data
qubits. 

\subsection{One-qubit overlaps: an open, nonbranching backbone}

Consider first two signal paths that share one data qubit. In the representative pair in Fig.~\ref{fig:three_qubit_overlap}(a), 
$Z(c_1)=Z_1Z_2Z_3$ and $Z(c_2)=Z_3Z_4Z_5$.
They share the factor $Z_3$ at the rough boundary of $R_0$.
Let $p_{\mathrm L}$ and $p_{\mathrm R}$ denote the two adjacent green plaquettes along the rough boundary of $R_1$.
Their active weight-three truncated stabilizers are
$B_{p_{\mathrm L}}^Z=Z_1Z_2Z_6$ and
$B_{p_{\mathrm R}}^Z=Z_4Z_5Z_6$.
It follows that
\begin{equation}
\label{eq:fig5-local-stabilizer-difference}
\begin{aligned}
Z(c_1)Z(c_2)
&=(Z_1Z_2Z_3)(Z_3Z_4Z_5) \\
&=Z_1Z_2Z_4Z_5 \\
&=(Z_1Z_2Z_6)(Z_4Z_5Z_6) \\
&=B_{p_{\mathrm L}}^ZB_{p_{\mathrm R}}^Z
\in\mathcal S_Z^{\mathrm{act}} .
\end{aligned}
\end{equation}
Here $Z_3$ cancels between the two signal operators, while $Z_6$ cancels between the two truncated plaquette stabilizers.
Thus $c_1$ and $c_2$ satisfy the stabilizer-difference condition.

The same local relation can be translated along the strip and
reflected between the two rough boundaries. Once one path in the
sequence is identified with the logical $\bar Z$, Eq.~\eqref{eq:fig5-local-stabilizer-difference} shows
that every successive path has the same logical action. To keep every
three-edge path connected to both rough boundaries, the shared
terminal data qubit alternates between the two boundaries. The
resulting extended sequence of overlapping supports is the
\emph{backbone} structure shown in Fig.~\ref{fig:three_qubit_overlap}(a). 

\begin{figure}
    \centering
\includegraphics[width=\columnwidth]{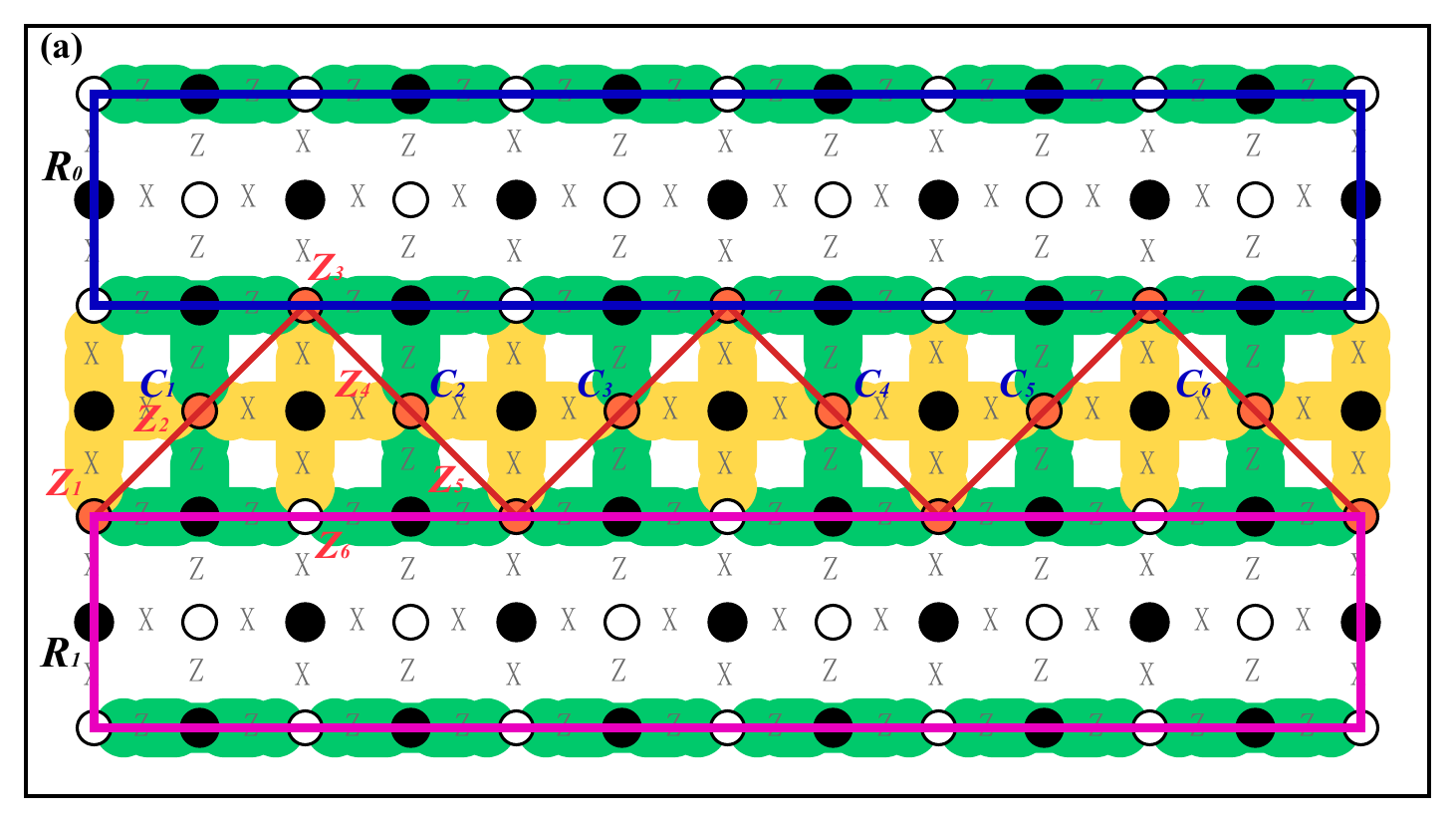}\par
\vspace{0.6ex}
\includegraphics[width=\columnwidth]{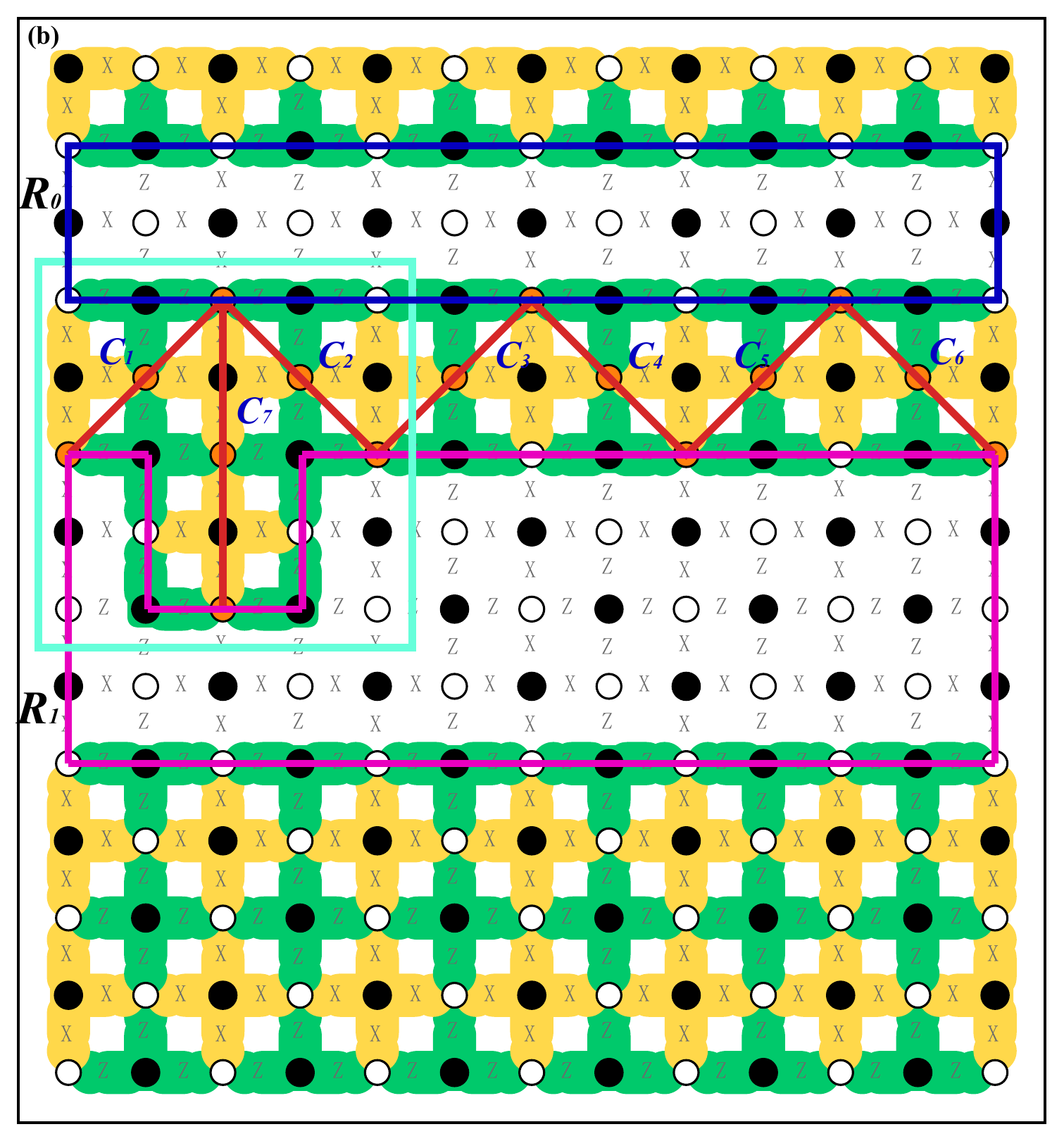}
    \caption{
  \textbf{Compatible placement of one-qubit-overlapping three-body $Z$-type signal paths on a punctured surface code lattice.}
  White circles denote data qubits, and black circles denote syndrome-measurement qubits.
  The green and yellow shapes indicate the underlying locations of
  $Z$-type and $X$-type stabilizers, respectively.
    The blue and pink outlines mark the rough boundaries associated with
  the two rough holes $R_0$ and $R_1$.
Each labeled red path $C_i$ represents the placed support $c_i$ of one three-body signal term, namely a simple three-edge primal path.
  (a) Six paths form a single open, nonbranching backbone.
Successive paths share one terminal data qubit, with the shared
contacts alternating between the two rough boundaries.
The labels $Z_1,\ldots,Z_6$ identify the representative local
stabilizer relation derived in
Eq.~\eqref{eq:fig5-local-stabilizer-difference}.
    (b) Local completion of the same backbone.
The cyan box highlights the added path $C_7$ near the left endpoint.
The paths $C_1$, $C_2$, and $C_7$ share the same terminal data qubit
on the rough boundary of $R_0$, while $C_7$ terminates on the nearby
indentation of the $R_1$ boundary.
This addition remains local and neither creates a second extended
backbone nor closes the original backbone.
    }
  \label{fig:three_qubit_overlap}
\end{figure}

\emph{Nonbranching backbone.}
Within this construction, successive supports must form a single
sequence (see Fig.~\ref{fig:three_qubit_overlap}(a)),
\begin{equation}
\begin{aligned}
Z_{q_1}Z_{q_2}Z_{q_3}
&\;\longrightarrow\;
Z_{q_3}Z_{q_4}Z_{q_5}
\;\longrightarrow\;
Z_{q_5}Z_{q_6}Z_{q_7}
\;\longrightarrow\;\cdots ,
\end{aligned}
\end{equation}
where each arrow represents a one-qubit overlap between neighboring
terms. This sequence is the extended backbone.

A genuine branch is not allowed. For example, the pattern
\begin{equation}
\begin{aligned}
&Z_{q_1}Z_{q_2}Z_{q_3}
\;\longrightarrow\;
Z_{q_3}Z_{q_4}Z_{q_5}
\;\longrightarrow\;
Z_{q_5}Z_{q_6}Z_{q_7},\\
&\phantom{Z_{q_1}Z_{q_2}}\searrow\\
&\phantom{Z_{q_1}Z_{q_2}Z_{q_3}\;\longrightarrow\;}
Z_{q_3}Z_{q_8}Z_{q_9}
\;\longrightarrow\;
Z_{q_9}Z_{q_{10}}Z_{q_{11}} .
\end{aligned}
\end{equation}
splits at $q_3$ into two independent long-range continuations. In the
local placement considered here, both continuations cannot be embedded
while simultaneously preserving the required stabilizer relations and
the legal rough-boundary contacts of every path. The extended structure
must therefore remain nonbranching.

\emph{Local completion.}
The construction does permit one additional local motif near an
endpoint of the backbone. At the operator level, it has the form (see Fig.~\ref{fig:three_qubit_overlap}(b))
\begin{equation}
\begin{aligned}
Z_{q_1}Z_{q_2}Z_{q_3}
&\;\longrightarrow\;
Z_{q_3}Z_{q_4}Z_{q_5}
\;\longrightarrow\;
Z_{q_5}Z_{q_6}Z_{q_7}
\;\longrightarrow\;\cdots\\
&\hspace{2.9em}\uparrow\\[-0.3ex]
&\hspace{2.2em}Z_{q_3}Z_{q_8}Z_{q_9} .
\end{aligned}
\end{equation}
This extra term is confined to a local neighborhood, as in
Fig.~\ref{fig:three_qubit_overlap}(b). Its detailed geometry can vary, but it cannot be continued
into another long sequence. It therefore does not produce a second
backbone. 

Geometric constraints on shared rough-boundary contacts, as well as the local extensions permitted by different lattice embeddings, are discussed in Appendix~\ref{app:local-overlap-constraints}. 

These placement rules explain why the backbone does not close.
The permitted local continuation moves connect neighboring terms end
to end but do not provide a route back to an earlier overlap junction.
Because the patch is planar and non-toroidal, the backbone also cannot
close by wrapping around the lattice. Finally, the local
completion is terminal and cannot seed another extended continuation.
The allowed one-qubit-overlap realization is therefore a single open,
nonbranching backbone.

\begin{figure}[htbp]
    \centering
    \includegraphics[width=1\linewidth]{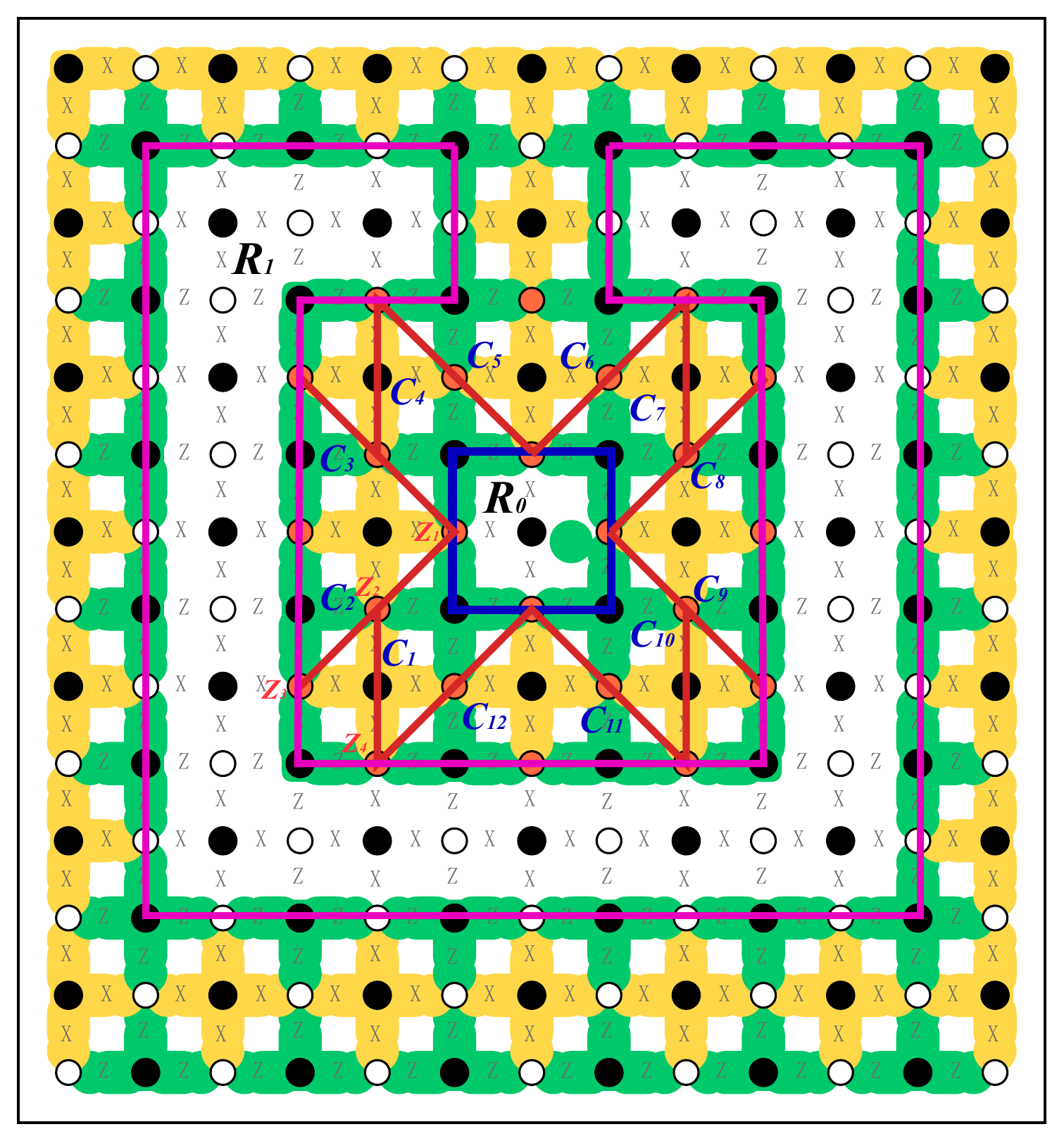}
    \caption{\textbf{A weight-two boundary stabilizer realizes a local $90^\circ$ turn between
overlapping three-body $Z$-type signal paths.}
The background notation follows Fig.~\ref{fig:three_qubit_overlap}.
The blue inner and pink outer outlines mark the rough boundaries of
$R_0$ and $R_1$, respectively, with the active code region lying between
them.
    The labeled red paths $C_1,\ldots,C_{12}$ are simple three-edge primal paths connecting the two rough boundaries.
The pair labeled $C_1$ and $C_2$ forms the representative lower-left turn
analyzed in the main text.
The labels $Z_1$ and $Z_2$ mark the shared data qubits, whereas $Z_3$ and
$Z_4$ mark the two non-shared data qubits on the perpendicular segments
meeting at an L-shaped corner of the $R_1$ boundary.
The operator $Z_3Z_4$ is the active weight-two truncated plaquette
stabilizer at this corner.
Repeating the same local motif gives the displayed sequence
$C_1,\ldots,C_{12}$.}
    \label{fig:weight2}
\end{figure}

\subsection{Two-qubit overlaps: a local 90$^\circ$ turn}

Now consider two three-body paths that share two data qubits, as in
Fig.~\ref{fig:weight2}. 

Write
\begin{equation}
Z(c_a)=Z_1Z_2Z_3,
\qquad
Z(c_b)=Z_1Z_2Z_4 .
\end{equation}
The shared factors $Z_1Z_2$ cancel in their product, giving
\begin{equation}
Z(c_a)Z(c_b)=Z_3Z_4 .
\end{equation}
Thus $c_a\oplus c_b$ consists of the two non-shared data-qubit edges.
By Eq.~\eqref{eq:stabilizer-difference-principle}, compatibility requires
$Z_3Z_4\in\mathcal S_Z^{\mathrm{act}}$.

In the placement shown in Fig.~\ref{fig:weight2}, the data qubits carrying $Z_3$ and
$Z_4$ lie on two perpendicular boundary segments that meet at an
L-shaped corner of the rough boundary of $R_1$. The product $Z_3Z_4$
is exactly the active weight-two truncated plaquette stabilizer at
that corner. The two signal paths therefore have the same logical
action, and their two non-shared edges form a local $90^\circ$ turn.
Repeating this motif produces the turning sequence displayed in the
figure.

\subsection{Design summary}

The stabilizer-difference condition in Eq.~\eqref{eq:stabilizer-difference-principle} is the central
algebraic requirement for any lattice realization with overlapping
supports. For the three-body interactions considered here, it leads
to two useful local building blocks:

\begin{itemize}[wide, labelwidth=!, labelindent=0pt, leftmargin=2em]

\item With a one-qubit overlap, neighboring supports can be
concatenated into one open, nonbranching backbone. A strictly local
completion may be added near an endpoint when the boundary geometry
allows it.

\item With a two-qubit overlap, the two non-shared edges must support
an active weight-two truncated plaquette stabilizer. An L-shaped rough-boundary
corner provides such a stabilizer and realizes a local $90^\circ$
turn.

\end{itemize}

These building blocks may be combined as long as every placed support
remains a legal path between the rough boundaries of $R_0$ and $R_1$.
Proposition~\ref{prop:clean_puncture_realization} then guarantees that all physical signal operators
$Z(c_\alpha)$ represent the same nontrivial logical $\bar Z$.

\section{Discussion}

In this paper, we study the metrological task of estimating a chosen linear functional $q=\sum_{\alpha} s_{\alpha}\lambda_{\alpha}$ of local many-body $Z$-type couplings on a punctured surface-code patch, where the weights $s_{\alpha}$ are known and the coupling strengths $\lambda_{\alpha}$ are unknown.
For a prescribed family of lattice-embedded $Z$-type signal paths,
we derive a global odd-witness criterion for their topological compatibility.
When a separate local clean-opening condition is also satisfied, the witness
construction identifies two actual rough holes such that all corresponding
physical signal operators $Z(c_\alpha)$ realize the same nontrivial logical $\bar Z$.
Once this compatibility is achieved, the logical Ramsey protocol reduces
the physical sensing Hamiltonian within the code space to an effective
logical Hamiltonian proportional to $q\bar Z$.
The target phase can be accumulated either by parallel loading on a fixed
punctured patch or by sequential hole motion, with the two implementations
having different dwell-time and control overheads.
In this way, the protocol provides a surface-code-protected realization of
weighted-sum sensing in a distributed quantum sensor network.
While we focus on $Z$-type signals for concreteness, the same framework applies to $X$-type couplings by exchanging $X$ and $Z$ throughout.

For the complementary design problem, in which abstract support structures are specified before lattice realization,
pairwise disjoint supports can be placed as separated legal primal paths on a sufficiently large patch, whereas lattice realizations of overlapping
supports are more restrictive and must satisfy the stabilizer-difference condition.
For the three-body constructions considered here, compatible one-qubit overlaps combine into a single open, nonbranching backbone, which may admit a local completion when the boundary geometry permits.
A two-qubit overlap can instead be realized as a local $90^\circ$ turn, in which the product of the two corresponding physical signal operators is an active weight-two truncated plaquette stabilizer.
Subject to the requirement that every placed support remain a legal path between the two rough boundaries, these two building blocks yield more elaborate overlapping layouts whose physical signal operators still realize the same nontrivial logical $\bar Z$.

Assuming syndrome extraction can be performed accurately and sufficiently fast during sensing, all errors that remain correctable by the code do not erase the accumulated logical phase.
Moreover, the phase is stored nonlocally, so an observer restricted to a correctable region cannot recover $q$ from local measurements within that region (Appendix~\ref{app:logical-ramsey-protocol}).
Future work can extend the same geometric criteria to richer code families, higher-weight interactions, and more general Pauli structure.
The appendices give the full proofs and the detailed construction algorithms.

\section*{Acknowledgments}

L.Q. thanks Haidong Yuan, Jun Feng, and Zhiyao Hu for helpful discussions. S.Z. acknowledges support from Perimeter Institute for Theoretical Physics, a research institute supported in part by the Government of Canada through the Department of Innovation, Science and Economic Development Canada and by the Province of Ontario through the Ministry of Colleges and Universities.

\section*{Author Contributions}

L.Q. and S.Z. jointly led the development of this work. L.Q. carried out the main theoretical analysis, developed the mathematical framework and proofs, prepared the figures, and wrote the original draft. S.Z. formulated the research problem, proposed the central idea, supervised the project, contributed substantially to the theoretical framework and interpretation of the results, and revised the manuscript. Z.O. contributed to theoretical discussions, analysis of the results, and revision of the manuscript. All authors reviewed and approved the final manuscript.

Large language model tools were used to assist with English-language editing and rephrasing. Their use was limited to language-related assistance. All scientific content was developed and verified by the authors, who take full responsibility for the work.

\appendix
\onecolumngrid
\counterwithin{theorem}{section}
\counterwithin{proposition}{section}
\counterwithin{lemma}{section}
\counterwithin{corollary}{section}

\renewcommand{\thetheorem}{\thesection.\arabic{theorem}}
\renewcommand{\theproposition}{\thesection.\arabic{proposition}}
\renewcommand{\thelemma}{\thesection.\arabic{lemma}}
\renewcommand{\thecorollary}{\thesection.\arabic{corollary}}

\section{Logical Invariance under Stabilizer Updates and Sequential Sensing}
\label{app:stabilizer-updates-logical-invariance}
This appendix gives the stabilizer-algebra justification for the sequential
hole-motion protocol used in the main text. In that protocol, local Pauli
measurements replace active stabilizer constraints and therefore change the
instantaneous code-space branch. All statements below are conditional on the
recorded measurement outcomes.

The point needed in the main text is that these local stabilizer updates
transport the encoded logical state when the logical Pauli frame is updated
consistently with the measurement record. Thus the physical supports of
convenient logical Pauli representatives may change, and known
measurement-record-dependent signs may appear, but the encoded Bloch vector
does not change. In particular, a Ramsey phase accumulated during an earlier
dwell interval is not erased by later code deformation; any known Pauli-frame
signs are tracked classically and absorbed into subsequent signed-control
choices.

\subsection*{Elementary stabilizer updates and logical invariance}

We work on the $n$-qubit Hilbert space
$\mathcal H=(\mathbb C^2)^{\otimes n}$ and denote the $n$-qubit Pauli
group by
$\mathcal P_n=\{\pm 1,\pm i\}\times\{I,X,Y,Z\}^{\otimes n}$.
Whenever a Pauli operator is used below as a stabilizer constraint, a logical
Pauli representative, a measured observable, or a sensing-Hamiltonian term,
we choose a Hermitian representative, so that $L=L^\dagger$ and $L^2=I$.
For a stabilizer constraint, its Hermitian sign specifies the eigenspace
defining the current code-space branch. For a sensing-Hamiltonian term, its
Hermitian sign specifies the direction of logical phase accumulation.

Let
\begin{equation}
\mathcal S
=
\left\langle
P_1,\ldots,P_m
\right\rangle
\subseteq
\mathcal P_n
\end{equation}
be generated by independent mutually commuting Hermitian Pauli operators,
with $P_i^2=I$, $[P_i,P_j]=0$ for all $i,j$, and
$-I\notin\mathcal S$. Its code space and projector are
\begin{equation}
\mathcal C
=
\left\{
|\psi\rangle\in\mathcal H
\ \middle|\
P_i|\psi\rangle=|\psi\rangle
\text{ for all } i
\right\},
\qquad
P_{\mathcal C}
=
\prod_{i=1}^{m}\frac{I+P_i}{2}.
\end{equation}

As in the two-rough-hole surface-code construction used in the main text, we
assume that $\mathcal C$ encodes one logical qubit. The Pauli normalizer of
$\mathcal S$ is
\begin{equation}
\mathcal N(\mathcal S)
=
\left\{
L\in\mathcal P_n
\ \middle|\
[L,S]=0
\text{ for all } S\in\mathcal S
\right\}.
\end{equation}
We choose Hermitian physical representatives
$\bar X,\bar Z\in\mathcal N(\mathcal S)$ for the logical Pauli operators,
with $\bar X\bar Z=-\bar Z\bar X$, and define
$\bar Y:=i\bar X\bar Z$. A general encoded state may then be written as
\begin{equation}
\rho_L
=
P_{\mathcal C}
\frac{
I+\alpha_x\bar X+\alpha_y\bar Y+\alpha_z\bar Z
}{2}
P_{\mathcal C},
\qquad
\|\vec{\alpha}\|\le 1,
\label{eq:app-initial-logical-state}
\end{equation}
where $\vec{\alpha}=(\alpha_x,\alpha_y,\alpha_z)$ is the Bloch vector of
the encoded logical qubit.

We will use the standard equivalence between physical representatives of
the same logical operator. If $L\in\mathcal N(\mathcal S)$ and
$S\in\mathcal S$, then $S$ acts as the identity on the code space and
therefore
\begin{equation}
P_{\mathcal C}(LS)P_{\mathcal C}
=
P_{\mathcal C}LP_{\mathcal C}.
\label{eq:app-physical-equivalence}
\end{equation}
Thus multiplying a logical Pauli representative by a stabilizer may change
its physical support without changing its encoded action. By contrast, an
additional overall sign is not removed by stabilizer equivalence:
$P_{\mathcal C}(-LS)P_{\mathcal C}
=
-P_{\mathcal C}LP_{\mathcal C}$.
This distinction will matter when a physical signal operator is compared with
the logical sensing generator in a later dwell interval.

We now consider the measurement of a Hermitian Pauli observable
$Q=Q^\dagger$ with $Q^2=I$. Replacing $Q$ by $-Q$ only exchanges the labels
of the two measurement outcomes. We therefore fix a convenient Hermitian
representative for each measured observable. In particular, if a measured
observable is proportional to an already imposed stabilizer constraint, we
choose the representative lying in $\mathcal S$, so that its deterministic
outcome is labeled $+1$.

With this convention, measuring a Hermitian Pauli observable $Q$ on a code
state gives three conceptually distinct cases:
\begin{enumerate}[wide, labelwidth=!, labelindent=0pt, leftmargin=2em, label={(\arabic*)}]
    \item If $Q\in\mathcal S$, then the outcome is deterministically $+1$
    and the encoded state is unchanged.

    \item If $Q\in\mathcal N(\mathcal S)$ but is not proportional to any
    element of $\mathcal S$, then, in the one-logical-qubit setting
    considered here, $Q$ acts as a nontrivial logical Pauli operator on the
    code space. Measuring $Q$ therefore performs a logical measurement on
    the encoded qubit.

    \item If $Q\notin\mathcal N(\mathcal S)$, then $Q$ anticommutes with at
    least one stabilizer generator. As an error operator, it creates a
    nonzero syndrome and is not a logical Pauli operator for the original
    code. As a measurement operator, however, it can be used to deform the
    code: after a suitable change of stabilizer generating set, one
    anticommuting stabilizer constraint is replaced by the
    outcome-dependent measured constraint.
\end{enumerate}

The third case is the one relevant for the code-deformation protocol studied
below. We first rewrite the generators of the original stabilizer group so
that $Q$ anticommutes with exactly one displayed generator. We then replace
that generator by the measured constraint determined by the recorded
outcome. Finally, we update the logical Pauli representatives so that they
represent the same encoded logical axes in the new code-space branch.

Because $Q\notin\mathcal N(\mathcal S)$, we may relabel the original
generators so that $\{Q,P_1\}=0$. Define an equivalent generating set by
\begin{equation}
\tilde P_1
:=
P_1,
\qquad
\tilde P_j
:=
\begin{cases}
P_j,
&
[Q,P_j]=0,
\\[4pt]
P_jP_1,
&
\{Q,P_j\}=0,
\end{cases}
\qquad
j=2,\ldots,m.
\label{eq:app-adapted-generator-definition}
\end{equation}
If $P_j$ anticommutes with $Q$, then $P_jP_1$ commutes with $Q$, since both
factors anticommute with $Q$. Moreover, whenever
$\tilde P_j=P_jP_1$, one has $P_j=\tilde P_j\tilde P_1$. Hence this
replacement is invertible and leaves the stabilizer group unchanged:
\begin{equation}
\mathcal S
=
\left\langle
\tilde P_1,\ldots,\tilde P_m
\right\rangle.
\end{equation}
The adapted generating set satisfies
\begin{equation}
\{Q,\tilde P_1\}=0,
\qquad
[Q,\tilde P_j]=0
\quad
\text{for all } j\ge 2.
\label{eq:app-adapted-generators}
\end{equation}
Thus measuring $Q$ can be viewed as replacing the single displayed
constraint $\tilde P_1$ while retaining the constraints
$\tilde P_2,\ldots,\tilde P_m$.

Let $\mu\in\{+1,-1\}$ denote the recorded outcome of measuring $Q$, and
define
\begin{equation}
\Pi_\mu
:=
\frac{I+\mu Q}{2},
\qquad
p_\mu
:=
\operatorname{Tr}\!\left(\rho_L\Pi_\mu\right),
\qquad
\rho_{L,\mu}'
:=
\frac{\Pi_\mu\rho_L\Pi_\mu}{p_\mu}.
\label{eq:app-measurement-definitions}
\end{equation}
Because $\rho_L$ is supported on the original code space, it is stabilized
by $\tilde P_1$. Using $\{Q,\tilde P_1\}=0$, we obtain
\begin{align}
\operatorname{Tr}(\rho_LQ)
&=
\operatorname{Tr}\!\left(
\tilde P_1\rho_L\tilde P_1Q
\right)
\nonumber\\
&=
\operatorname{Tr}\!\left(
\rho_L\tilde P_1Q\tilde P_1
\right)
\nonumber\\
&=
-\operatorname{Tr}(\rho_LQ).
\end{align}
Therefore,
\begin{equation}
\operatorname{Tr}(\rho_LQ)=0,
\qquad
p_\mu=\frac{1}{2}.
\label{eq:app-equal-probability}
\end{equation}

The outcome $\mu$ specifies the updated stabilizer branch. Let
$P_0:=\prod_{j=2}^{m}(I+\tilde P_j)/2$ denote the projector onto the common
$+1$ eigenspace of the retained stabilizer constraints. The updated
stabilizer group and the corresponding code-space projector are
\begin{equation}
\mathcal S_\mu'
=
\left\langle
\mu Q,\tilde P_2,\ldots,\tilde P_m
\right\rangle,
\qquad
P_{\mathcal C_\mu'}
=
\Pi_\mu P_0.
\label{eq:app-updated-stabilizer}
\end{equation}
The updated generators commute by
Eq.~\eqref{eq:app-adapted-generators}. They remain independent because $\mu Q$ cannot be generated by
$\tilde P_2,\ldots,\tilde P_m$. Indeed, every product of the latter
commutes with $\tilde P_1$, whereas $\mu Q$ anticommutes with
$\tilde P_1$. Since $p_\mu=1/2$, the post-measurement
state is nonzero, so the measured constraints are consistent and
$\mathcal C_\mu'$ is a nonempty code space of the same dimension as
$\mathcal C$.

We now update the logical Pauli representatives across this stabilizer
replacement. For each $\bar L\in\{\bar X,\bar Z\}$, define
\begin{equation}
\epsilon_L
:=
\begin{cases}
0,
&
[Q,\bar L]=0,
\\[4pt]
1,
&
\{Q,\bar L\}=0,
\end{cases}
\qquad
\bar L'
:=
\bar L\tilde P_1^{\epsilon_L}.
\label{eq:app-transported-representatives}
\end{equation}
If $\bar L$ already commutes with $Q$, it is retained unchanged. If
$\bar L$ anticommutes with $Q$, it is multiplied by the removed old
constraint $\tilde P_1$. This multiplication has two simultaneous effects.

First, $\bar L'$ represents the same logical operator as $\bar L$ before
the measurement. Indeed, $\tilde P_1$ is a stabilizer of the original code,
so Eq.~\eqref{eq:app-physical-equivalence} gives
\begin{equation}
P_{\mathcal C}\bar L'P_{\mathcal C}
=
P_{\mathcal C}\bar LP_{\mathcal C},
\qquad
\bar L\in\{\bar X,\bar Z\}.
\label{eq:app-same-old-logical-action}
\end{equation}

Second, $\bar L'$ is compatible with the updated stabilizer constraints.
If $\epsilon_L=1$, then both $\bar L$ and $\tilde P_1$ anticommute with
$Q$, so their product commutes with $Q$; if $\epsilon_L=0$, this commutation
already holds by definition. Hence $[\bar L',\mu Q]=0$. In addition,
$\bar L$ commutes with every original stabilizer generator because it is a
logical representative, and $\tilde P_1$ commutes with every retained
generator because the stabilizer group is Abelian. Therefore,
\begin{equation}
[\bar L',\mu Q]=0,
\qquad
[\bar L',\tilde P_j]=0
\quad
\text{for all } j\ge 2,
\qquad
\bar L\in\{\bar X,\bar Z\}.
\label{eq:app-updated-logical-commutation}
\end{equation}
It follows that $\bar X',\bar Z'\in\mathcal N(\mathcal S_\mu')$.

The updated representatives preserve the logical Pauli anticommutation
relation. Since $\bar X$ and $\bar Z$ commute with $\tilde P_1$, we obtain
\begin{align}
\bar X'\bar Z'
&=
\bar X\tilde P_1^{\epsilon_X}
\bar Z\tilde P_1^{\epsilon_Z}
\nonumber\\
&=
\bar X\bar Z
\tilde P_1^{\epsilon_X+\epsilon_Z}
\nonumber\\
&=
-\bar Z'\bar X'.
\end{align}
Define $\bar Y':=i\bar X'\bar Z'$. Because $\bar X'$ and $\bar Z'$
anticommute, neither can belong to $\mathcal S_\mu'$: every stabilizer
element commutes with every operator in
$\mathcal N(\mathcal S_\mu')$, whereas $\bar X'$ and $\bar Z'$ do not
commute with each other. Therefore, $\bar X',\bar Y',\bar Z'$ form a
nontrivial logical Pauli frame for the updated branch
$\mathcal C_\mu'$.

It remains to show that the encoded state is carried to the updated branch
without any change in its logical Bloch vector. Define
\begin{equation}
R'(\vec{\alpha})
:=
\frac{
I+\alpha_x\bar X'
+\alpha_y\bar Y'
+\alpha_z\bar Z'
}{2}.
\end{equation}
We now record the commutation relations needed in the state-update
calculation and explain their origin.

First, each updated logical representative is obtained by multiplying an
original logical representative by a power of the old stabilizer
$\tilde P_1$. Both the original logical representative and $\tilde P_1$
commute with every generator of the original stabilizer group. Hence
$\bar X'$, $\bar Y'$, and $\bar Z'$ commute with every original stabilizer
generator, and therefore
\begin{equation}
[R'(\vec{\alpha}),P_{\mathcal C}]=0.
\label{eq:app-R-old-projector-commutation}
\end{equation}
This relation states that the updated logical representatives still preserve
the original code space before the measurement.

Second, by construction each updated logical representative commutes with
the measured operator $Q$. Since $\Pi_\mu=(I+\mu Q)/2$, it follows that
\begin{equation}
[R'(\vec{\alpha}),\Pi_\mu]=0.
\label{eq:app-R-measurement-projector-commutation}
\end{equation}
This relation states that the updated logical representatives are compatible
with the measured constraint defining the selected outcome branch.

Third, Eq.~\eqref{eq:app-updated-logical-commutation} shows that the
updated logical representatives commute with every retained stabilizer
generator $\tilde P_j$ for $j\ge 2$. Since $P_0$ is the product of the
corresponding retained-constraint projectors, we have
\begin{equation}
[R'(\vec{\alpha}),P_0]=0.
\label{eq:app-R-retained-projector-commutation}
\end{equation}
Together with
Eq.~\eqref{eq:app-R-measurement-projector-commutation} and
$P_{\mathcal C_\mu'}=\Pi_\mu P_0$, this implies
\begin{equation}
[R'(\vec{\alpha}),P_{\mathcal C_\mu'}]=0.
\label{eq:app-R-new-projector-commutation}
\end{equation}
Thus the updated logical representatives also preserve the new code-space
branch after the measurement.

Because $\bar X'$ and $\bar Z'$ act identically to $\bar X$ and $\bar Z$
on $\mathcal C$, and because the same is then true for
$\bar Y'=i\bar X'\bar Z'$, the initial encoded state in
Eq.~\eqref{eq:app-initial-logical-state} may be written as
$\rho_L=P_{\mathcal C}R'(\vec{\alpha})P_{\mathcal C}$.
Furthermore, since $\Pi_\mu$ commutes with $P_0$ and
$\Pi_\mu\tilde P_1\Pi_\mu=0$, we have
\begin{align}
\Pi_\mu P_{\mathcal C}\Pi_\mu
&=
\Pi_\mu
\frac{I+\tilde P_1}{2}
P_0
\Pi_\mu
\nonumber\\
&=
\frac{1}{2}\Pi_\mu P_0
\nonumber\\
&=
\frac{1}{2}P_{\mathcal C_\mu'}.
\label{eq:app-projector-identity}
\end{align}
Using Eqs.~\eqref{eq:app-R-old-projector-commutation},
\eqref{eq:app-R-measurement-projector-commutation},
\eqref{eq:app-R-new-projector-commutation},
and \eqref{eq:app-projector-identity}, together with $p_\mu=1/2$, we
obtain
\begin{align}
\rho_{L,\mu}'
&=
\frac{
\Pi_\mu
P_{\mathcal C}
R'(\vec{\alpha})
P_{\mathcal C}
\Pi_\mu
}{
p_\mu
}
\nonumber\\
&=
2R'(\vec{\alpha})
\Pi_\mu P_{\mathcal C}\Pi_\mu
\nonumber\\
&=
R'(\vec{\alpha})P_{\mathcal C_\mu'}
\nonumber\\
&=
P_{\mathcal C_\mu'}
R'(\vec{\alpha})
P_{\mathcal C_\mu'}
\nonumber\\
&=
P_{\mathcal C_\mu'}
\frac{
I+\alpha_x\bar X'
+\alpha_y\bar Y'
+\alpha_z\bar Z'
}{2}
P_{\mathcal C_\mu'}.
\label{eq:app-logical-invariance}
\end{align}

Equation~\eqref{eq:app-logical-invariance} is the elementary logical
invariance statement used below. Measuring $Q$ replaces one stabilizer
constraint by the outcome-dependent constraint $\mu Q$ and changes the
instantaneous code space from $\mathcal C$ to $\mathcal C_\mu'$. The logical
Pauli representatives are simultaneously updated from
$(\bar X,\bar Y,\bar Z)$ to $(\bar X',\bar Y',\bar Z')$, while the encoded
Bloch vector $\vec{\alpha}$ remains unchanged.

In the surface-code geometry below, we apply this update rule twice. Only after the final branch has been reached do we make one additional stabilizer-equivalent representative choice by shortening the final logical $X$ loop so that its support visibly surrounds the displaced rough hole.
This final shortening does not enter the logical-invariance proof; it only
provides a geometrically convenient logical representative in the final
code-space branch.

\subsection*{Example: moving one rough hole by one lattice cell}

We now give a concrete local realization of the preceding elementary update
for an $X$-cut hole, namely, a rough hole in the convention used in the main
text. This is the $X\leftrightarrow Z$ counterpart of the standard
single-cell motion of a $Z$-cut hole described in
Ref.~\cite{PhysRevA.86.032324}. For an $X$-cut hole, a logical
$\bar Z$ is represented by a $Z$-string connecting the two rough holes,
whereas its conjugate logical $\bar X$ is represented by an $X$ loop
surrounding one of the holes.

In the algebraic discussion below, the superscript $k=0,1,2$ labels the
logical Pauli representatives before either local measurement, after the
first measurement, and after the second measurement, respectively. The three panels in Fig.~\ref{fig:FIG6} illustrate the local
deformation process that moves the lower rough hole downward by one lattice
cell. They show the initial configuration, the intermediate enlarged-hole
configuration, and the completed displaced-hole configuration, together with
convenient representatives of the logical Pauli operators during the
deformation.

\begin{figure}[t]
    \centering
    \includegraphics[width=0.82\linewidth]{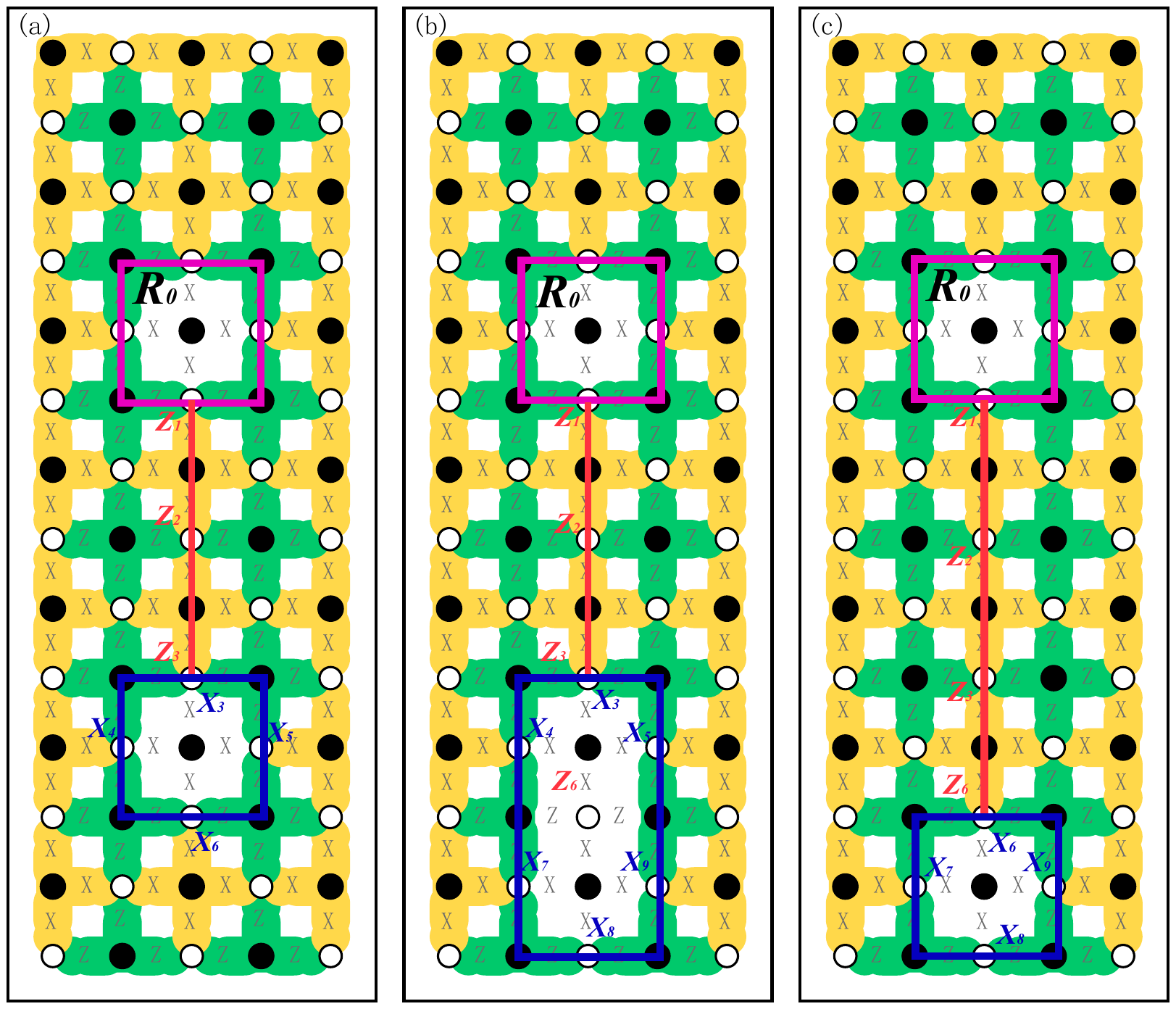}
    \caption{
    Single-cell downward motion of the lower rough ($X$-cut) hole, with the
    upper rough hole $R_0$ fixed. Panel (a) shows the initial code-space
    configuration. The lower hole omits the check
    $A_{\uparrow}=X_3X_4X_5X_6$, whose boundary loop is chosen as
    $\bar X^{(0)}=X_3X_4X_5X_6$, while the connecting logical $Z$-string is
    $\bar Z^{(0)}=Z_1Z_2Z_3$. Panel (b) depicts the transition geometry for
    the two consecutive local updates. First, measuring $Q_1=Z_6$ with
    outcome $\mu$ replaces the active lower check
    $A_{\downarrow}=X_6X_7X_8X_9$ by $\mu Z_6$, thereby enlarging the
    lower hole and updating the logical $X$ representative to
    $\bar X^{(1)}=X_3X_4X_5X_7X_8X_9$, while the algebraic representative
    immediately after this first update remains
    $\bar Z^{(1)}=Z_1Z_2Z_3$. Second, while passing through this enlarged-hole
    transition geometry, measuring
    $Q_2=A_{\uparrow}=X_3X_4X_5X_6$ with outcome $\nu$ replaces
    $\mu Z_6$ by $\nu A_{\uparrow}$ and requires the logical $Z$
    representative to be updated to
    $\bar Z^{(2)}=\mu Z_1Z_2Z_3Z_6$. Panel (c) shows the completed displaced-hole
    configuration, in which the final logical representatives may be chosen
    as $\bar Z^{(2)}=\mu Z_1Z_2Z_3Z_6$ and
    $\bar X^{(2)}=\nu X_6X_7X_8X_9$. The signs $\mu$ and $\nu$ are known
    measurement-record-dependent Pauli-frame signs and are therefore not
    displayed in the operator-support labels of the figure.
    }
    \label{fig:FIG6}
\end{figure}

Consider the initial configuration shown in
Fig.~\ref{fig:FIG6}(a), with code space denoted by
$\mathcal C^{(0)}$. We choose
\begin{equation}
\bar Z^{(0)}
=
Z_1Z_2Z_3,
\qquad
\bar X^{(0)}
=
A_{\uparrow}
=
X_3X_4X_5X_6,
\qquad
A_{\downarrow}
=
X_6X_7X_8X_9.
\label{eq:app-initial-local-representatives}
\end{equation}
Here $A_{\uparrow}$ is the $X$-type check omitted inside the initial rough
hole, whereas $A_{\downarrow}$ is the neighboring active $X$-type
stabilizer immediately below the hole. In the initial configuration, the $X$-type check $A_{\uparrow}$ is switched
off, and its boundary loop is used as the logical $\bar X$ representative.

Moving the lower rough hole downward by one lattice cell replaces the omitted
check $A_{\uparrow}$ by the omitted check $A_{\downarrow}$. The move is
implemented through two elementary measurements. In the local calculation
below, we display only the exchanged stabilizer constraint and the logical
strings relevant to the sensing protocol. Any neighboring stabilizer checks
not displayed in Fig.~\ref{fig:FIG6} may be represented by an
equivalent compatible generating set in each measured branch. This uses only
stabilizer-generator freedom and does not affect the logical-string relations
derived below.

The first measurement enlarges the lower hole. We measure $Q_1=Z_6$ and
take $A_{\downarrow}$ as the old stabilizer constraint replaced in this
elementary update. The relevant commutation relations are
\begin{equation}
\{Q_1,A_{\downarrow}\}=0,
\qquad
\{Q_1,\bar X^{(0)}\}=0,
\qquad
[Q_1,\bar Z^{(0)}]=0.
\label{eq:app-first-measurement-commutation}
\end{equation}
If the recorded outcome is $\mu\in\{+1,-1\}$, then the local stabilizer
constraint is updated as
\begin{equation}
A_{\downarrow}
\longmapsto
\mu Z_6,
\label{eq:app-first-local-update}
\end{equation}
and we denote the resulting intermediate code-space branch by
$\mathcal C_{\mu}^{(1)}$.

Applying Eq.~\eqref{eq:app-transported-representatives} to this first
measurement gives
\begin{equation}
\bar X^{(1)}
:=
A_{\downarrow}\bar X^{(0)}
=
X_3X_4X_5X_7X_8X_9,
\qquad
\bar Z^{(1)}
:=
\bar Z^{(0)}
=
Z_1Z_2Z_3.
\label{eq:app-first-logical-update}
\end{equation}
The logical $X$ representative must be modified because
$\bar X^{(0)}$ anticommutes with $Q_1$, whereas the logical $Z$
representative is unchanged because it already commutes with $Q_1$.
In particular, the first measurement does not require the logical
$Z$-string to acquire the factor $Z_6$.

The enlarged-hole geometry produced by the first measurement is the geometry
shown in panel (b). Since $\mu Z_6$ is already a stabilizer constraint of
$\mathcal C_{\mu}^{(1)}$, the extended string
$\mu Z_1Z_2Z_3Z_6$ would also be a stabilizer-equivalent logical $Z$
representative in this intermediate branch. However, for the purpose of
making the elementary-update logic explicit, we retain
$\bar Z^{(1)}=Z_1Z_2Z_3$ after the first update and introduce the extended
representative only when it becomes required by the second update. This is
why panel (b) may display the extended red-string while the
algebraic representative associated immediately with the first update in
Eq.~\eqref{eq:app-first-logical-update} remains unextended.

The second measurement contracts the upper part of the enlarged hole. We
measure $Q_2=A_{\uparrow}=X_3X_4X_5X_6$ on the intermediate code-space
branch $\mathcal C_{\mu}^{(1)}$. At this stage, the intermediate stabilizer
constraint replaced by the measurement is $\mu Z_6$. Although $Q_2$
coincides as a physical Pauli string with the initial representative
$\bar X^{(0)}$, it is not a logical measurement in the intermediate branch:
it anticommutes with the current stabilizer constraint $\mu Z_6$. The
relations relevant to the logical update are
\begin{equation}
[Q_2,\bar X^{(1)}]=0,
\qquad
\{Q_2,\bar Z^{(1)}\}=0,
\qquad
\{Q_2,\mu Z_6\}=0.
\label{eq:app-second-measurement-commutation}
\end{equation}
The second relation follows because $Q_2$ and $\bar Z^{(1)}$ overlap on
qubit $3$ only. The third relation identifies $\mu Z_6$ as the constraint
replaced in this second elementary update. Thus the anticommutation with
$\bar Z^{(1)}$ does not mean that the logical state is measured; rather, it
means that the logical $Z$ representative must be updated according to the
general rule proved above.

If the recorded outcome of measuring $Q_2$ is
$\nu\in\{+1,-1\}$, then the local stabilizer constraint is updated as
\begin{equation}
\mu Z_6
\longmapsto
\nu A_{\uparrow},
\label{eq:app-second-local-update}
\end{equation}
and we denote the resulting final code-space branch by
$\mathcal C_{\mu,\nu}^{(2)}$. Because $\bar Z^{(1)}$ anticommutes with
$Q_2$, Eq.~\eqref{eq:app-transported-representatives} requires it to be
multiplied by the removed old constraint $\mu Z_6$:
\begin{equation}
\bar Z^{(2)}
:=
\bar Z^{(1)}(\mu Z_6)
=
\mu Z_1Z_2Z_3Z_6.
\label{eq:app-final-z}
\end{equation}
The updated representative $\bar Z^{(2)}$ commutes with the newly imposed
check $\nu A_{\uparrow}$ because the corresponding $Z$-string overlaps with
$A_{\uparrow}$ on qubits $3$ and $6$, yielding an even number of local
$X$--$Z$ anticommutations. Hence the factor $Z_6$ enters the logical
$Z$-string specifically because the second measurement reimposes the upper
$X$-type check.

For the logical $X$ operator, no modification is required by the second
measurement update because $[Q_2,\bar X^{(1)}]=0$.
Consequently, $\bar X^{(1)}$ is already a valid logical $X$ representative
in the final branch. To display the loop surrounding the displaced lower
hole, however, we now make one stabilizer-equivalent representative choice
within $\mathcal C_{\mu,\nu}^{(2)}$. Since $\nu A_{\uparrow}$ is a
stabilizer constraint of this final branch, we choose
\begin{align}
\bar X^{(2)}
&:=
(\nu A_{\uparrow})\bar X^{(1)}
\nonumber\\
&=
\nu
(X_3X_4X_5X_6)
(X_3X_4X_5X_7X_8X_9)
\nonumber\\
&=
\nu X_6X_7X_8X_9.
\label{eq:app-final-x}
\end{align}
With $\bar Y^{(2)}:=i\bar X^{(2)}\bar Z^{(2)}$, the operators
$\bar X^{(2)},\bar Y^{(2)},\bar Z^{(2)}$ provide a geometrically convenient
logical Pauli frame in the final code-space branch. The shortening of the
logical $X$ loop in Eq.~\eqref{eq:app-final-x} changes only its physical
representative, not its encoded logical action.

At the level of the explicitly exchanged stabilizer constraint, the
single-cell motion is summarized by
\begin{equation}
A_{\downarrow}
\overset{Z_6,\mu}{\longmapsto}
\mu Z_6
\overset{A_{\uparrow},\nu}{\longmapsto}
\nu A_{\uparrow}.
\label{eq:app-local-move}
\end{equation}
At the level of the logical representatives, the first measurement updates
the logical $X$ loop while leaving the logical $Z$-string unchanged; the
second measurement necessarily extends the logical $Z$-string to the
displaced hole, after which the logical $X$ loop is shortened using the
newly imposed final stabilizer. By
Eq.~\eqref{eq:app-logical-invariance}, the logical Bloch vector is
preserved through both elementary measurement updates.

The recorded outcomes also determine how a bare physical signal operator in
the final geometry is related to the logical sensing generator. Define
$Z(c_{\rm final}):=Z_1Z_2Z_3Z_6$. From
Eq.~\eqref{eq:app-final-z}, we obtain
\begin{equation}
\bar Z^{(2)}
=
\mu Z(c_{\rm final}),
\qquad
Z(c_{\rm final})
=
\mu\bar Z^{(2)}.
\label{eq:app-frame-sign}
\end{equation}
Thus the same geometrical string operator $Z(c_{\rm final})$ realizes either
$\bar Z^{(2)}$ or $-\bar Z^{(2)}$ in the final logical frame, depending on
the known outcome $\mu$.

Suppose that a subsequent dwell interval is intended to contribute the
logical Hamiltonian term $(\sigma\lambda/2)\bar Z^{(2)}$, where
$\sigma\in\{+1,-1\}$ is the desired logical sensing sign fixed by the
prescribed weighted sum. Since
$Z(c_{\rm final})
=
\mu \bar Z^{(2)}$,
the signed control applied to the bare physical signal operator in this branch
must include the known factor $\mu$. The resulting controlled Hamiltonian is
\begin{equation}
H_{\rm dwell}
=
\frac{\sigma\mu\lambda}{2}
Z(c_{\rm final}).
\end{equation}
Indeed, on the final code space,
\begin{equation}
P_{\mathcal C_{\mu,\nu}^{(2)}}
H_{\rm dwell}
P_{\mathcal C_{\mu,\nu}^{(2)}}
=
\frac{\sigma\lambda}{2}
P_{\mathcal C_{\mu,\nu}^{(2)}}
\bar Z^{(2)}
P_{\mathcal C_{\mu,\nu}^{(2)}}.
\label{eq:app-controlled-dwell-hamiltonian}
\end{equation}
The factor $\mu$ is known once the hole-motion measurement record has been
obtained. Within the signed-control convention used in the main text, its effect on
the subsequent sensing evolution can therefore be compensated
adaptively. In particular, for any edge $e$ in the support of
$c_{\rm final}$,
\begin{equation}
X_e Z(c_{\rm final}) X_e
=
-
Z(c_{\rm final}),
\end{equation}
so the same local-$X$ echo control used to implement a prescribed negative
sensing sign may also absorb the known branch-dependent factor $\mu$.
Accordingly, $\mu$ remains part of the recorded Pauli frame, but its known
effect on the dwell Hamiltonian is absorbed into the control convention and
is not displayed separately in the main text or in Appendix~\ref{app:logical-ramsey-protocol}. The outcome $\nu$
similarly fixes a known Pauli-frame sign for the conjugate logical operator
and is incorporated into subsequent logical control or final readout.

\subsection*{Sequential hole motion and preservation of the accumulated logical phase}

We now apply the elementary-update statement to a full sequential
hole-motion protocol. Let the dwell configurations be indexed by
$r=1,\ldots,L$. Conditioned on all previously recorded local measurement
outcomes, denote the instantaneous code space during the $r$th dwell
interval by $\mathcal C^{(r)}$, and choose logical Pauli representatives
$\bar X^{(r)}$, $\bar Y^{(r)}$, and $\bar Z^{(r)}$ in that branch.

Assume that the deformation from configuration $r$ to configuration $r+1$
is implemented by a finite sequence of elementary stabilizer updates of the
type proved above. At each elementary measurement, the logical Pauli
representatives are updated according to
Eq.~\eqref{eq:app-transported-representatives}. Once the final code-space
branch for a dwell configuration has been reached, a logical representative
may additionally be replaced by a stabilizer-equivalent representative
within that branch when a geometrically simpler $Z$-string or dual $X$-loop is desired,
as in Eq.~\eqref{eq:app-final-x}. Repeated application of
Eq.~\eqref{eq:app-logical-invariance} gives
\begin{align}
\rho_L^{(r)}
&=
P_{\mathcal C^{(r)}}
\frac{
I+\alpha_x\bar X^{(r)}
+\alpha_y\bar Y^{(r)}
+\alpha_z\bar Z^{(r)}
}{2}
P_{\mathcal C^{(r)}}
\nonumber\\
&\longmapsto
P_{\mathcal C^{(r+1)}}
\frac{
I+\alpha_x\bar X^{(r+1)}
+\alpha_y\bar Y^{(r+1)}
+\alpha_z\bar Z^{(r+1)}
}{2}
P_{\mathcal C^{(r+1)}}.
\label{eq:app-sequential-bloch-transport}
\end{align}
Thus a deformation between successive dwell configurations changes the
physical code space and may change the physical supports and known
Pauli-frame signs of the logical representatives, but it does not change the
encoded Bloch vector $(\alpha_x,\alpha_y,\alpha_z)$.

For the Ramsey-type state used in the sensing protocol, the same statement
directly expresses preservation of the accumulated logical phase. If, after a
dwell interval in configuration $r$, the encoded state stores a phase $\phi$
as
\begin{equation}
\rho_{\phi}^{(r)}
=
P_{\mathcal C^{(r)}}
\frac{
I+\cos\phi\,\bar X^{(r)}+\sin\phi\,\bar Y^{(r)}
}{2}
P_{\mathcal C^{(r)}},
\end{equation}
then after the subsequent hole deformation its conditional state is
\begin{equation}
\rho_{\phi}^{(r+1)}
=
P_{\mathcal C^{(r+1)}}
\frac{
I+\cos\phi\,\bar X^{(r+1)}+\sin\phi\,\bar Y^{(r+1)}
}{2}
P_{\mathcal C^{(r+1)}}.
\label{eq:app-sequential-phase-transport}
\end{equation}
Hence the accumulated Ramsey phase $\phi$ is preserved in the updated
logical frame.

This conclusion concerns the conditional code-deformation map generated by
the stabilizer-update measurements. Throughout this work, we assume that
the total duration of all hole-motion and stabilizer-update operations is
negligible compared with the total dwell time over which the sensing phase
is accumulated. We therefore neglect any additional phase contribution from
the unknown sensing Hamiltonian during the deformation operations.

\section{Algebraic Form of the Global Odd-Witness Criterion}
\label{app:chain-cochain-notation-odd-witness}
This appendix gives the algebraic version of the global odd-witness
criterion used in the main text. We keep the notation close to the geometric
language there, but make the chain--cochain formulation explicit.
Let
\begin{equation}
G_{\rm phys}=(V,E,F_{\rm int})
\end{equation}
be the physical planar patch. Its underlying space is homeomorphic to a closed
disk. The edge set $E$ labels data qubits, and $F_{\rm int}$ labels the
physical plaquettes.

For homological analysis, we add the unbounded exterior region as one extra
face $f_\infty$ and set
\begin{equation}
G=(V,E,F),
\qquad
F=F_{\rm int}\cup\{f_\infty\}.
\end{equation}
This operation does not identify opposite boundary points. It does not turn
the planar patch into a torus. It only turns the planar cellulation into a
cellulation of a sphere. Hence
\begin{equation}
H_1(G;\mathbb F_2)=0.
\end{equation}
Unless stated otherwise, all chain and cochain groups below are those of this
extended complex $G$.

All chains and cochains are taken over $\mathbb F_2=\{0,1\}$.
Thus each cell is either absent or present in a chain.
Addition is modulo $2$, so if the same cell appears in two chains, it cancels
in their sum.
We denote this addition by $\oplus$.

We use the cellular chain complex of the extended cellulation $G$.
The chain spaces are
\begin{equation}
C_0(G;\mathbb F_2)\cong \mathbb F_2^V,
\qquad
C_1(G;\mathbb F_2)\cong \mathbb F_2^E,
\qquad
C_2(G;\mathbb F_2)\cong \mathbb F_2^F .
\end{equation}
Here $C_0$ is generated by vertices, $C_1$ by edges, and $C_2$ by faces.
Equivalently, a chain is a $0/1$ vector whose entries indicate which vertices,
edges, or faces are included.
For example, if $E_c\subseteq E$ is a set of edges, then it defines the
$1$-chain
\begin{equation}
c=\bigoplus_{e\in E_c} e
\in C_1(G;\mathbb F_2).
\end{equation}
Conversely, every $1$-chain can be written in this form for a unique subset
$E_c\subseteq E$.

The boundary maps are
\begin{equation}
\partial_1:C_1\to C_0,
\qquad
\partial_2:C_2\to C_1 .
\end{equation}

In matrix form, we write the edge--vertex incidence matrix as
\begin{equation}
D:\mathbb F_2^E\to \mathbb F_2^V,
\qquad
Dc=\partial_1 c,
\end{equation}
and the face--edge incidence matrix as
\begin{equation}
B:\mathbb F_2^F\to \mathbb F_2^E,
\qquad
Bs=\partial_2 s .
\end{equation}
Thus $Dc$ gives the endpoint parity vector of an edge set $c$, and $Bs$
gives the XOR sum of selected plaquette boundaries.

Cochains are linear functionals on chains. Over $\mathbb F_2$, we identify
them with the same vector spaces:
\begin{equation}
C^0(G;\mathbb F_2)\cong \mathbb F_2^V,
\qquad
C^1(G;\mathbb F_2)\cong \mathbb F_2^E,
\qquad
C^2(G;\mathbb F_2)\cong \mathbb F_2^F .
\end{equation}
For example, a $1$-cochain $w\in C^1(G;\mathbb F_2)$ assigns a value
$w_e\in\mathbb F_2$ to each edge $e\in E$. If a $1$-chain is written as
\begin{equation}
c=\bigoplus_{e\in E_c}e,
\end{equation}
where $E_c\subseteq E$, then the pairing between $w$ and $c$ is
\begin{equation}
\langle w,c\rangle
=
\bigoplus_{e\in E_c}w_e .
\end{equation}
Thus $\langle w,c\rangle$ records whether $c$ contains an odd or even number
of edges marked by $w$.
The coboundary operators are the transposes of the boundary matrices:
\begin{equation}
\delta^0=D^\top:C^0\to C^1,
\qquad
\delta^1=B^\top:C^1\to C^2 .
\end{equation}

In the main text, a witness is a simple geometric object. It is represented
by a simple closed dual loop and is used to guide the construction of the
rough holes. This simplicity requirement is part of the main-text construction.

In this appendix, we first introduce a slightly broader algebraic object in
order to express the parity conditions linearly. In this algebraic step, the
symbol $w$ may denote an arbitrary $1$-cochain in $C^1(G;\mathbb F_2)$.
Such a cochain assigns a value in $\mathbb F_2$ to each edge, and it is not
yet required to be a simple closed dual loop.

The relevant algebraic conditions are
$
B^\top w=0
$
and
$
\langle w,c_\alpha\rangle=1
$
for all $\alpha$.
The first condition says that $w$ is closed as a dual $1$-cochain. The second
condition says that $w$ has odd pairing with every signal chain.

This broader use of $w$ is only an intermediate algebraic step. The criterion
below shows that, whenever such a closed algebraic cochain exists, the same
odd-pairing condition can also be realized by a simple witness of the kind used
in the main text.

A $1$-cochain $w\in C^1(G;\mathbb F_2)$ is called closed, or a
$1$-cocycle, if
$
B^\top w=0
$.
Concretely, this means that for every face $f\in F$, the number of boundary
edges $e\in\partial f$ with $w_e=1$ is even. Equivalently,
\begin{equation}
\bigoplus_{e\in\partial f} w_e=0
\qquad
\text{for every } f\in F .
\end{equation}

This is the same as saying that $w$ has no endpoints in the dual picture.
Indeed, a face of $G$ becomes a vertex of the dual graph, and an edge with
$w_e=1$ becomes a selected dual edge crossing $e$. The condition above says
that an even number of selected dual edges meet at every dual vertex. Since
the exterior face $f_\infty$ is included in $F$, this evenness condition is also
imposed at the outer dual vertex.

Thus a closed $1$-cochain gives an even dual subgraph. It need not be one
simple closed loop. It may be a disjoint union of closed dual cycles, or more
generally an even dual subgraph.

On a simply connected patch, every $1$-cocycle is a coboundary. This gives
a binary potential on vertices, or equivalently a global $0/1$ coloring. This
is the basic algebraic tool used below.

\begin{proposition}[Cocycle--potential correspondence]
\label{prop:cocycle-potential-correspondence}
Let $G$ be the extended complex obtained from the physical planar patch by
adding the outer face. Since $G$ is topologically a sphere,
\begin{equation}
H_1(G;\mathbb F_2)=0.
\end{equation}
Let $w\in C^1(G;\mathbb F_2)$ satisfy $B^\top w=0$. Then the following hold.

\noindent(i) There exists a $0$-cochain
$t\in C^0(G;\mathbb F_2)\cong \mathbb F_2^V$ such that
\begin{equation}
w=D^\top t.
\end{equation}
Here a $0$-cochain is a binary function on the vertex set $V$, or equivalently
a $0/1$ coloring of all primal vertices. For every edge $e=(u,v)$, the equation
above means
\begin{equation}
w_e=t(u)\oplus t(v).
\end{equation}
Thus $w_e=1$ exactly when the two endpoints of $e$ have different
$t$-values. Equivalently, $w$ records the edges across which the binary vertex
coloring changes.

\noindent(ii) The potential $t$ is unique up to a global flip. If $t'$ also
satisfies $w=D^\top t'$, then there exists a constant $c\in\{0,1\}$ such that
\begin{equation}
t'(v)=t(v)\oplus c
\qquad
\text{for all }v\in V.
\end{equation}

In other words, every closed $1$-cochain on the extended sphere is the
coboundary of a binary vertex coloring. The coloring is unique up to flipping
all vertex colors.
\end{proposition}
\textbf{Proof.}
We construct the vertex potential $t$ by summing $w$ along paths, and then
show that this construction is independent of the chosen path.

First, we show that $w$ pairs trivially with every closed $1$-chain. Let
$C\in C_1(G;\mathbb F_2)$ satisfy $\partial_1 C=0$. Since
$H_1(G;\mathbb F_2)=0$, there exists $S\in C_2(G;\mathbb F_2)$ such that
\begin{equation}
C=\partial_2 S=BS.
\end{equation}
Using the natural pairing
\begin{equation}
\langle w,C\rangle:=\sum_{e\in E}w_e C_e\in\mathbb F_2,
\end{equation}
we compute
\begin{equation}
\langle w,C\rangle
=
\langle w,BS\rangle
=
\langle B^\top w,S\rangle.
\end{equation}
Because $B^\top w=0$, we get
\begin{equation}
\langle w,C\rangle=0.
\end{equation}
Thus the mod-$2$ sum of $w_e$ along any closed $1$-chain is zero.

Next we define $t$. Choose a reference vertex $v_0\in V$ and set
\begin{equation}
t(v_0):=0.
\end{equation}
For any vertex $v\in V$, choose an edge path from $v_0$ to $v$,
\begin{equation}
\pi:v_0\to v_1\to\cdots\to v_k=v,
\end{equation}
with edges $e_i=(v_{i-1},v_i)$. Define
\begin{equation}
t(v):=\bigoplus_{i=1}^k w_{e_i}.
\end{equation}
Thus $t(v)$ is the parity of the edges in the chosen path on which $w_e=1$.

We now show that this does not depend on the chosen path. Let $\pi_1$ and
$\pi_2$ be two edge paths from $v_0$ to $v$. Viewed as $1$-chains over
$\mathbb F_2$, their sum
\begin{equation}
C:=\pi_1\oplus \pi_2
\end{equation}
is a closed $1$-chain, because endpoints cancel in $\mathbb F_2$. By the
closed-chain result above,
\begin{equation}
\langle w,C\rangle=0.
\end{equation}
This is exactly
\begin{equation}
\bigoplus_{e\in\pi_1}w_e
=
\bigoplus_{e\in\pi_2}w_e.
\end{equation}
Therefore $t(v)$ is path independent. So $t$ is a well-defined $0$-cochain on
$V$.

We now prove $w=D^\top t$. Take any edge $e=(u,v)$. Choose a path $\pi_u$
from $v_0$ to $u$, and extend it by $e$ to get a path $\pi_v$ from $v_0$ to
$v$. By definition,
\begin{equation}
t(v)
=
t(u)\oplus w_e.
\end{equation}
Rearranging gives
\begin{equation}
w_e=t(u)\oplus t(v).
\end{equation}
Since this holds for every edge, we obtain $w=D^\top t$. This proves (i).

Finally we prove (ii). Suppose $t$ and $t'$ both satisfy
$w=D^\top t=D^\top t'$. Define $s:=t\oplus t'$. Then
\begin{equation}
D^\top s
=
D^\top t\oplus D^\top t'
=
w\oplus w
=
0.
\end{equation}
So for every edge $(u,v)$,
\begin{equation}
s(u)\oplus s(v)=0.
\end{equation}
Hence $s(u)=s(v)$ on every adjacent pair. Since the $1$-skeleton
of $G$ is connected, $s$ must be constant on all vertices. Thus there exists
$c\in\{0,1\}$ such that
\begin{equation}
t'(v)=t(v)\oplus c
\qquad
\text{for all }v\in V.
\end{equation}
Thus $t'$ differs from $t$ only by a global flip. This proves part (ii).
\qed

We now apply the cocycle--potential correspondence to the prescribed signal
chains. This step only tests the global parity condition for the signal
chains. It does not yet
construct rough holes.

Let $\mathcal I$ be the set of signal labels. For each $\alpha\in\mathcal I$,
let $c_\alpha\in C_1(G;\mathbb F_2)$ be a simple primal path with endpoints
$u_\alpha$ and $v_\alpha$. Thus
\begin{equation}
D c_\alpha=u_\alpha\oplus v_\alpha .
\end{equation}
Here we identify a vertex with the corresponding basis vector in
$C_0(G;\mathbb F_2)$.
We use \textit{chain} when treating $c_{\alpha}$ as an element of $C_{1}\left(G ; \mathbb{F}_{2}\right)$, and \textit{path} when emphasizing that it is simple and connected; both refer to the same object.

Define the endpoint graph $\Gamma_{\rm end}$ as follows. Its vertex set is
the set of all endpoints that appear among the chains. For each signal chain
$c_\alpha$, we add one abstract edge $\eta_\alpha$ between $u_\alpha$ and
$v_\alpha$. This graph is an abstract multigraph. The edge $\eta_\alpha$
records only the two endpoints of $c_\alpha$. It does not record the full
geometric shape of $c_\alpha$.

We also use the following terminology. A \emph{simple odd witness} is a
simple closed dual loop $w_{\rm s}$ in the discretized dual cellulation. We
always take such a loop to be in dual position, meaning that it crosses primal
edges transversely and does not pass through primal vertices. This is only a
generic-position convention used to make the crossing parity unambiguous.
A small perturbation of the loop can always enforce this convention without
changing its mod-$2$ intersection parity with any primal signal chain.

With this convention, we identify $w_{\rm s}$ with the corresponding
$1$-cochain, also denoted by $w_{\rm s}$, whose value on a primal edge $e$
is $1$ exactly when the loop crosses $e$ an odd number of times. Under this identification, the pairing
\begin{equation}
\langle w_{\rm s},c_\alpha\rangle
\end{equation}
is the mod-$2$ number of intersections between the dual loop $w_{\rm s}$
and the primal chain $c_\alpha$. We call $w_{\rm s}$ odd for the signal family
if
\begin{equation}
\langle w_{\rm s},c_\alpha\rangle=1
\qquad
\text{for all }\alpha\in\mathcal I .
\end{equation}
Equivalently, every signal chain crosses $w_{\rm s}$ an odd number of times.
Since $w_{\rm s}$ is a closed dual loop, the associated $1$-cochain is closed,
so
\begin{equation}
B^\top w_{\rm s}=0.
\end{equation}

For the main-text version of Theorem~\ref{thm:global_odd_witness}, this appendix adds one equivalent algebraic condition, formulated as a closed $1$-cochain with odd pairing against every prescribed signal path.
This condition, written explicitly in Eq.~\eqref{eq:alg-odd-witness}, is one
of the equivalent formulations of the simple-witness criterion used in the
main text. Its role is to express the parity constraints in linear algebraic
form, so that they can be checked through the cocycle--potential correspondence
and the induced coloring of the endpoint graph.

\begin{proposition}[Global odd-parity and simple-witness criterion]
\label{prop:global-odd-parity-simple-witness}
The following four statements are equivalent.

\noindent(i) There exists a $1$-cochain
$w_{\rm alg}\in C^1(G;\mathbb F_2)$ such that
\begin{equation}
B^\top w_{\rm alg}=0,
\qquad
\langle w_{\rm alg},c_\alpha\rangle=1
\quad
\text{for all }\alpha\in\mathcal I .
\label{eq:alg-odd-witness}
\end{equation}
We call such a cochain an algebraic odd witness.

\noindent(ii) The endpoint graph $\Gamma_{\rm end}$ is bipartite.

\noindent(iii) There is no subset $\mathcal A\subseteq\mathcal I$ such that
$|\mathcal A|$ is odd and
\begin{equation}
D\left(\bigoplus_{\alpha\in\mathcal A}c_\alpha\right)=0 .
\label{eq:no-odd-closed-dependency}
\end{equation}

\noindent(iv) There exists a simple odd witness $w_{\rm s}$, in the sense
defined above, such that
\begin{equation}
\langle w_{\rm s},c_\alpha\rangle=1
\quad
\text{for all }\alpha\in\mathcal I .
\label{eq:simple-odd-witness}
\end{equation}
\end{proposition}

\textbf{Proof.}
We first prove the equivalence of (i), (ii), and (iii). Then we prove that
(ii) is equivalent to the existence of a simple odd witness (iv).

First, we prove that (i) implies (ii). Assume that $w_{\rm alg}$ exists.
Since $B^\top w_{\rm alg}=0$, Proposition~\ref{prop:cocycle-potential-correspondence} gives a $0$-cochain
$t\in C^0(G;\mathbb F_2)$ such that
\begin{equation}
w_{\rm alg}=D^\top t .
\end{equation}
For each signal chain $c_\alpha$, we have
\begin{equation}
1
=
\langle w_{\rm alg},c_\alpha\rangle
=
\langle D^\top t,c_\alpha\rangle
=
\langle t,Dc_\alpha\rangle
=
t(u_\alpha)\oplus t(v_\alpha).
\end{equation}
Thus the two endpoints of every endpoint edge $\eta_\alpha$ have opposite
colors under $t$. The restriction of $t$ to the endpoint vertices gives a
binary coloring of $\Gamma_{\rm end}$ in which every endpoint edge connects
opposite colors. Hence, $\Gamma_{\rm end}$ is bipartite.

Next, we prove that (ii) implies (i). Assume that $\Gamma_{\rm end}$ is
bipartite. Choose a binary coloring of its vertices, still denoted by $t$,
such that
\begin{equation}
t(u_\alpha)\oplus t(v_\alpha)=1
\qquad
\text{for all }\alpha\in\mathcal I .
\end{equation}
Extend $t$ arbitrarily to all vertices of $G$. Define
\begin{equation}
w_{\rm alg}:=D^\top t .
\end{equation}
Then
\begin{equation}
B^\top w_{\rm alg}
=
B^\top D^\top t
=
(DB)^\top t
=
0,
\end{equation}
because $DB=0$. Also, for each $\alpha$,
\begin{equation}
\langle w_{\rm alg},c_\alpha\rangle
=
\langle D^\top t,c_\alpha\rangle
=
\langle t,Dc_\alpha\rangle
=
t(u_\alpha)\oplus t(v_\alpha)
=
1.
\end{equation}
Thus $w_{\rm alg}$ satisfies statement (i). This proves the
equivalence of (i) and (ii).

We now prove the equivalence of (ii) and (iii). Let
\begin{equation}
\partial_{\rm end}:C_1(\Gamma_{\rm end};\mathbb F_2)
\to
C_0(\Gamma_{\rm end};\mathbb F_2)
\end{equation}
be the boundary map of the endpoint graph, defined on each endpoint edge by
\begin{equation}
\partial_{\rm end}\eta_\alpha
=
u_\alpha\oplus v_\alpha .
\end{equation}
Since $u_\alpha$ and $v_\alpha$ are the endpoints of the physical signal chain
$c_\alpha$, we have
\begin{equation}
\partial_{\rm end}\eta_\alpha
=
D c_\alpha .
\end{equation}
Therefore, for every subset $\mathcal A\subseteq\mathcal I$,
\begin{equation}
\partial_{\rm end}
\left(
\bigoplus_{\alpha\in\mathcal A}\eta_\alpha
\right)
=
D
\left(
\bigoplus_{\alpha\in\mathcal A}c_\alpha
\right).
\label{eq:end-boundary-physical-boundary}
\end{equation}
Thus
\begin{equation}
D
\left(
\bigoplus_{\alpha\in\mathcal A}c_\alpha
\right)
=0
\end{equation}
if and only if the selected endpoint edges
$\{\eta_\alpha:\alpha\in\mathcal A\}$ form a closed $1$-chain in
$\Gamma_{\rm end}$ over $\mathbb F_2$.

We now prove that (ii) implies (iii). Assume that $\Gamma_{\rm end}$ is
bipartite. Let its two vertex classes be $V_0$ and $V_1$. Then every endpoint
edge has one endpoint in $V_0$ and one endpoint in $V_1$.

Suppose that $\mathcal A\subseteq\mathcal I$ satisfies
\begin{equation}
D
\left(
\bigoplus_{\alpha\in\mathcal A}c_\alpha
\right)
=0 .
\end{equation}
By Eq.~\eqref{eq:end-boundary-physical-boundary}, the corresponding selected
endpoint edges form a closed $1$-chain in $\Gamma_{\rm end}$. Equivalently,
every endpoint vertex is incident to an even number of selected endpoint
edges.

Let $V_{\rm end}$ denote the vertex set of $\Gamma_{\rm end}$. For a vertex $x\in V_{\rm end}$, let $\deg_{\mathcal A}(x)$ be the number of
selected endpoint edges incident to $x$. The equation above says that every
vertex has even selected degree:
\begin{equation}
\deg_{\mathcal A}(x)\equiv 0 \pmod 2
\qquad
\text{for every }x\in V_{\rm end}.
\end{equation}

Now count the selected endpoint edges from the $V_0$ side. Since every
selected endpoint edge has exactly one endpoint in $V_0$, we have the integer
equality
\begin{equation}
|\mathcal A|
=
\sum_{x\in V_0}\deg_{\mathcal A}(x).
\end{equation}
Each term on the right-hand side is even. Therefore $|\mathcal A|$ is even.
Hence no odd-cardinality subset $\mathcal A$ can satisfy
Eq.~\eqref{eq:no-odd-closed-dependency}. This proves (iii).

Conversely, we prove that (iii) implies (ii) by contrapositive. Assume
that $\Gamma_{\rm end}$ is not bipartite. Then $\Gamma_{\rm end}$ contains an
odd cycle. Let the endpoint edges on this odd cycle be
\begin{equation}
\eta_{\alpha_1},\eta_{\alpha_2},\ldots,\eta_{\alpha_{2k+1}} .
\end{equation}
Set
\begin{equation}
\mathcal A
=
\{\alpha_1,\alpha_2,\ldots,\alpha_{2k+1}\}.
\end{equation}
Then $|\mathcal A|=2k+1$, so $\mathcal A$ has odd cardinality.

The selected endpoint edges form a cycle. Hence every vertex on this cycle
has selected degree $2$, and every other vertex has selected degree $0$.
Therefore
\begin{equation}
\partial_{\rm end}
\left(
\bigoplus_{\alpha\in\mathcal A}\eta_\alpha
\right)
=0 .
\end{equation}
Using Eq.~\eqref{eq:end-boundary-physical-boundary}, we obtain
\begin{equation}
D
\left(
\bigoplus_{\alpha\in\mathcal A}c_\alpha
\right)
=0 .
\end{equation}
Thus there exists an odd-cardinality subset $\mathcal A$ whose physical
boundary cancels. This violates (iii). Therefore, if (iii) holds,
$\Gamma_{\rm end}$ must be bipartite.

This proves the equivalence of (ii) and (iii).

It remains to relate the bipartition to a simple witness. We first prove that
(ii) implies (iv). Assume that $\Gamma_{\rm end}$ is bipartite. Choose a
bipartition of its endpoint vertices,
\begin{equation}
\mathsf U_0\sqcup \mathsf U_1 .
\end{equation}
Thus each signal chain has one endpoint in $\mathsf U_0$ and one endpoint in
$\mathsf U_1$. Write
\begin{equation}
D c_\alpha
=
u_\alpha^{(0)}\oplus u_\alpha^{(1)},
\qquad
u_\alpha^{(0)}\in\mathsf U_0,
\qquad
u_\alpha^{(1)}\in\mathsf U_1 .
\end{equation}

We work on the sphere obtained by adding the outer face. Choose an embedded
tree $T_0$ on the sphere such that
\begin{equation}
\mathsf U_0\subset T_0,
\qquad
T_0\cap \mathsf U_1=\varnothing .
\end{equation}
This is possible because $\mathsf U_1$ is finite, so the punctured sphere
$S^2\setminus \mathsf U_1$ is path connected. Hence the points in
$\mathsf U_0$ can be connected by embedded arcs that avoid $\mathsf U_1$.
Taking finitely many such arcs gives a connected embedded graph containing
$\mathsf U_0$ and disjoint from $\mathsf U_1$. If this graph contains cycles,
we delete redundant arcs and keep a spanning tree.

Let $N(T_0)$ be a small closed regular neighborhood of $T_0$, chosen so that
\begin{equation}
\mathsf U_0\subset \operatorname{int}N(T_0),
\qquad
\mathsf U_1\cap N(T_0)=\varnothing .
\end{equation}
Since $T_0$ is a tree, $N(T_0)$ is a disk. Hence its boundary
\begin{equation}
\gamma:=\partial N(T_0)
\end{equation}
is a simple closed curve separating $\mathsf U_0$ from $\mathsf U_1$. We choose $\gamma$ in dual position with respect to the primal cellulation.
This means that $\gamma$ crosses primal edges transversely, does not pass
through primal vertices, and does not overlap with any primal edge. This is
only a generic-position convention, and it can be achieved by a small
perturbation without changing which endpoint set lies on which side of
$\gamma$.
After this harmless local deformation, we regard $\gamma$ as a simple closed loop on the dual lattice.
The deformation does not change its mod-$2$ intersection parity with any primal chain.

Let $w_{\rm s}\in C^1(G;\mathbb F_2)$ be the crossing cochain associated with
$\gamma$. Thus $(w_{\rm s})_e=1$ exactly when $\gamma$ crosses the primal
edge $e$ an odd number of times.

Now fix any signal path $c_\alpha$. Its endpoint
$u_\alpha^{(0)}$ lies inside $N(T_0)$, while its endpoint
$u_\alpha^{(1)}$ lies outside $N(T_0)$. Hence the two endpoints lie on
different sides of the simple closed curve $\gamma$. As the path
$c_\alpha$ moves from one side of $\gamma$ to the other, it must cross
$\gamma$ an odd number of times. Therefore
\begin{equation}
\langle w_{\rm s},c_\alpha\rangle=1
\qquad
\text{for all }\alpha\in\mathcal I .
\end{equation}
Thus a simple odd witness exists. This proves (iv).

Finally, (iv) implies (i). Let $w_{\rm s}$ be a simple odd witness, and
identify it with its crossing cochain. Since $w_{\rm s}$ is represented by a
closed curve in dual position, it has no dual endpoints. Equivalently, each
face boundary is crossed an even number of times, including the outer face.
Thus
\begin{equation}
B^\top w_{\rm s}=0.
\end{equation}
Together with
\begin{equation}
\langle w_{\rm s},c_\alpha\rangle=1
\qquad
\text{for all }\alpha\in\mathcal I ,
\end{equation}
this is exactly statement (i), with $w_{\rm alg}=w_{\rm s}$. Since we have already proved that (i) implies (ii), a simple odd witness also
forces the endpoint graph $\Gamma_{\rm end}$ to be bipartite.

All four statements are equivalent.
\qed

The proposition separates two ideas. The algebraic object in statement (i) is
a closed dual parity test. Its support may be a union of closed dual cycles.
It need not be one simple loop.

The planar geometry gives more. Statement (iv) says that, once the endpoint
graph is bipartite, we can construct a simple odd witness. We do not obtain
this simple loop by choosing one connected component of the algebraic
$w_{\rm alg}$. Instead, we use the endpoint bipartition. We connect one
endpoint class by a tree, thicken the tree to a disk, and take the boundary
of that disk. This boundary gives a simple odd witness in the sense defined above. Since every signal
path has endpoints on opposite sides, every signal path crosses this loop
oddly.

The obstruction in the global criterion is not the existence of a closed
$1$-cycle in the union of the signal chains. The obstruction is an odd closed
dependency among the signal chains. Thus an even closed dependency is not itself an obstruction. For example, four chains may satisfy
\begin{equation}
D(c_1\oplus c_2\oplus c_3\oplus c_4)=0,
\end{equation}
and still admit a simple odd witness, because
\begin{equation}
1\oplus 1\oplus 1\oplus 1=0.
\end{equation}
In contrast, three chains satisfying
\begin{equation}
D(c_1\oplus c_2\oplus c_3)=0
\end{equation}
cannot all pair oddly with a closed witness, because
\begin{equation}
1\oplus 1\oplus 1=1.
\end{equation}
This distinction explains the role of Fig.~\ref{fig:caveat} in the main text. The configuration shown in that figure
may satisfy the global odd-witness criterion even though the signal chains
form a closed cycle. Its possible failure is local. It belongs to the clean
rough-hole realization step, not to the global witness criterion.

We now give two simple global situations in which the odd-witness part of
the criterion is guaranteed.

\begin{corollary}[Useful sufficient cases]
\label{cor:useful-sufficient-cases}
The following two sufficient conditions imply that a simple odd witness exists.

\noindent(i) If the endpoint graph $\Gamma_{\rm end}$ is a forest, then a
simple odd witness exists.

\noindent(ii) Assume the signal chains are edge-disjoint. If the physical
union graph
\begin{equation}
\Gamma_{\rm sig}:=\bigcup_{\alpha\in\mathcal I}c_\alpha
\end{equation}
has no cycle, then a simple odd witness exists.
\end{corollary}
\textbf{Proof.}
For (i), a forest is bipartite. Hence Proposition~\ref{prop:global-odd-parity-simple-witness} gives a simple odd
witness.

For (ii), it is enough to show that $\Gamma_{\rm end}$ is a forest. Suppose
not. Then $\Gamma_{\rm end}$ contains a cycle with endpoint edges
\begin{equation}
\eta_{\alpha_1},\eta_{\alpha_2},\ldots,\eta_{\alpha_m}.
\end{equation}
Let
\begin{equation}
\mathcal A=\{\alpha_1,\alpha_2,\ldots,\alpha_m\}.
\end{equation}
Since these endpoint edges form a cycle,
\begin{equation}
\partial_{\rm end}
\left(
\bigoplus_{\alpha\in\mathcal A}\eta_\alpha
\right)
=0.
\end{equation}
Using the relation
\begin{equation}
\partial_{\rm end}\eta_\alpha=Dc_\alpha,
\end{equation}
we get
\begin{equation}
D
\left(
\bigoplus_{\alpha\in\mathcal A}c_\alpha
\right)
=0.
\end{equation}
Define
\begin{equation}
z_{\rm sig}:=
\bigoplus_{\alpha\in\mathcal A}c_\alpha .
\end{equation}
Then $Dz_{\rm sig}=0$. Also $z_{\rm sig}\neq 0$, because the selected signal chains have pairwise
disjoint edge supports. Hence no edge appearing in one selected chain can be
cancelled by an edge from another selected chain.

Thus $z_{\rm sig}$ is a nonzero closed $1$-chain in the finite graph
$\Gamma_{\rm sig}$. The support of any nonzero closed $1$-chain in a finite
graph contains an ordinary graph cycle. Hence $\Gamma_{\rm sig}$ contains a
cycle. This contradicts the assumption that $\Gamma_{\rm sig}$ has no cycle.

Therefore $\Gamma_{\rm end}$ is a forest. By (i), a simple odd witness exists.
\qed

The edge-disjointness assumption in
Corollary~\ref{cor:useful-sufficient-cases}(ii) is used only for that
sufficient condition. It prevents mod-$2$ cancellation of shared physical
edges when selected signal chains are summed. It is not part of the global
odd-witness criterion in
Proposition~\ref{prop:global-odd-parity-simple-witness}.

\section{Constructing rough holes from simple witnesses}
\label{app:rough-hole-construction-simple-witnesses}
This appendix gives the local geometric construction that converts a simple
odd witness into clean rough holes. The global odd-witness criterion itself was
proved in Appendix~\ref{app:chain-cochain-notation-odd-witness}; here we
explain the additional local step needed to realize the corresponding punctured
surface code.

By Proposition~\ref{prop:global-odd-parity-simple-witness}, in the globally compatible case there exists a simple
odd witness $w^{(1)}$. Thus $w^{(1)}$ is a simple closed dual loop and
\begin{equation}
\langle w^{(1)},c_\alpha\rangle=1
\qquad
\text{for all }\alpha\in\mathcal I .
\label{eq:first-simple-witness-condition}
\end{equation}

Equivalently, $w^{(1)}$ separates the sphere into two connected sides. Since
$c_\alpha$ crosses $w^{(1)}$ oddly, the two endpoints of $c_\alpha$ lie on
opposite sides of $w^{(1)}$. We choose one side and call it $\Omega_0$. In
practice, we usually choose the side from which a clean rough hole can be
implemented with fewer turned-off star checks. This choice is only an
economical convention and has no additional topological significance.

Throughout Appendices~\ref{app:rough-hole-construction-simple-witnesses} and~\ref{app:relative-homology-fixed-rough-holes}, a witness-guided region $\Omega_i$ or an
actual rough-hole region $R_i$ may be described either by the corresponding
set of star-check vertices or by its disk-like cellular carrier in the planar
complex. We use the same symbol for these two associated descriptions. The
star-check-vertex description is used when discussing which checks are turned
off and how the prescribed signal chains contact a region, whereas the
cellular-carrier description is used when discussing boundaries, quotients,
and relative homology.
Let $U_0$ be the set of endpoints of the signal chains lying in $\Omega_0$:
\begin{equation}
U_0
:=
\{u_\alpha^{(0)}:\alpha\in\mathcal I\}.
\end{equation}
If several signal chains share the same endpoint vertex in $\Omega_0$, that
vertex appears only once in $U_0$; the different chains are still distinguished
by their labels $\alpha$.
For each $\alpha$, the other endpoint of $c_\alpha$ is denoted by
$u_\alpha^{(1)}$. Thus
\begin{equation}
D c_\alpha
=
u_\alpha^{(0)}\oplus u_\alpha^{(1)},
\qquad
u_\alpha^{(0)}\in \Omega_0,
\qquad
u_\alpha^{(1)}\notin \Omega_0 .
\label{eq:first-endpoint-sides}
\end{equation}

The region $\Omega_0$ is only a guide. It is not itself the rough hole. No
star check is turned off merely because it lies in $\Omega_0$.

The actual first rough hole is a simply connected disk-like star-check region
\begin{equation}
R_0\subseteq \Omega_0 .
\end{equation}
As in the main text, opening $R_{0}$ also requires the corresponding removal of interior $Z$-type plaquette checks and truncation of the affected plaquette checks along its boundary; we do not repeat this remark when $ R_{1} $ is opened below.
It must contain the endpoint set $U_0$:
\begin{equation}
U_0\subseteq R_0 .
\end{equation}
It must also satisfy the clean condition. Let $\operatorname{Vert}(c_\alpha)$
denote the set of star-check vertices touched by the path $c_\alpha$. For the
first rough hole, the only allowed contact of $c_\alpha$ with $R_0$ is the
assigned endpoint $u_\alpha^{(0)}$:
\begin{equation}
R_0\cap \operatorname{Vert}(c_\alpha)
\subseteq
\{u_\alpha^{(0)}\}
\qquad
\text{for all }\alpha\in\mathcal I .
\label{eq:R0-clean-condition}
\end{equation}
Equivalently, if $v$ is a non-terminal vertex of any prescribed signal chain,
then the star check at $v$ is not included in $R_0$ and is therefore kept active.
In the simple-path case considered here, a non-terminal vertex is a star-check
vertex where two edges of the same path meet, so the path passes through that
vertex rather than terminating there.

This is a local geometric condition. It is not part of the global
odd-witness criterion. The choice of $R_0$ is not canonical; it is a local
search inside $\Omega_0$ for a simply connected disk-like star-check region that
contains the assigned endpoints and avoids all forbidden non-terminal
vertices.

Equivalently, one may search for $R_0$ in two complementary ways. One may
connect the endpoint vertices in $U_0$ inside $\Omega_0$ by allowed
star-check paths and then thicken the resulting connected set. Or one may
start from all of $\Omega_0$, remove the non-terminal vertices forbidden by
the clean rule, and check whether the remaining allowed region has a
simply connected disk-like component containing $U_0$. If no such choice exists,
then the simple odd witness remains valid, but this particular
witness-guided side fails only at the local clean-realization step.

When $R_0$ is chosen, we turn off exactly the selected $X$-type star checks
in $R_0$. This creates the first rough boundary. After this operation, each
signal chain $c_\alpha$ has one terminal contact with the rough boundary
created by $R_0$, while all its non-terminal star-check vertices remain in the
active code region. In this sense, each $c_\alpha$ terminates on the first
rough boundary.

For the auxiliary quotient construction, we now use the disk-like cellular
carrier associated with the actual rough-hole region $R_0$. Define
\begin{equation}
\widetilde G_0:=G/R_0
\end{equation}
to be the quotient complex obtained by collapsing $R_0$ to one vertex. We
denote the collapsed vertex by $\tilde r_0$. Thus all vertices contained in
$R_0$, including the endpoints $u_\alpha^{(0)}$, are identified with the same
quotient vertex $\tilde r_0$. Let
\begin{equation}
\pi_0:G\to \widetilde G_0
\end{equation}
be the quotient map.

Because $R_0$ is contractible and disk-like, collapsing $R_0$ does not create
a handle. The quotient $\widetilde G_0$ is again a planar complex on the sphere.
In particular,
\begin{equation}
H_1(\widetilde G_0;\mathbb F_2)=0 .
\label{eq:quotient-h1-zero}
\end{equation}

Let
\begin{equation}
\widetilde D_0:C_1(\widetilde G_0;\mathbb F_2)\to C_0(\widetilde G_0;\mathbb F_2),
\qquad
\widetilde B_0:C_2(\widetilde G_0;\mathbb F_2)\to C_1(\widetilde G_0;\mathbb F_2)
\end{equation}
be the incidence matrices of $\widetilde G_0$.

The quotient map $\pi_0:G\to \widetilde G_0$ induces a cellular chain map.
We use the same notation $(\pi_0)_\#$ for the induced maps in all degrees. In
degree one, this gives
\begin{equation}
(\pi_0)_\#:C_1(G;\mathbb F_2)\to C_1(\widetilde G_0;\mathbb F_2).
\end{equation}
On a basis edge $e$, the map sends $e$ to zero if $e$ is contained in the
collapsed subcomplex $R_0$, because such an edge is contracted to a point. If
$e$ is not contained in $R_0$, then $e$ is sent to its image edge in the quotient
complex, with any endpoint in $R_0$ identified with $\tilde r_0$. The map is
extended to an arbitrary $1$-chain by applying it edge by edge and adding the
resulting image edges over $\mathbb F_2$.

Each signal chain $c_\alpha$ induces the quotient chain
\begin{equation}
\tilde c_\alpha:=(\pi_0)_\# c_\alpha
\in C_1(\widetilde G_0;\mathbb F_2).
\end{equation}
By the clean condition for $R_0$, the chain $c_\alpha$ meets $R_0$ only at
its assigned endpoint $u_\alpha^{(0)}$. Hence no non-terminal edge or vertex
of $c_\alpha$ is collapsed, and $\tilde c_\alpha$ is still a simple primal path
after the quotient. The endpoint $u_\alpha^{(0)}\in R_0$ becomes
$\tilde r_0$, while the other endpoint $u_\alpha^{(1)}$ remains outside
$R_0$. We write
\begin{equation}
\tilde u_\alpha:=\pi_0(u_\alpha^{(1)}) .
\end{equation}
Thus
\begin{equation}
\widetilde D_0\tilde c_\alpha
=
\tilde r_0\oplus \tilde u_\alpha .
\label{eq:quotient-chain-boundary}
\end{equation}
Let
\begin{equation}
\widetilde U_1
:=
\{\tilde u_\alpha:\alpha\in\mathcal I\}
\subseteq V(\widetilde G_0).
\end{equation}

We now construct the second witness on $\widetilde G_0$. The quotient chains
$\tilde c_\alpha$ all have one endpoint at $\tilde r_0$ and the other endpoint
in $\widetilde U_1$. Therefore the endpoint multigraph of the quotient chains
is bipartite, with bipartition
\begin{equation}
\{\tilde r_0\}\sqcup \widetilde U_1 .
\end{equation}
Hence Proposition~\ref{prop:global-odd-parity-simple-witness} applies on the quotient complex $\widetilde G_0$.

Thus there exists a simple odd witness $w^{(2)}$ on $\widetilde G_0$ such that
\begin{equation}
\langle w^{(2)},\tilde c_\alpha\rangle=1
\qquad
\text{for all }\alpha\in\mathcal I .
\label{eq:second-simple-witness-condition}
\end{equation}
Since each $\tilde c_\alpha$ crosses $w^{(2)}$ oddly, the vertex $\tilde r_0$
and every vertex in $\widetilde U_1$ lie on opposite sides of $w^{(2)}$.
Let $\widetilde\Omega_1$ be the side of $w^{(2)}$ that contains
$\widetilde U_1$. Then
\begin{equation}
\widetilde U_1\subseteq \widetilde\Omega_1,
\qquad
\tilde r_0\notin \widetilde\Omega_1 .
\label{eq:second-region-quotient}
\end{equation}

Since $w^{(2)}$ is a simple closed dual loop in dual position, it does not
pass through the collapsed primal vertex $\tilde r_0$. The quotient map
$\pi_0$ changes the geometry only inside the disk-like rough-hole region
$R_0$, which is represented in $\widetilde G_0$ by the single vertex
$\tilde r_0$. Therefore, after restoring the collapsed region, the loop
$w^{(2)}$ is naturally identified with a simple closed dual loop in the
original patch that is disjoint from $R_0$. This is the meaning of lifting
$w^{(2)}$ back to the original patch in the geometric construction described
in the main text. We use the notation $w^{(2)}$ for this corresponding
witness loop in either geometry.

The side of $w^{(2)}$ containing the remaining endpoint set can accordingly
be pulled back to the original patch. We define the second
witness-guided star-check region by
\begin{equation}
\Omega_1
:=
\pi_0^{-1}(\widetilde\Omega_1)\cap V(G).
\label{eq:omega1-definition}
\end{equation}
Since $\tilde r_0\notin\widetilde\Omega_1$, the region $\Omega_1$ is disjoint
from $R_0$. Also, by construction,
\begin{equation}
U_1:=\{u_\alpha^{(1)}:\alpha\in\mathcal I\}
\subseteq \Omega_1 .
\end{equation}

Again, $\Omega_1$ is only a witness-guided region. The actual second rough hole is a simply connected disk-like star-check region
\begin{equation}
R_1\subseteq \Omega_1
\end{equation}
such that
\begin{equation}
U_1\subseteq R_1 .
\end{equation}
It must satisfy the same clean condition as $R_0$. For the second rough hole, the only allowed contact of
$c_\alpha$ with $R_1$ is the assigned endpoint $u_\alpha^{(1)}$:
\begin{equation}
R_1\cap \operatorname{Vert}(c_\alpha)
\subseteq
\{u_\alpha^{(1)}\}
\qquad
\text{for all }\alpha\in\mathcal I .
\label{eq:R1-clean-condition}
\end{equation}
Equivalently, if $v$ is a non-terminal vertex of any prescribed signal chain
$c_\alpha$, then the star check at $v$ is not included in $R_1$ and is therefore
kept active. Thus a star-check site may lie in the witness-guided region
$\Omega_1$, but still not belong to the actual rough hole $R_1$.

After $R_1$ is chosen, we turn off exactly the selected $X$-type star checks
in $R_1$. If two disjoint clean simply connected disk-like rough holes
\begin{equation}
R_0\subseteq\Omega_0,
\qquad
R_1\subseteq\Omega_1
\end{equation}
can be chosen in the interior of the physical patch, then every prescribed
signal chain $c_\alpha$ has one terminal contact with each rough hole, while
all of its non-terminal star-check vertices remain in the active code region.
This is the object referred to in the main text as a legal $Z$-type path.
In the appendices, where the two rough holes are fixed and relative homology
is used, we call such a path a legal \textit{bridge} from $R_0$ to $R_1$.

Thus every prescribed $c_\alpha$ is a legal bridge from $R_0$ to
$R_1$. Appendix~\ref{app:relative-homology-fixed-rough-holes} then fixes these
two clean rough holes and proves that all legal bridges between them represent
the same nonzero relative homology class. It follows that all prescribed signal
operators $Z(c_\alpha)$ act as the same nontrivial logical operator $\bar Z$ on
the punctured surface-code space.

\section{Relative homology for fixed clean rough holes}
\label{app:relative-homology-fixed-rough-holes}
This appendix proves the fixed-hole homological statement used in the main
text. The input is a pair of disjoint clean actual rough holes
$R_0$ and $R_1$ in the interior of the planar patch, obtained from the
construction in
Appendix~\ref{app:rough-hole-construction-simple-witnesses}. From this point
on, the witness-guided regions $\Omega_i$ no longer enter the argument. The fixed-hole analysis uses only the terminal-contact condition supplied by
the construction. A legal primal $Z$-chain may contact each $R_i$ only at
its terminal endpoint on that rough boundary, while all non-terminal
star-check vertices remain in the active code region. When such a chain has
one terminal endpoint on $R_0$ and the other on $R_1$, we call it a bridge
from $R_0$ to $R_1$.

We show that every bridge represents the same nonzero relative homology
class. Consequently, all bridge operators act as the same nontrivial logical
operator $\bar Z$ on the punctured-code space. We also include the proof of
Proposition~\ref{prop:clean_puncture_realization} from the main text.

We use the extended planar complex $G=(V,E,F)$ and the face--edge incidence
map
\begin{equation}
B:\mathbb F_2^{F}\to \mathbb F_2^{E}
\end{equation}
introduced in Appendix~\ref{app:chain-cochain-notation-odd-witness}. A
$Z$-string support is a $1$-chain $c\in C_1(G;\mathbb F_2)$, and
$Z(c)$ denotes the product of $Z$ operators on the corresponding edge
qubits. Since $\operatorname{im}B$ is the plaquette-boundary span, two
chains differing by an element of $\operatorname{im}B$ give operators that
act identically on the code space.

For homological analysis, we collapse each actual rough hole $R_i$ to a
single marked vertex. Equivalently, we quotient $G$ by the relation that
identifies all cells in $R_i$, separately for $i=0,1$. We write the
resulting quotient complex as
\begin{equation}
\widehat G:=G/(R_0\cup R_1).
\end{equation}
The two collapsed vertices are denoted by
$\hat r_0$
and
$\hat r_1$.
We also set
\begin{equation}
A:=\{\hat r_0,\hat r_1\}.
\end{equation}

This quotient is only an analysis device. Physically, a signal chain
terminates on the rough boundary of $R_i$. In the quotient complex, all
allowed endpoints on that rough boundary are identified with the single
marked vertex $\hat r_i$. Thus $A$ is the set of allowed endpoints for
legal $Z$-chains.

Let
\begin{equation}
\pi:G\to \widehat G
\end{equation}
be the quotient map. It induces a linear map on $1$-chains,
\begin{equation}
\pi_\#:C_1(G;\mathbb F_2)\to C_1(\widehat G;\mathbb F_2).
\end{equation}
On a basis edge $e$, the map $\pi_\#$ sends $e$ to zero if $e$ is
collapsed to a point, and otherwise sends $e$ to its image edge in
$\widehat G$. The map is then extended linearly over $\mathbb F_2$.
For a physical $Z$-chain $c$, we write
\begin{equation}
\hat c:=\pi_\#(c)\in C_1(\widehat G;\mathbb F_2).
\end{equation}
For the clean chains considered below, no non-terminal part of the chain is
collapsed; only the allowed endpoint contacts are identified with
$\hat r_0$ or $\hat r_1$.

Let
\begin{equation}
\widehat D:C_1(\widehat G;\mathbb F_2)\to C_0(\widehat G;\mathbb F_2)
\end{equation}
be the edge--vertex incidence map of $\widehat G$. Thus
$\widehat D\hat c$ records the endpoints of $\hat c$, counted modulo
$2$.

Let
\begin{equation}
\widehat B:
C_2(\widehat G;\mathbb F_2)
\to
C_1(\widehat G;\mathbb F_2)
\end{equation}
be the face--edge incidence map of the quotient complex $\widehat G$.
Equivalently, $\widehat B$ is the cellular boundary map $\partial_2$ on
$\widehat G$.

A legal $Z$-chain is allowed to have endpoints only on the rough
boundaries. After collapsing each rough hole to a marked vertex, this means
that the boundary of the quotient chain $\hat c$ is supported only on
$A=\{\hat r_0,\hat r_1\}$. Equivalently,
\begin{equation}
\widehat D\hat c\in C_0(A;\mathbb F_2).
\end{equation}
Thus $\hat c$ need not be an ordinary $1$-cycle in $\widehat G$, because its
boundary may be nonzero. However, it is a relative $1$-cycle for the pair
$(\widehat G,A)$, because its boundary becomes zero after quotienting out the
$0$-chains supported on $A$.

Since $A$ consists only of $0$-cells, $C_1(A;\mathbb F_2)=0$. Hence a
relative $1$-chain can be represented directly by an ordinary $1$-chain in
$\widehat G$. The only difference is the cycle condition: ordinary homology
requires $\widehat D\hat c=0$, whereas relative homology only requires
$\widehat D\hat c\in C_0(A;\mathbb F_2)$. This is why relative homology is
the natural language for legal open $Z$-strings.

For a legal bridge $c$ from $R_0$ to $R_1$, in the sense introduced in
Appendix~\ref{app:rough-hole-construction-simple-witnesses}, the quotient
chain $\hat c$ has one endpoint at each marked rough vertex. Hence
\begin{equation}
\widehat D\hat c
=
\hat r_0\oplus \hat r_1 .
\label{eq:bridge-boundary-app}
\end{equation}
This is the boundary condition for the bridge class in
$H_1(\widehat G,A;\mathbb F_2)$.

We now compute the relative group $H_1(\widehat G,A;\mathbb F_2)$. We use
the long exact sequence of the pair $(\widehat G,A)$~\cite{Hatcher2001AlgebraicTopology}. For any pair $A\subset X$, the
relative chain group is defined by
\begin{equation}
C_k(X,A;\mathbb F_2)
:=
C_k(X;\mathbb F_2)/C_k(A;\mathbb F_2).
\end{equation}
This gives a short exact sequence of chain complexes,
\begin{equation}
0
\to
C_*(A;\mathbb F_2)
\to
C_*(X;\mathbb F_2)
\to
C_*(X,A;\mathbb F_2)
\to
0 .
\end{equation}
Taking homology gives the long exact sequence
\begin{equation}
\cdots
\to
H_1(A)
\to
H_1(X)
\to
H_1(X,A)
\xrightarrow{\partial_{\rm rel}}
H_0(A)
\xrightarrow{i_*}
H_0(X)
\to
H_0(X,A)
\to
0 .
\end{equation}
We apply this to
\begin{equation}
X=\widehat G,
\qquad
A=\{\hat r_0,\hat r_1\}.
\end{equation}

The space $A$ is a discrete space with two contractible components. Hence
\begin{equation}
H_1(A;\mathbb F_2)=0.
\end{equation}
Also, by construction, each actual rough hole is a simply connected disk-like
subcomplex in the interior of the disk-like patch. After the outer face is
added, $G$ is a cellulation of $S^2$. Collapsing $R_0$ and $R_1$ to points
only shrinks two contractible regions; it does not create a handle. Thus
$\widehat G$ is homotopy equivalent to $S^2$, and hence
\begin{equation}
H_1(\widehat G;\mathbb F_2)=0.
\end{equation}
Finally, $\widehat G$ is connected and $A$ is nonempty. Therefore the map
\begin{equation}
i_*:H_0(A;\mathbb F_2)\to H_0(\widehat G;\mathbb F_2)
\end{equation}
is onto, and hence
\begin{equation}
H_0(\widehat G,A;\mathbb F_2)=0.
\end{equation}
Thus the relevant part of the long exact sequence reduces to
\begin{equation}
0
\to
H_1(\widehat G,A;\mathbb F_2)
\xrightarrow{\partial_{\rm rel}}
H_0(A;\mathbb F_2)
\xrightarrow{i_*}
H_0(\widehat G;\mathbb F_2)
\to
0 .
\label{eq:les-pair-app}
\end{equation}

We now compute the kernel of $i_*$. Since $A$ has two connected
components,
\begin{equation}
H_0(A;\mathbb F_2)
=
\operatorname{span}_{\mathbb F_2}
\{[\hat r_0],[\hat r_1]\}.
\end{equation}
Here $[\hat r_i]$ denotes the connected component of the point
$\hat r_i$ inside $A$. Since $\widehat G$ is connected,
$H_0(\widehat G;\mathbb F_2)$ has one generator. We denote it by $[\widehat G]$.

The inclusion-induced map $i_*$ sends both points of $A$ to this same
connected component:
\begin{equation}
i_*([\hat r_0])=[\widehat G],
\qquad
i_*([\hat r_1])=[\widehat G].
\end{equation}
Therefore, for $a,b\in\mathbb F_2$,
\begin{equation}
i_*
\left(
a[\hat r_0]\oplus b[\hat r_1]
\right)
=
(a\oplus b)[\widehat G].
\end{equation}
The elements mapped to zero are exactly those with $a=b$. Hence
\begin{equation}
\ker i_*
=
\operatorname{span}_{\mathbb F_2}
\{[\hat r_0]\oplus[\hat r_1]\}
\cong
\mathbb F_2.
\label{eq:kernel-i-star-two-points}
\end{equation}

By exactness of Eq.~\eqref{eq:les-pair-app}, the image of
$\partial_{\rm rel}$ is equal to $\ker i_*$:
\begin{equation}
\operatorname{im}\partial_{\rm rel}
=
\ker i_* .
\end{equation}
The first zero in Eq.~\eqref{eq:les-pair-app} also implies that
$\partial_{\rm rel}$ is injective. Therefore
\begin{equation}
H_1(\widehat G,A;\mathbb F_2)
\cong
\operatorname{im}\partial_{\rm rel}
=
\ker i_* .
\end{equation}
Using Eq.~\eqref{eq:kernel-i-star-two-points}, we obtain
\begin{equation}
H_1(\widehat G,A;\mathbb F_2)
\cong
\mathbb F_2 .
\label{eq:relative-h1-two-holes-app}
\end{equation}

The computation above shows that there is exactly one nonzero relative
$1$-class. The next proposition records the stabilizer interpretation of this
class and proves that it is represented precisely by legal bridges from
$R_0$ to $R_1$.

\begin{proposition}[Unique bridge class for fixed clean rough holes]
\label{prop:fixed-holes-one-bridge-class}
Fix two disjoint clean actual rough holes $R_0$ and $R_1$. Let
$A=\{\hat r_0,\hat r_1\}$ be the two marked vertices in the quotient
complex $\widehat G$.

Let $c$ and $c'$ be legal $Z$-chains for this punctured code. Then their
quotient images $\hat c$ and $\hat c'$ are relative $1$-cycles for the pair
$(\widehat G,A)$.

Moreover, the following two conditions are equivalent:
\begin{equation}
[\hat c]=[\hat c']
\quad
\text{in}
\quad
H_1(\widehat G,A;\mathbb F_2),
\end{equation}
and
\begin{equation}
\hat c\oplus \hat c'
\in
\operatorname{im}\widehat B .
\end{equation}
Here $\operatorname{im}\widehat B$ is the image, in the quotient complex
$\widehat G$, of the plaquette-boundary span of the punctured code. The
quotient map only identifies the cells inside each actual rough hole to the
marked vertex $\hat r_i$ and is identical on the active code region. Therefore,
away from the rough holes, $\widehat B$ is the same plaquette-boundary map as
the original $Z$-stabilizer boundary map. The auxiliary outer face is not an
additional physical stabilizer; its boundary is already the mod-$2$ sum of the
physical plaquette boundaries. Thus membership in $\operatorname{im}\widehat B$
is exactly equivalence by products of active $Z$-type plaquette stabilizers,
with the allowed rough-boundary endpoint identifications already taken into
account.

In particular, if $c$ and $c'$ are both bridges from $R_0$ to $R_1$, then
\begin{equation}
[\hat c]=[\hat c']\neq 0.
\end{equation}
Hence all bridge operators act as the same nontrivial logical operator
$\bar Z$.
\end{proposition}
\textbf{Proof.}
The equivalence statement follows directly from the definition of relative homology.
By definition,
\begin{equation}
H_1(\widehat G,A;\mathbb F_2)
=
Z_1(\widehat G,A;\mathbb F_2)
/
B_1(\widehat G,A;\mathbb F_2).
\end{equation}
Since $C_1(A;\mathbb F_2)=0$, a relative $1$-chain can be represented by an
ordinary $1$-chain in $\widehat G$. Under this identification,
\begin{equation}
Z_1(\widehat G,A;\mathbb F_2)
=
\{
\hat x\in C_1(\widehat G;\mathbb F_2):
\widehat D\hat x\in C_0(A;\mathbb F_2)
\}.
\end{equation}
Thus these are precisely the quotient chains whose endpoints are supported
on the marked vertices in $A$. The subgroup
$B_1(\widehat G,A;\mathbb F_2)$ is the group of relative $1$-boundaries.

Thus two quotient chains define the same relative class exactly when their
difference is a relative boundary:
\begin{equation}
[\hat c]=[\hat c']
\quad
\Longleftrightarrow
\quad
\hat c\oplus\hat c'
\in
B_1(\widehat G,A;\mathbb F_2).
\end{equation}
All chain groups are over $\mathbb F_2$, so subtraction is the same as
addition.

Since $A=\{\hat r_0,\hat r_1\}$ consists only of $0$-cells, we have
\begin{equation}
C_1(A;\mathbb F_2)=0,
\qquad
C_2(A;\mathbb F_2)=0.
\end{equation}
Hence
\begin{equation}
C_1(\widehat G,A;\mathbb F_2)
=
C_1(\widehat G;\mathbb F_2),
\qquad
C_2(\widehat G,A;\mathbb F_2)
=
C_2(\widehat G;\mathbb F_2).
\end{equation}
Therefore the relative $1$-boundaries are just the ordinary boundaries of
$2$-chains in $\widehat G$. These are exactly the plaquette boundaries:
\begin{equation}
B_1(\widehat G,A;\mathbb F_2)
=
\operatorname{im}\widehat B .
\end{equation}
Hence
\begin{equation}
[\hat c]=[\hat c']
\quad
\Longleftrightarrow
\quad
\hat c\oplus\hat c'
\in
\operatorname{im}\widehat B .
\end{equation}

This gives the stabilizer statement. If
\begin{equation}
\hat c\oplus\hat c'=\widehat B s ,
\end{equation}
for some $s\in C_2(\widehat G;\mathbb F_2)$, then, in the punctured-code
stabilizer language, $Z(c)$ and $Z(c')$ differ by the product of the
$Z$-type plaquette stabilizers selected by $s$. Hence they act in the same
way on the code space.

Now let $c$ be a bridge from $R_0$ to $R_1$. Then
\begin{equation}
\widehat D\hat c
=
\hat r_0\oplus\hat r_1 .
\end{equation}
So $[\hat c]$ cannot be the zero class. Indeed, if $[\hat c]=0$, then
$\hat c\in\operatorname{im}\widehat B$, and therefore
\begin{equation}
\widehat D\hat c=0,
\end{equation}
because $\widehat D\widehat B=0$. This contradicts
\begin{equation}
\widehat D\hat c
=
\hat r_0\oplus\hat r_1
\neq 0 .
\end{equation}

By Eq.~\eqref{eq:relative-h1-two-holes-app},
\begin{equation}
H_1(\widehat G,A;\mathbb F_2)
\cong
\mathbb F_2 .
\end{equation}
Thus there is only one nonzero relative class. Since every bridge is nonzero,
all bridges represent this same class.

Therefore any two bridge operators are equivalent up to $Z$-type plaquette
stabilizers. They act as the same nontrivial logical operator on the code
space. We denote this common logical operator by $\bar Z$.
This proves the proposition.
\qed

\begin{proposition}[Simple dual witness for fixed holes]
\label{prop:fixed-holes-simple-dual-witness}
Fix the same two disjoint clean actual rough holes $R_0$ and $R_1$.
Then there exists a simple closed dual loop $w$ that separates $R_0$ from
$R_1$. For any such $w$, every bridge $c$ from $R_0$ to $R_1$ satisfies
\begin{equation}
\langle w,c\rangle=1 .
\end{equation}

Moreover, if $\hat w$ is any closed dual $1$-cochain on $\widehat G$ such
that
$
\langle \hat w,\hat c_\star\rangle=1
$
for one bridge $c_\star$, then
$
\langle \hat w,\hat c\rangle=1
$
for every bridge $c$ from $R_0$ to $R_1$.
\end{proposition}
\textbf{Proof.}
The existence of a simple separating dual loop follows from the same
planar separation argument used in the proof of Proposition~\ref{prop:global-odd-parity-simple-witness}. Namely,
take a small regular neighborhood of $R_0$ that is disjoint from $R_1$ and
choose its boundary in dual position. This boundary is a simple closed dual
loop separating $R_0$ from $R_1$. Any bridge has one endpoint on each side,
and therefore crosses this loop an odd number of times. Hence
\begin{equation}
\langle w,c\rangle=1 .
\end{equation}

Now let $\hat w$ be any closed dual $1$-cochain. The closed condition is
\begin{equation}
\widehat B^{\top}\hat w=0 .
\end{equation}
If two legal chains differ by a plaquette boundary,
\begin{equation}
\hat c'=\hat c\oplus \widehat B s,
\end{equation}
then
\begin{equation}
\langle \hat w,\hat c'\rangle
=
\langle \hat w,\hat c\rangle
\oplus
\langle \widehat B^{\top}\hat w,s\rangle
=
\langle \hat w,\hat c\rangle .
\end{equation}
Thus the pairing with $\hat w$ depends only on the relative homology class
of the bridge.

By Proposition~\ref{prop:fixed-holes-one-bridge-class}, all bridges from $R_0$ to $R_1$ represent the same
nonzero relative class. Therefore, if $\hat w$ pairs oddly with one quotient bridge $\hat c_\star$,
then it pairs oddly with every quotient bridge $\hat c$.
\qed

We can now prove the clean puncture realization statement used in the main
text.

\textbf{Proof of Proposition~\ref{prop:clean_puncture_realization} in the main text.}

Assume first that the two-stage construction is cleanly openable, and let
$R_0$ and $R_1$ be the resulting clean
rough holes. Let $\widehat G$ be the quotient complex
obtained by collapsing $R_0$ and $R_1$ to two marked vertices $\hat r_0$ and
$\hat r_1$. Set
\begin{equation}
A=\{\hat r_0,\hat r_1\}.
\end{equation}

By construction and by the clean condition, each prescribed signal chain $c_\alpha$ has one
endpoint on the rough boundary of $R_0$ and one endpoint on the rough
boundary of $R_1$. Its interior stays in the active code region. Therefore,
after the quotient, $\hat c_\alpha$ has ordinary boundary
\begin{equation}
\widehat D\hat c_\alpha
=
\hat r_0\oplus \hat r_1
\qquad
\text{for all }\alpha .
\end{equation}
Since this boundary lies in $C_0(A;\mathbb F_2)$, each $\hat c_\alpha$ is a
relative $1$-cycle for the pair $(\widehat G,A)$.

Choose any fixed reference primal path $c_\star$ with one endpoint on the
rough boundary of $R_0$ and the other endpoint on the rough boundary of
$R_1$, whose non-terminal part remains in the active code region.
Then $c_\star$ is a legal bridge from $R_0$ to $R_1$. By
Proposition~\ref{prop:fixed-holes-one-bridge-class}, all bridges from $R_0$ to $R_1$ represent the same
nonzero class in
\begin{equation}
H_1(\widehat G,A;\mathbb F_2).
\end{equation}
Hence
\begin{equation}
[\hat c_\alpha]
=
[\hat c_\star]
\neq 0
\qquad
\text{for all }\alpha .
\end{equation}
Therefore
\begin{equation}
\hat c_\alpha\oplus \hat c_\star
\in
\operatorname{im}\widehat B
\qquad
\text{for all }\alpha .
\end{equation}
Equivalently, for each $\alpha$, there exists a $2$-chain $s_\alpha$ such
that
\begin{equation}
\hat c_\alpha\oplus \hat c_\star
=
\widehat B s_\alpha .
\end{equation}

In the punctured-code stabilizer language, this is exactly the stabilizer
equivalence. The operators $Z(c_\alpha)$ and $Z(c_\star)$ are equivalent up
to the $Z$-type plaquette stabilizers selected by $s_\alpha$. Therefore
$Z(c_\alpha)$ and $Z(c_\star)$ have the same action on the code space:
\begin{equation}
P_{\mathcal C}Z(c_\alpha)P_{\mathcal C}
=
P_{\mathcal C}Z(c_\star)P_{\mathcal C}
\qquad
\text{for all }\alpha .
\end{equation}

Moreover, $[\hat c_\star]\neq 0$. Thus $\hat c_\star$ is not a quotient
plaquette boundary. Hence, in the punctured-code stabilizer language,
$Z(c_\star)$ is not equivalent to a product of $Z$-type plaquette
stabilizers. Since $c_\star$
is a legal bridge, it commutes with all active star checks and preserves the
code space. Therefore it is a nontrivial logical $Z$ operator. We denote its
logical action by $\bar Z$. Thus
\begin{equation}
P_{\mathcal C}Z(c_\alpha)P_{\mathcal C}
=
P_{\mathcal C}Z(c_\star)P_{\mathcal C}
=
P_{\mathcal C}\bar ZP_{\mathcal C}
\qquad
\text{for all }\alpha .
\end{equation}
Thus all prescribed signal operators realize the same nontrivial logical
operator $\bar Z$. This proves the forward direction.

Conversely, suppose that two clean actual rough holes $R_0$ and $R_1$ are
already fixed in the two-hole planar setting considered here. Assume that
each prescribed signal path $c_\alpha$ has one endpoint on the rough boundary
of $R_0$ and the other endpoint on the rough boundary of $R_1$, while its
non-terminal part remains in the active code region. Then each $c_\alpha$ is
a legal bridge from $R_0$ to $R_1$.

By Proposition~\ref{prop:fixed-holes-simple-dual-witness}, we may choose a simple closed dual loop $w$ separating
$R_0$ from $R_1$. Every bridge from $R_0$ to $R_1$ crosses this loop with odd
parity. Hence
\begin{equation}
\langle w,c_\alpha\rangle=1
\qquad
\text{for all }\alpha .
\end{equation}
Thus the global odd-witness condition is necessary for any fixed clean
two-hole realization.

This proves the proposition.
\qed

\section{Multiplicity and nonlocal storage in the Ramsey protocol}
\label{app:logical-ramsey-protocol}

This appendix explains two points used in the Ramsey protocol.
We first show how multiplicity realizes integer and rational weights. We then show why
the accumulated Ramsey phase cannot be accessed by measurements supported on
any correctable region.

We use the fixed-hole relative-homology result proved in
Appendix~\ref{app:relative-homology-fixed-rough-holes}. After two clean rough
holes $R_0$ and $R_1$ have been fixed, every legal signal bridge from $R_0$ to
$R_1$ represents the same nonzero relative homology class and therefore acts
as the same nontrivial logical $\bar Z$ operator on the code space. Thus, for every
signal path used below,
\begin{equation}
P_{\mathcal C}Z(c)P_{\mathcal C}
=
P_{\mathcal C}\bar ZP_{\mathcal C},
\label{eq:app-ramsey-common-logical-assumption}
\end{equation}
where $\mathcal C$ is the code space and $\bar Z$ denotes the common
logical action.

General normalized real weights are realized in the main text by signed activation profiles. Here we record an equivalent multiplicity construction for integer and rational weights. Suppose that, for each unknown coupling $\lambda_\alpha$,
there are $n_\alpha$ available copies with signal paths
\begin{equation}
c_{\alpha,r},
\qquad
r=1,\ldots,n_\alpha .
\end{equation}
The signal paths $c_{\alpha,r}$ may be geometrically different, but each must be a
legal bridge from $R_0$ to $R_1$ in the common nonzero relative homology class.
Therefore
\begin{equation}
P_{\mathcal C}Z(c_{\alpha,r})P_{\mathcal C}
=
P_{\mathcal C}\bar ZP_{\mathcal C}
\qquad
\text{for all }\alpha\in\mathcal I
\text{ and } r=1,\ldots,n_\alpha .
\label{eq:app-ramsey-multiplicity-same-Z}
\end{equation}
With a possible sign reversal implemented by echo control, the multiplicity
Hamiltonian is
\begin{equation}
H_{\rm mult}
=
\frac{1}{2}
\sum_{\alpha\in\mathcal I}
\sigma_\alpha
\sum_{r=1}^{n_\alpha}
\lambda_\alpha Z(c_{\alpha,r}),
\qquad
\sigma_\alpha\in\{\pm1\}.
\label{eq:app-ramsey-multiplicity-H}
\end{equation}
Using Eq.~\eqref{eq:app-ramsey-multiplicity-same-Z}, its restriction to the
code space is
\begin{equation}
P_{\mathcal C}H_{\rm mult}P_{\mathcal C}
=
\frac{1}{2}
\left(
\sum_{\alpha\in\mathcal I}
\sigma_\alpha n_\alpha\lambda_\alpha
\right)
P_{\mathcal C}\bar ZP_{\mathcal C}.
\label{eq:app-ramsey-multiplicity-logical-H}
\end{equation}
Thus multiple physical copies of the same unknown coupling add coherently at
the logical level. In particular, multiplicity realizes integer coefficients
directly.

For rational weights, write
\begin{equation}
s_\alpha
=
\sigma_\alpha\frac{p_\alpha}{q_\alpha},
\qquad
p_\alpha\in\mathbb Z_{\ge 0},
\qquad
q_\alpha\in\mathbb Z_{\ge 1}.
\end{equation}
Choose
\begin{equation}
Q_{\rm lcm}
:=
\operatorname{lcm}\{q_\alpha:\alpha\in\mathcal I\},
\qquad
n_\alpha
:=
Q_{\rm lcm}\frac{p_\alpha}{q_\alpha}
=
Q_{\rm lcm}|s_\alpha|
\in
\mathbb Z_{\ge0}.
\label{eq:app-ramsey-multiplicity-numbers}
\end{equation}
Then
\begin{equation}
\sum_{\alpha\in\mathcal I}
\sigma_\alpha n_\alpha\lambda_\alpha
=
Q_{\rm lcm}
\sum_{\alpha\in\mathcal I}
s_\alpha\lambda_\alpha
=
Q_{\rm lcm}q .
\label{eq:app-ramsey-multiplicity-rescaled-q}
\end{equation}
The factor $Q_{\rm lcm}$ is known, so the final estimate is rescaled by the
same known factor. Hence multiplicity realizes integer weights exactly and
rational weights up to a known normalization.

Finally, we spell out the nonlocal storage property used in the main text.
The Ramsey schedule stores the target parameter in the logical phase
generated by
\begin{equation}
U_{\rm log}
=
\exp\!\left(
-\frac{i t}{2}q\,\bar Z
\right).
\label{eq:app-ramsey-target-U-compact}
\end{equation}
Let $\mathcal E$ be a correctable region of the punctured code. Then every
operator supported on $\mathcal E$ acts trivially on the logical qubit.
Equivalently, for any operator $O_{\mathcal E}$ supported on $\mathcal E$,
\begin{equation}
P_{\mathcal C}O_{\mathcal E}P_{\mathcal C}
=
\gamma(O_{\mathcal E})P_{\mathcal C}
\label{eq:app-ramsey-correctable-region}
\end{equation}
for some scalar $\gamma(O_{\mathcal E})$. Therefore all measurement statistics
obtainable inside $\mathcal E$ are independent of the encoded logical state.
In particular, they are independent of the Ramsey phase in
Eq.~\eqref{eq:app-ramsey-target-U-compact} and reveal no information about the
target parameter $q$. The accumulated phase is therefore carried by the logical
degree of freedom rather than by any correctable region, while the Ramsey
protocol extracts it through a nonlocal logical observable, such as the final
logical $\bar X$ measurement.

\section{Rough-boundary contacts and a diagonal staircase realization of the one-qubit-overlap backbone}
\label{app:local-overlap-constraints}

\subsection*{One-qubit overlaps at a rough boundary}

The requirement that neighboring paths share a terminal data qubit
imposes a stronger geometric condition than merely requiring their
supports to overlap on one edge. To see this, write a three-edge path
as
\[
c=(e_1,e_2,e_3),
\]
where \(e_1\) and \(e_3\) are the terminal edges and \(e_2\) is the
middle edge. The endpoints of \(c\) are the primal vertices at the
outer ends of \(e_1\) and \(e_3\). A legal path between the two rough
holes may meet a rough boundary only at these endpoint vertices; its
middle edge and its non-terminal vertices must remain in the active
code region.

Suppose a data-qubit edge is shared by two paths \(c_a\) and \(c_b\)
at a rough boundary. If this edge is terminal for \(c_a\) but is the
middle edge of \(c_b\), then the boundary vertex that serves as an
endpoint of \(c_a\) is an interior vertex of \(c_b\). Placing that
vertex on the rough boundary would therefore cause \(c_b\) to meet
the boundary through the interior of the path, rather than at one of
its endpoints. The path \(c_b\) would no longer be a legal path
between \(R_0\) and \(R_1\).

Consequently, whenever several paths share a data qubit at a rough
boundary, that data qubit must be terminal for every path involved.
Moreover, the same incident primal vertex on the rough boundary must
serve as an endpoint of all those paths. This boundary-contact rule explains why the shared
contacts in the backbone cannot be arranged as arbitrary overlap
junctions.

Figure~\ref{fig:FIG7} contrasts the two forbidden cases with the allowed configuration.

\begin{figure*}[htbp]
    \centering
    \includegraphics[width=0.68\textwidth]{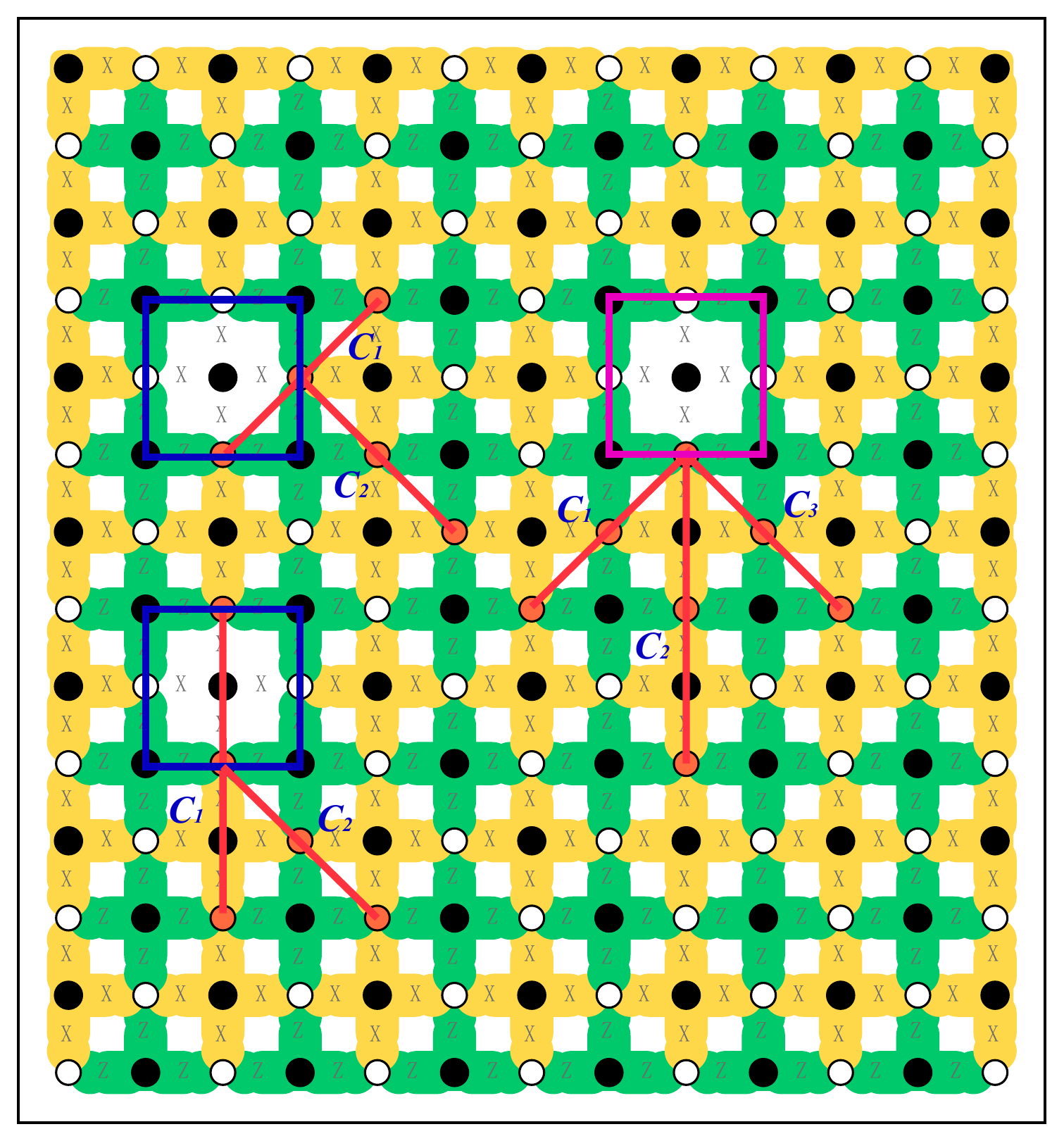}
    \caption{
    \textbf{One-qubit overlaps at a rough boundary.}
    The background notation follows Fig.~\ref{fig:three_qubit_overlap}.
    Each red path labeled $C_i$ represents a placed support $c_i$,
    namely a simple three-edge primal path, with physical signal operator $Z(c_i)$.
    The blue and pink outlines mark local portions of the rough boundaries.
    In each of the two configurations on the left, the shared data qubit is terminal for one path but is the middle data qubit of another.
    The endpoint vertex of the first path would therefore be a non-terminal vertex of the second path.
    Placing this vertex on the rough boundary would make the second path meet the boundary through its interior, so that path would not be legal.
    In the configuration on the right, the shared data qubit is terminal for all three paths $C_1$, $C_2$, and $C_3$.
    The same incident vertex lies on the rough boundary and is an endpoint of all three paths.
    }
    \label{fig:FIG7}
\end{figure*}

\subsection*{Diagonal staircase realization}

The horizontal backbone in Fig.~\ref{fig:three_qubit_overlap}(a) is not the only lattice
realization of the one-qubit-overlap structure. Figure~\ref{fig:FIG8} gives a diagonal staircase realization. The paths approach
the two rough boundaries through a different local geometry, but each
path remains a legal simple three-edge path, and every shared
rough-boundary contact obeys the boundary-contact rule described
above.

For each neighboring pair in the staircase, the product of the two
signal operators is again a product of active \(Z\)-type plaquette
stabilizers. Equivalently, their symmetric difference belongs to
\(\mathcal S_Z^{\mathrm{act}}\). Therefore, once one path in the
staircase is identified with the logical \(\bar Z\), every path
in the staircase has the same logical action. The staircase example
shows that logical compatibility does not depend on the backbone
being horizontally aligned; it depends on the stabilizer-difference
condition and on the legality of the rough-boundary contacts.

The choice of embedding nevertheless affects which additional paths
can be attached. In the horizontal realization of Fig.~\ref{fig:three_qubit_overlap}, the local
geometry near an endpoint leaves enough space for the local
completion shown in Fig.~\ref{fig:three_qubit_overlap}(b). In the diagonal staircase realization,
the backbone itself already occupies the corresponding endpoint
neighborhoods. The same completion therefore cannot be added while
keeping both the staircase paths and the rough boundaries fixed. 
One would have to rearrange part of the backbone or modify the local rough-boundary geometry.
Thus the abstract overlap structure may
admit several logically equivalent lattice realizations, while the
set of permitted local extensions depends on the particular
geometric embedding.

\begin{figure*}[htbp]
    \centering
    \includegraphics[width=0.65\textwidth]{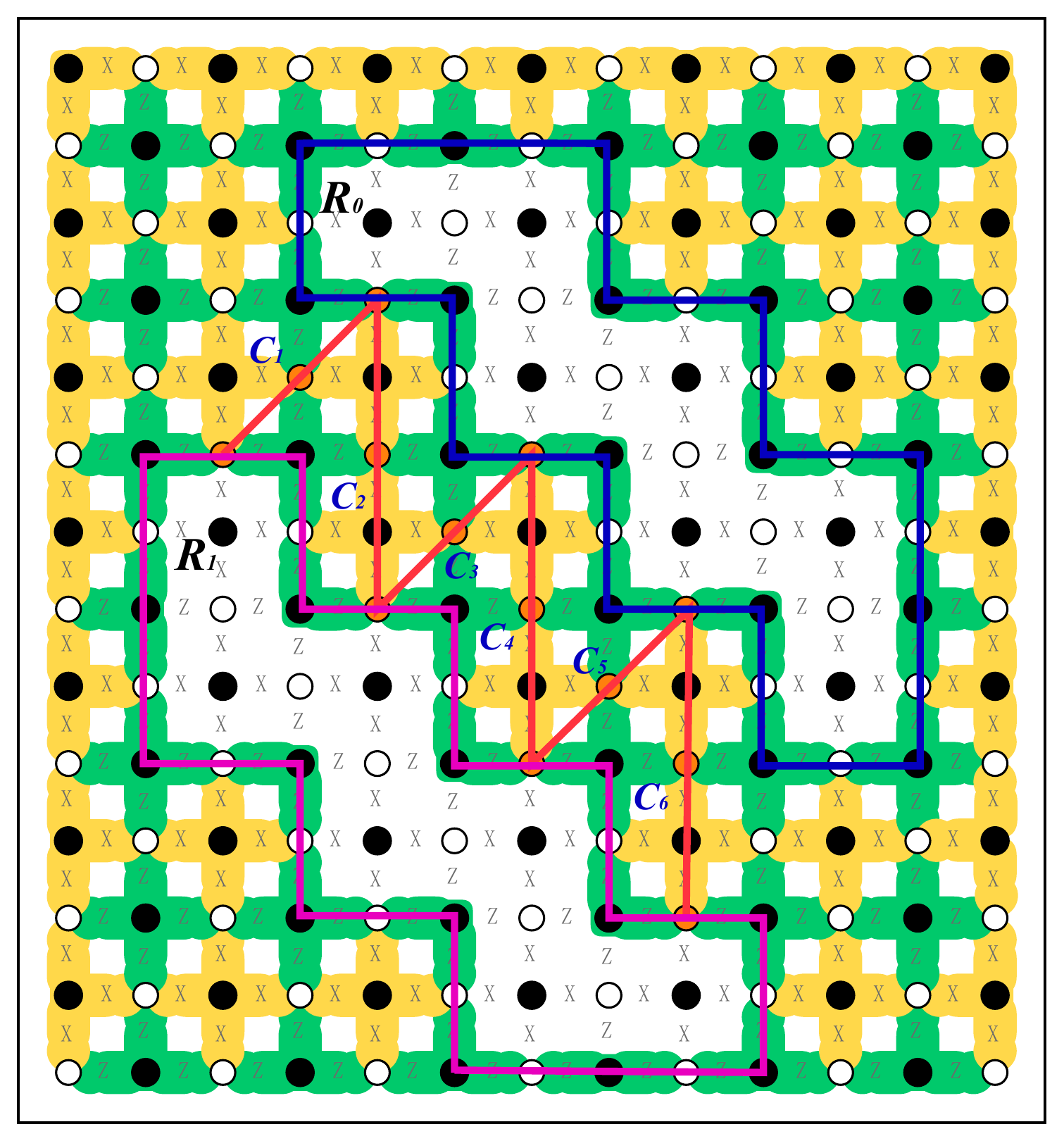}
    \caption{
    \textbf{Diagonal staircase realization of the
    one-qubit-overlap backbone.}
    The background notation follows Fig.~\ref{fig:three_qubit_overlap}.
    The paths $C_1,\ldots,C_6$ are legal paths between the two rough boundaries and together form a single open staircase backbone.
    At each shared rough-boundary contact, the shared data qubit is terminal
    for every path involved, and the same incident vertex is an endpoint of all those paths.
    }
    \label{fig:FIG8}
\end{figure*}

\FloatBarrier
\twocolumngrid
\bibliography{bib_text} 
\end{document}